\documentclass{article}

\usepackage[final]{neurips_2021}

\usepackage[utf8]{inputenc} 
\usepackage[T1]{fontenc}    
\usepackage[colorlinks=true,citecolor=black,linkcolor=blue]{hyperref}
\usepackage{url}            
\usepackage{booktabs}       
\usepackage{amsfonts}       
\usepackage{nicefrac}       
\usepackage{microtype}      
\usepackage{xcolor}         
\usepackage{subcaption}     
\usepackage{amsmath,amsfonts,amsthm,amssymb}
\usepackage{bbm}
\usepackage{bm}
\usepackage{xcolor}

\newtheorem{example}{Example}

\newtheorem{assumption}{Assumption}

\usepackage{prettyref}
\newcommand{\rref}[2][]{\prettyref{#2}}
\newrefformat{model}{Model\,\ref{#1}}
\newrefformat{listing}{Listing\,\ref{#1}}
\newrefformat{algm}{Algorithm\,\ref{#1}}
\newrefformat{line}{line\,\ref{#1}}
\newrefformat{sec}{Section\,\ref{#1}}
\newrefformat{subsec}{Subsection\,\ref{#1}}
\newrefformat{section}{Section\,\ref{#1}}
\newrefformat{appendix}{Appendix\,\ref{#1}}
\newrefformat{app}{Appendix\,\ref{#1}}
\newrefformat{def}{Definition\,\ref{#1}}
\newrefformat{defn}{Definition\,\ref{#1}}
\newrefformat{thm}{Theorem\,\ref{#1}}
\newrefformat{ax}{\ref{#1}}
\newrefformat{prop}{Proposition\,\ref{#1}}
\newrefformat{lemma}{Lemma\,\ref{#1}}
\newrefformat{cor}{Corollary\,\ref{#1}}
\newrefformat{corollary}{Corollary\,\ref{#1}}
\newrefformat{ex}{Example\,\ref{#1}}
\newrefformat{tab}{Table\,\ref{#1}}
\newrefformat{fig}{Fig.\,\ref{#1}}
\newrefformat{eqn}{Equation~(\ref{#1})}
\newrefformat{problem}{Problem\,\ref{#1}}
\newrefformat{assumption}{Assumption\,\ref{#1}}

\newcommand{\cA}{\mathcal{A}}

\newcommand{\cC}{\mathcal{C}}
\newcommand{\cD}{\mathcal{D}}
\newcommand{\cE}{\mathcal{E}}

\newcommand{\cG}{\mathcal{G}}

\newcommand{\cI}{\mathcal{I}}

\newcommand{\cT}{\mathcal{T}}

\newcommand{\bbE}{\mathbb{E}}

\DeclareMathOperator{\pa}{pa}
\DeclareMathOperator{\an}{an}

\newcommand{\rmP}{{\mathrm{P}}}
\newcommand{\rmQ}{{\mathrm{Q}}}
\usepackage{graphicx}
\graphicspath{{fig/}}
\usepackage[ruled,linesnumbered]{algorithm2e}
\usepackage{multibib}
\newcites{appendix}{References of Appendix}

\newtheorem{lemma}{Lemma}
\newtheorem{theorem}{Theorem}
\newtheorem{proposition}{Proposition}
\newtheorem{definition}{Definition}
\newtheorem{corollary}{Corollary}
\newtheorem{observation}{Observation}

\title{Matching a Desired Causal State \\via Shift Interventions}

\author{
  Jiaqi Zhang \\
  LIDS, EECS, and IDSS, MIT \\
  Cambridge, MA 02139 \\
  \texttt{viczhang@mit.edu} \\
   \And
   Chandler Squires \\
   LIDS, EECS, and IDSS, MIT \\
   Cambridge, MA 02139 \\
   \texttt{csquires@mit.edu} \\
   \AND
   Caroline Uhler \\
   LIDS, EECS, and IDSS, MIT \\
   Cambridge, MA 02139 \\
   \texttt{cuhler@mit.edu} \\
}

\begin{document}

\maketitle
\addtocontents{toc}{\protect\setcounter{tocdepth}{0}}

\begin{abstract}
Transforming a causal system from a given initial state to a desired target state is an important task permeating multiple fields including control theory, biology, and materials science. 
In causal models, such transformations can be achieved by performing a set of interventions.
In this paper, we consider the problem of identifying a shift intervention that matches the desired mean of a system through active learning.
We define the Markov equivalence class that is identifiable from shift interventions and propose two active learning strategies that are guaranteed to exactly match a desired mean.
We then derive a worst-case lower bound for the number of interventions required and show that these strategies are optimal for certain classes of graphs. 
In particular, we show that our strategies may require exponentially fewer interventions than the previously considered approaches, which optimize for structure learning in the underlying causal graph. %
In line with our theoretical results, we also demonstrate experimentally that our proposed active learning strategies require fewer interventions compared to several baselines.
\end{abstract}

\section{Introduction}

Consider an experimental biologist attempting to turn cells from one type into another, e.g., from fibroblasts to neurons \citep{vierbuchen2010direct}, by altering gene expression.
This is known as cellular reprogramming and has shown great promise in recent years for regenerative medicine \citep{rackham2016predictive}.
A common approach is to model gene expression of a cell, which is governed by an underlying gene regulatory network, using a \textit{structural causal model} \citep{friedman2000using,badsha2019learning}.
Through a set of \textit{interventions}, such as gene knockouts or over-expression of transcription factors \citep{dominguez2016beyond}, a biologist can infer the structure of the underlying regulatory network.
After inferring enough about this structure, a biologist can identify the intervention needed to successfully reprogram a cell.
More generally, transforming a causal system from an initial state to a desired state through interventions is an important task pervading multiple applications.
Other examples include closed-loop control \citep{touchette2004information} and pathway design of microstructures \citep{wodo2015automated}.

With little prior knowledge of the underlying causal model, this task is intrinsically difficult. 
Previous works have addressed the problem of intervention design to achieve full identifiability of the causal model \citep{hauser2014two, greenewald2019sample, squires2020active}. 
However, since interventional experiments tend to be expensive in practice, one wishes to minimize the number of trials and learn \textit{just enough} information about the causal model to be able to identify the intervention that will transform it into the desired state.
Furthermore, in many realistic cases, the set of interventions which can be performed is constrained.
For instance, in CRISPR experiments, only a limited number of genes can be knocked out to keep the cell alive; or in robotics, a robot can only manipulate a certain number of arms at once.

\textbf{Contributions.} 
We take the first step towards the task of \textit{causal matching} (formalized in \rref{sec:problem-step}), where an experimenter can perform a series of interventions in order to identify a \textit{matching intervention} which transforms the system to a desired state.
In particular, we consider the case where the goal is to match the \textit{mean} of a distribution.
We focus on a subclass of interventions called \textit{shift} interventions, which can for example be used to model gene over-expression experiments \citep{triantafillou2017predicting}. 
These interventions directly increase or decrease the values of their perturbation targets, with their effect being propagated to variables which are downstream (in the underlying causal graph) of these targets.
We show that there always exists a unique shift intervention (which may have multiple perturbation targets) that exactly transforms the mean of the variables into the desired mean (Lemma~\ref{lm:1}). We call this shift intervention the \emph{matching intervention}.

To find the matching intervention, in \rref{sec:identifiability} we characterize the \textit{Markov equivalence class} of a causal graph induced by shift interventions, i.e., the edges in the causal graph that are identifiable from shift interventions; in particular, we show that the resulting Markov equivalence classes can be more refined than previous notions of interventional Markov equivalence classes. 
We then propose two \textit{active} learning strategies in \rref{sec:algorithms} based on this characterization, which are guaranteed to identify the matching intervention.
These active strategies proceed in an adaptive manner, where each intervention is chosen based on all the information gathered so far.

In \rref{sec:theory}, we derive a worst-case lower bound on the number of interventions required to identify the matching intervention and show that the proposed strategies are optimal up to a logarithmic factor.
Notably, the proposed strategies may use \textit{exponentially} fewer interventions than previous active strategies for structure learning.
Finally, in \rref{sec:experiments}, we demonstrate also empirically that our proposed strategies outperform previous methods as well as other baselines in various settings. 

\subsection{Related Works}
\textbf{Experimental Design.} Previous work on experimental design in causality has considered two closely related goals: learning the most structural information about the underlying DAG given a fixed budget of interventions \citep{ghassami2018budgeted}, and fully identifying the underlying DAG while minimizing the total number or cost \citep{shanmugam2015learning,kocaoglu2017cost} of interventions.
These works can also be classified according to whether they consider a passive setting, i.e., the interventions are picked at a single point in time \citep{hyttinen2013experiment,shanmugam2015learning,kocaoglu2017cost}, or an active setting, i.e., interventions are decided based on the results of previous interventions \citep{he2008active,agrawal2019abcd,greenewald2019sample,squires2020active}.
The setting addressed in the current work is closest to the active, full-identification setting. The primary difference is that in order to match a desired mean, one does not require full identification; in fact, as we show in this work, we may require significantly less interventions.

\textbf{Causal Bandits.}
Another related setting is the bandit problem in sequential decision making, where an agent aims to maximize the cumulative reward by selecting an arm at each time step. Previous works considered the setting where there are causal relations between regrets and arms \citep{lattimore2016causal, lee2018structural, yabe2018causal}. Using a known causal structure, these works were able to improve the dependence on the total number of arms compared to previous regret lower-bounds \citep{bubeck2012regret,lattimore2016causal}. 
These results were further extended to the case when the causal structure is unknown \textit{a priori} \citep{de2020causal}. 
In all these works the variables are discrete, with arms given by \textit{do}-interventions (i.e., setting variables to a given value), so that there are only a finite number of arms.
In our work, we are concerned with the continuous setting and shift interventions, which corresponds to an infinite (continuous) set of arms.

\textbf{Correlation-based Approaches.}
There are also various correlation-based approaches for this task that do not make use of any causal information. For example, previous works have proposed score-based \citep{cahan2014cellnet}, entropy-based \citep{d2015systematic} and distance-based approaches \citep{rackham2016predictive} for cellular reprogramming. 
However, as shown in bandit settings \citep{lattimore2016causal}, when the system follows a causal structure, this structure can be exploited to learn the optimal intervention more efficiently.
Therefore, we here focus on developing a causal approach.

\section{Problem Setup}\label{sec:problem-step}
We now formally introduce the \textit{causal matching problem} of identifying an intervention to match the desired state in a causal system under a given metric. 
Following \citet{koller2009probabilistic}, a \textit{causal structural model} is given by a directed acyclic graph (DAG) $\cG$ with nodes $[p]=\{1,\dots , p\}$, and a set of random variables $X=\{X_1,...,X_p\}$ whose joint distribution $\rmP$ factorizes according to $\cG$. Denote by $\pa_{\cG}(i)=\{j\in [p]\mid j\to i\}$ the \textit{parents} of node $i$ in $\cG$.
An \textit{intervention} $I\subset [p]$ with multiple \textit{perturbation targets} $i\in I$ either removes all incoming edges to $X_i$ (\textit{hard} intervention) or modifies the conditional probability $\rmP(X_i|X_{\pa_{\cG}(i)})$ (\textit{soft} intervention) for all $i\in I$. This results in an interventional distribution $\rmP^{I}$. Given a desired joint distribution $\rmQ$ over $X$, the goal of causal matching is to find an \textit{optimal matching intervention} $I$ such that $\rmP^{I}$ best matches $\rmQ$ under some metric. In this paper, we address a special case of the causal matching problem, which we call \textit{causal mean matching}, where the distance metric between $\rmP^I$ and $\rmQ$ depends only on their expectations.

We focus on causal mean matching for a class of soft interventions, called \textit{shift interventions} \citep{rothenhausler2015backshift}. Formally, a shift intervention with perturbation targets $I\subset [p]$ and shift values $\{a_i\}_{i\in I}$ modifies the conditional distribution as $\rmP^{I}(X_i=x+a_i|X_{\pa_{\cG}(i)})=\rmP(X_i=x|X_{\pa_{\cG}(i)})$. Here, the shift values $\{a_i\}_{i\in I}$ are assumed to be deterministic.
We aim to find $I\subset [p]$ and $\{a_i\}_{i\in I}\in \mathbb{R}^{|I|}$ such that the \textit{mean} of $\rmP^I$ is closest to that of $\rmQ$, i.e., minimizes $ d(\bbE_{\rmP^{I}}(X), \bbE_{\rmQ}(X))$ for some metric $d$. In fact, as we show in the following lemma, there always exists a unique shift intervention, which we call the \emph{matching intervention}, that achieves \textit{exact mean matching}.\footnote{To lighten notation, we use $I$ to denote both the perturbation targets and the shift values of this intervention.}

\begin{lemma}\label{lm:1}
For any causal structural model and desired mean $\bbE_\rmQ (X)$,
there exists a unique shift intervention $I^*$ such that $\bbE_{\rmP^{I^*}}(X) = \bbE_{\rmQ}(X)$.
\end{lemma}

We assume throughout that the underlying causal DAG $\cG$ is \textit{unknown}.
But we assume \textit{causal sufficiency} \citep{spirtes2000causation}, which excludes the existence of latent confounders, as well as access to enough observational data to determine the joint distribution $\rmP$ and thus the Markov equivalence class of $\cG$ \citep{andersson1997characterization}.
It is well-known that with enough interventions, the causal DAG $\cG$ becomes fully identifiable \citep{yang2018characterizing}.
Thus one strategy for causal mean matching is to first use interventions to fully identify the structure of $\cG$, and then solve for the matching intervention given full knowledge of the graph. 
However, in general this strategy requires more interventions than needed. 
In fact, the number of interventions required by such a strategy can be \textit{exponentially} larger than the number of interventions required by a strategy that directly attempts to identify the matching intervention, as illustrated in Figure~\ref{fig:shift-vs-structure-learning-gap} and proven in Theorem~\ref{thm:2}.

\begin{figure}[t]
     \centering
     \begin{subfigure}[b]{0.4\textwidth}
         \centering
         \includegraphics[width=0.64\textwidth]{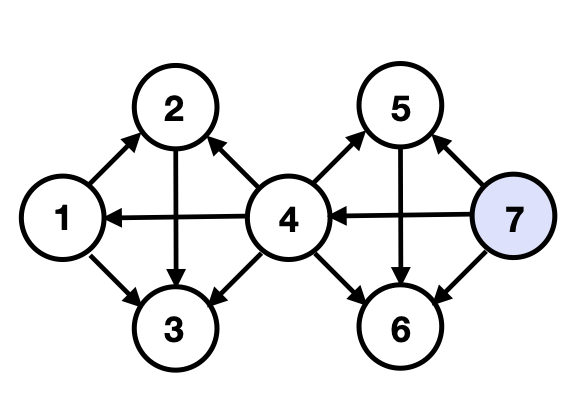}
         \caption{}
     \end{subfigure}
     \begin{subfigure}[b]{0.4\textwidth}
         \centering
         \includegraphics[width=0.64\textwidth]{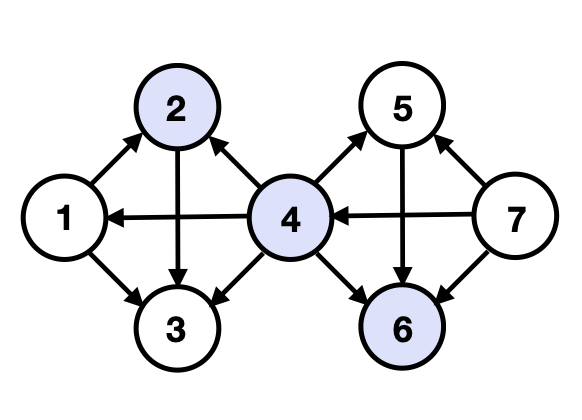}
         \caption{}
     \end{subfigure}
    \caption{
    Completely identifying a DAG can require exponentially more interventions than identifying the matching intervention. 
    Consider a graph constructed by joining $r$ size-$4$ cliques, where the matching intervention has the source node as the only perturbation target, as pictured in \textbf{(a)} with $r=2$ and the source node in purple;
    \textbf{(b)} shows the minimum size set intervention (in purple) that completely identifies the DAG, which grows as $O(r)$ \citep{squires2020active}.
    In Theorem \ref{thm:2}, we show that the matching intervention can be identified in $O(\log r)$ single-node interventions.}
    \label{fig:shift-vs-structure-learning-gap}
\end{figure}

In this work, we consider \textit{active} intervention designs, where a series of interventions are chosen adaptively to learn the matching intervention. This means that the information obtained after performing each intervention is taken into account for future choices of interventions. We here focus on the \textit{noiseless} setting, where for each intervention enough data is obtained to decide the effect of each intervention. Direct implications for the noisy setting are discussed in \rref{appendix:noisy}. To incorporate realistic cases in which the system cannot withstand an intervention with too many target variables, as is the case in CRISPR experiments, where knocking out too many genes at once often kills the cell, we consider the setting where there is a \textit{sparsity} constraint $S$ on the maximum number of perturbation targets in each intervention, i.e., we only allow $I$ where $|I| \leq S$.

\section{Identifiability}\label{sec:identifiability}
In this section, we characterize and provide a graphical representation of the \emph{shift interventional Markov equivalence class} (shift-$\mathcal{I}$-MEC), i.e., the equivalence class of DAGs that is identifiable by shift interventions $\mathcal{I}$.
We also introduce \textit{mean interventional faithfulness}, an assumption that guarantees identifiability of the underlying DAG up to its shift-$\mathcal{I}$-MEC.
Proofs are given in \rref{appendix:identify}.

\subsection{Shift-interventional Markov Equivalence Class}\label{sec:id-1}

For any DAG $\cG$ with nodes $[p]$, a distribution $f$ is \textit{Markov} with respect to $\cG$ if it factorizes according to $f(X)=\prod_{i\in [p]} f(X_i|X_{\pa_{\cG}(i)})$. 
Two DAGs are \textit{Markov equivalent} or in the same \textit{Markov equivalence class} (MEC) if any positive distribution $f$ which is Markov with respect to (w.r.t.) one DAG is also Markov w.r.t. the other DAG. 
With observational data, a DAG is only identifiable up to its MEC \citep{andersson1997characterization}. 
However, the identifiability improves to a smaller class of DAGs with interventions. 
For a set of interventions $\cI$ (not necessarily shift interventions), the pair $(f, \{f^I\}_{I\in \cI})$ is \textit{$\cI$-Markov} w.r.t.~$\cG$ if $f$ is Markov w.r.t.~$\cG$ and $f^I$ factorizes according to
\begin{equation*}
f^I(X) = \prod_{i\notin I}f(X_i|X_{\pa_{\cG}(i)}) \prod_{i\in I} f^{I}(X_i|X_{\pa_{\cG}(i)}),\quad \forall I \in \cI.   
\end{equation*}
Similarly, the \textit{interventional Markov equivalence class} ($\cI$-MEC) of a DAG can be defined, and \citet{yang2018characterizing} provided a structural characterization of the $\cI$-MEC for general interventions $\cI$ (not necessarily shift interventions).

Following, we show that if $\cI$ consists of shift interventions, then the $\cI$-MEC becomes smaller, i.e., identifiability of the causal DAG is improved. 
The proof utilizes Lemma~\ref{lm:2} on the relationship between conditional probabilities. 
For this, denote by $\an_{\cG}(i)$ the ancestors of node $i$, i.e., all nodes $j$ for which there is a directed path from $j$ to $i$ in $\cG$.
For a subset of nodes $I$, we say that $i\in I$ is a \textit{source} w.r.t. $I$ if $\an_{\cG}(i)\cap I=\varnothing$. 
A subset $I'\subset I$ is a \textit{source} w.r.t.~$I$ if every node in $I'$ is a source w.r.t.~$I$.

\begin{lemma}\label{lm:2}
For any distribution $f$ that factorizes according to $\cG$, the interventional distribution $f^I$ for a shift intervention $I\subset [p]$ with shift values $\{a_i\}_{i\in I}$ satisfies 
\[
\bbE_{f^I}(X_i) = \bbE_{f}(X_i) + a_i, 
\]
for any source $i\in I$. Furthermore, if $i\in I$ is not a source w.r.t. $I$, then there exists a positive distribution $f$ such that $\bbE_{f^I}(X_i) \neq \bbE_{f}(X_i) + a_i$.
\end{lemma}
Hence, we can define the \textit{shift-$\cI$-Markov property} and \textit{shift-interventional Markov equivalence class} (shift-$\cI$-MEC) as follows.
\begin{definition}\label{def:1}
For a set of shift interventions $\cI$, the pair $(f, \{f^I\}_{I\in \cI})$ is \textnormal{shift-$\cI$-Markov} w.r.t. $\cG$ if $(f, \{f^I\}_{I\in \cI})$ is $\cI$-Markov w.r.t. $\cG$ and 
\[
\bbE_{f^I}(X_i) = \bbE_{f}(X_i) + a_i, \quad \forall~i\in I\in\cI~s.t.~\an_{\cG}(i)\cap I = \varnothing.
\]
Two DAGs are in the same \textnormal{shift-$\cI$-MEC} if any positive distribution that is shift-$\cI$-Markov w.r.t.~one DAG is shift-$\cI$-Markov also w.r.t. the other DAG.
\end{definition}

The following graphical characterizations are known: Two DAGs are in the same MEC if and only if they share the same skeleton (adjacencies) and v-structures (induced subgraphs $i\rightarrow j \leftarrow k$), see \citet{verma1991equivalence}. 
For general interventions $\cI$, two DAGs are in the same $\cI$-MEC, if they are in the same MEC and they have the same directed edges $\{i\rightarrow j|i\in I, j\notin I, I\in\cI, i-j\}$, where $i-j$ means that either $i \to j$ or $j \to i$ \citep{hauser2012characterization, yang2018characterizing}. 
In the following theorem, we provide a graphical criterion for two DAGs to be in the same shift-$\cI$-MEC.
\begin{theorem}\label{thm:1}
Let $\cI$ be a set of shift interventions. Then two DAGs $\cG_1$ and $\cG_2$ belong to the same shift-$\cI$-MEC if and only if they have the same skeleton, v-structures, directed edges $\{i\rightarrow j|i\in I, j\notin I, I\in\cI, i-j\}$, as well as source nodes of $I$ for every $I \in \cI$.
\end{theorem}
In other words, two DAGs are in the same shift-$\cI$-MEC if and only if they are in the same $\cI$-MEC and they have the same source perturbation targets. 
Figure \ref{fig2} shows an example; in particular,  
to represent an MEC, we use the \textit{essential graph} (EG), which has the same skeleton as any DAG in this class and directed edges $i\rightarrow j$ if $i\rightarrow j$ for every DAG in this class. 
The essential graphs corresponding to the MEC, $\cI$-MEC and shift-$\cI$-MEC of a DAG $\cG$ are referred to as EG, $\cI$-EG and shift-$\cI$-EG of $\cG$, respectively.
They can be obtained from the aforementioned graphical criteria (along with a set of logical rules known as the Meek rules \citep{meek2013causal}; see details in \rref{appendix:pre}).
Figure \ref{fig2} shows an example of EG, $\cI$-EG and shift-$\cI$-EG of a four-node DAG.

\begin{figure}[t]
     \centering
     \begin{subfigure}[b]{0.2\textwidth}
         \centering
         \includegraphics[width=0.65\textwidth]{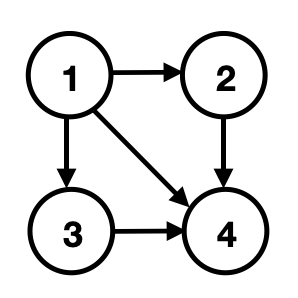}
         \caption{}
     \end{subfigure}
     \begin{subfigure}[b]{0.2\textwidth}
         \centering
         \includegraphics[width=0.65\textwidth]{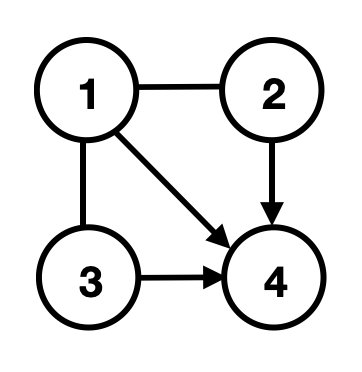}
         \caption{}
     \end{subfigure}
    \begin{subfigure}[b]{0.2\textwidth}
         \centering
         \includegraphics[width=0.65\textwidth]{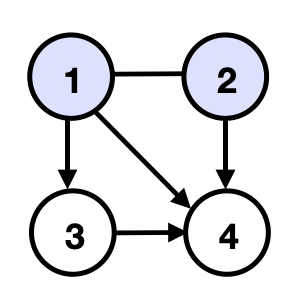}
         \caption{}\label{fig2b}
    \end{subfigure}
    \begin{subfigure}[b]{0.2\textwidth}
         \centering
         \includegraphics[width=0.65\textwidth]{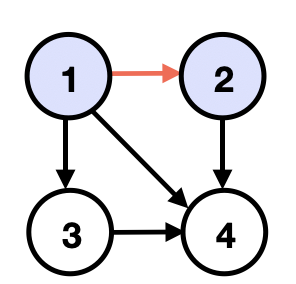}
         \caption{}
    \end{subfigure}
        \caption{Three types of essential graphs. \textbf{(a).} DAG $\cG$; \textbf{(b).} EG of $\cG$; \textbf{(c).} $\cI$-EG of $\cG$ where $\cI$ contains one intervention with perturbation targets $X_1,X_2$ (purple); \textbf{(d).} shift-$\cI$-EG of $\cG$, which can identify an additional edge compared to $\cI$-EG (red).
        }
        \label{fig2}
\end{figure}

\subsection{Mean Interventional Faithfulness}
For the causal mean matching problem, the underlying $\cG$ can be identified from shift interventions $\cI$ up to its shift-$\cI$-MEC.
However, we may not need to identify the entire DAG to find the matching intervention $I^*$. 
Lemma \ref{lm:1} implies that if $i$ is neither in nor downstream of $I^*$, then the mean of $X_i$ already matches the desired state, i.e., $\bbE_{\rmP}(X_i)= \bbE_{\rmQ}(X_i)$;
this suggest that these variables may be negligible when learning $I^*$. 
Unfortunately, the reverse is not true; one may design ``degenerate" settings where a variable is in (or downstream of) $I^*$, but its marginal mean is also unchanged:
\begin{example}
Let $X_3 = X_1 + 2 X_2$, with $\bbE_\rmP (X_1) = 1$ and $\bbE_\rmP (X_2) = 1$, so that $\bbE_\rmP (X_3) = 3$.
Suppose $I^*$ is a shift intervention with perturbation targets $\{ X_1, X_2, X_3 \}$, with $a_1 = 1$, $a_2 = -1$, and $a_3 = 1$.
Then $\bbE_{\rmP^I} (X_3) = 3$, i.e., the marginal mean of $X_3$ is unchanged under the intervention.
\end{example}

Such degenerate cases arise when the shift on a node $X_j$ (deemed $0$ if not shifted) exactly cancels out the contributions of shifts on its ancestors.
Formally, the following assumption rules out these cases.

\begin{assumption}[Mean Interventional Faithfulness]\label{assumption:1}
If $i\in [p]$ satisfies $\bbE_{\rmP}(X_i)= \bbE_{\rmQ}(X_i)$, then $i$ is neither a nor downstream of any perturbation target, i.e., $i\notin I^*, \an_{\cG}(i)\cap I^*=\varnothing$.
\end{assumption}
This is a particularly weak form of faithfulness, which is implied by interventional faithfulness assumptions in prior work \citep{yang2018characterizing,squires2020permutation,jaber2020causal}. 

Let $T$ be the collection of nodes $i\in [p]$ for which $\bbE_{\rmP}(X_i)\neq \bbE_{\rmQ}(X_i)$.
The following lemma shows that under the mean interventional faithfulness assumption we can focus on the subgraph $\cG_{T}$ induced by $T$, since $I^*\subset T$ and interventions on $X_{T}$ do not affect $X_{[p]\setminus T}$.
\begin{lemma}\label{lm:3}
If \rref{assumption:1} holds, then any edge $i-j$ with $j\in T$ and $i \notin T$ has orientation $j\leftarrow i$. Conversely, if \rref{assumption:1} does not hold, then there exists some $i - j$, $j \in T$, $i \not\in T$ such that $j \to i$.
\end{lemma}

\section{Algorithms}\label{sec:algorithms}
Having shown that shift interventions allow the identification of source perturbation targets and that the mean interventional faithfulness assumption allows reducing the problem to an induced subgraph, we now propose two algorithms to learn the matching intervention. 
The algorithms actively pick a shift intervention $I_t$ at time $t$ based on the current shift-interventional essential graph (shift-$\cI_{t}$-EG). 
Without loss of generality and for ease of discussion, we assume that the mean interventional faithfulness assumption holds and we therefore only need to consider $\cG_{T}$.
In \rref{appendix:alg}, we show that the faithfulness violations can be identified and thus \rref{assumption:1} is not necessary for identifying the matching intervention, but additional interventions may be required.

\begin{figure}[t]
     \centering
     \begin{subfigure}[b]{0.32\textwidth}
         \centering
         \includegraphics[width=0.52\textwidth]{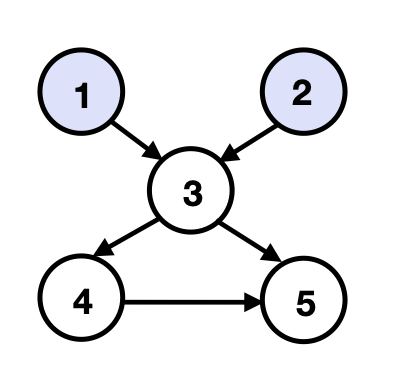}
         \caption{}
     \end{subfigure}
     \begin{subfigure}[b]{0.32\textwidth}
         \centering
         \includegraphics[width=0.52\textwidth]{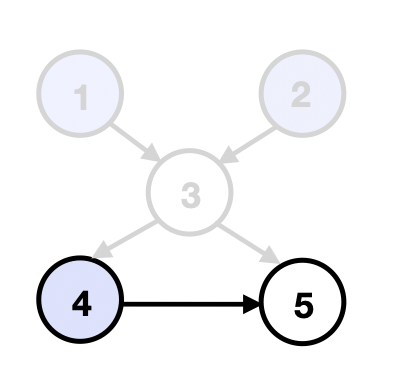}
         \caption{}
     \end{subfigure}
     \begin{subfigure}[b]{0.32\textwidth}
         \centering
         \includegraphics[width=0.52\textwidth]{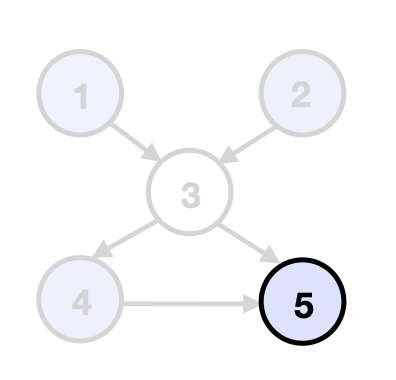}
         \caption{}
     \end{subfigure}
        \caption{Learning $I^*$ when the structure is known. Undimmed parts represent the current subgraph with source nodes (in purple). $I^*=\{1,2,4,5\}$ is solved in three steps. Shift values are omitted. \textbf{(a).} $\cG_{T}$ and $U_{T}$; \textbf{(b).} $\cG_{T_1}$ and $U_{T_1}$; \textbf{(c).} $\cG_{T_2}$ and $U_{T_2}$.
        }
        \label{fig:upstream-search-example}
\end{figure}

\textbf{Warm-up: Upstream Search.} Consider solving for the matching intervention $I^*$ when the structure of $\cG_{T}$ is known, i.e., the current shift-$\cI_{t}$-EG is fully directed. 
Let $U_{T}=\{i|i\in T, \an_{\cG_{T}}(i)\cap T = \varnothing \}$ be the non-empty set of source nodes in $T$. 
We make the following key observation.
\begin{observation}\label{obs:1}
$U_{T}\subset I^*$, and the shift values are $a_{i}= \bbE_{\rmQ}(X_i) - \bbE_{\rmP}(X_i)$~for each $i \in U_T$.
\end{observation}
This follows since shifting other variables in $T$ cannot change the mean of nodes in $U_T$.
Further, the shifted means of variables in $U_T$ match the desired mean (Lemma \ref{lm:2}). 
Given the resulting intervention $U_T$, we obtain a new distribution $\rmP^{U_T}$.
Assuming mean interventional faithfulness on this distribution, we may now remove those variables whose means in $\rmP^{U_T}$ already match $\rmQ$. 
We then repeat this process on the new set of unmatched source nodes, $T_1$, to compute the corresponding shift intervention $U_{T_1}$.
Repeating until we have matched the desired mean for all variables yields $I^*$.
We illustrate this procedure in Figure~\ref{fig:upstream-search-example}.

The idea of upstream search extends to shift-$\cI_{t}$-EG with partially directed or undirected $\cG_T$. In this case, if a node or nodes of $\cG_T$ are identified as source, Observation \ref{obs:1} still holds. Hence, we solve a part of $I^*$ with these source nodes and then intervene on them to reduce the unsolved graph size.

\textbf{Decomposition of Shift Interventional Essential Graphs:} In order to find the source nodes, we decompose the current shift-$\cI_t$-EG into undirected components. \citet{hauser2014two} showed that every interventional essential graph is a chain graph with chordal chain components, where the orientations in one chain component do not affect the orientations in other components.\footnote{The \textit{chain components} of a chain graph are the undirected connected components after removing all its directed edges, and an undirected graph is \textit{chordal} if all cycles of length greater than 3 contain a chord.} By a similar argument, we can obtain an analogous decomposition for shift interventional essential graphs, and show that there is at least one chain component with no incoming edges.
Let us separate out all of the chain components of shift-$\cI_t$-EG with no incoming edges. The following lemma proves that all sources are contained within these components.
\begin{lemma}\label{lm:4}
For any shift-$\cI$-EG of $\cG$, each chain component has exactly one source node w.r.t.~this component. 
This node is a source w.r.t.~$\cG$ if and only if there are no incoming edges to this component.
\end{lemma}

These results hold when replacing $\cG$ with any induced subgraph of it. Thus, we can find the source nodes in $T$ by finding the source nodes in each of its chain components with no incoming edges.

\subsection{Two Approximate Strategies}\label{sec:4.3}
Following the chain graph decomposition, we now focus on how to find the source node of an undirected connected chordal graph $\cC$. 
If there is no sparsity constraint on the number of perturbation targets in each shift intervention, then directly intervening on \textit{all} of the variables in $\cC$ gives the source node, since by Theorem \ref{thm:1}, all DAGs in the shift-$\cI$-MEC share the same source node. 
However, when the maximum number of perturbation targets in an intervention is restricted to $S <|\cC|$, multiple interventions may be necessary to find the source node.

After intervening on $S$ nodes, the remaining unoriented part can be decomposed into connected components.
In the worst case, the source node of $\cC$ is in the \textit{largest} of these connected components.
Therefore we seek the set of nodes, within the sparsity constraint, that minimizes the largest connected component size after being removed. 
This is known as the \textit{MinMaxC} problem \citep{lalou2018critical}, which we show is NP-complete on chordal graphs (\rref{appendix:alg}). 
We propose two approximate strategies to solve this problem, one based on the clique tree representation of chordal graphs and the other based on robust supermodular optimization. 
The overall algorithm with these subroutines is summarized in Algorithm \ref{alg:1}. We outline the subroutines here, and give further details in \rref{appendix:alg}.

\begin{algorithm}[t]
\SetAlgoLined
\KwIn{Joint distribution $\rmP$, desired joint distribution $\rmQ$, sparsity constraint $S$.}
Initialize $I^*=\varnothing$ and $\cI=\{\varnothing\}$\;
\While{$\bbE_{\rmP^{I^*}}(X)\neq \bbE_{\rmQ}(X)$}{
  let $T=\{i|i\in [p], \bbE_{\rmP^{I^*}}(X_i)\neq\bbE_{\rmQ}(X_i)\}$\;
  let $\cG$ be the subgraph of shift-$\cI$-EG induced by $T$\;
  let $U_T$ be the identified source nodes in $T$\;
  \While{$U_T=\varnothing$}{
    let $\cC$ be a chain component of $\cG$ with no incoming edges\;
    select shift intervention $I$ by running $\texttt{CliqueTree}(\cC,S)$ or $\texttt{Supermodular}(\cC,S)$\;
    perform $I$ and append it to $\cI$\;
    update $\cG$ and $U_T$ as the outer loop\;
   }
   set $a_i=\bbE_{\rmQ}(X_i)-\bbE_{\rmP^{I^*}}(X_i)$ for $i$ in $U_T$\;
   include perturbation targets $U_T$ and shift values $\{a_i\}_{i\in U_T}$ in $I^*$ and perform $I^*$\;
 }
 \KwOut{Matching Intervention $I^*$}
 \caption{Active Learning for Causal Mean Matching}
 \label{alg:1}
\end{algorithm}

\textbf{Clique Tree Strategy.}
Let $C(\cC)$ be the set of maximal cliques in the chordal graph $\cC$. There exists a \textit{clique tree} $\cT(\cC)$ with nodes in $C(\cC)$ and edges satisfying that $\forall C_1,C_2\in C(\cC)$, their intersection $C_1\cap C_2$ is a subset of any clique on the unique path between $C_1,C_2$ in $\cT(\cC)$ \citep{blair1993introduction}. %
Thus, deleting a clique which is not a leaf node in the clique tree will break $\cC$ into at least two connected components, each corresponding to a subtree in the clique tree.
Inspired by the central node algorithm \citep{greenewald2019sample, squires2020active}, we find the \textit{$S$-constrained central clique} of $\cT(\cC)$ by iterating through $C(\cC)$ and returning the clique with no more than $S$ nodes that separates the graph most, i.e., solving MinMaxC when interventions are constrained to be maximal cliques.
We denote this approach as \texttt{CliqueTree}.

\textbf{Supermodular Strategy.}
Our second approach, denoted \texttt{Supermodular}, optimizes a lower bound of the objective of MinMaxC. Consider the following equivalent formulation of MinMaxC
\begin{equation}\label{eq:4.3-1}
    \min_{A \subset V_{\cC}} \max_{i\in V_{\cC}} f_i(A),\quad |A|\leq S,
\end{equation}
where $V_{\cC}$ represents the nodes of $\cC$ and $\forall~i\in V_{\cC}$, $f_i(A) = \sum_{j\in V_\cC} g_{i,j}(A)$ with $g_{i,j}(A) = 1$ if $i$ and $j$ are the same or connected after removing nodes in $A$ from $\cC$ and $g_{i,j}(A) = 0$ otherwise.

MinMaxC \eqref{eq:4.3-1} resembles the problem of robust supermodular optimization \citep{krause2008robust}. Unfortunately, $f_i$ is not supermodular for chordal graphs (\rref{appendix:alg}). Therefore, we propose to optimize for a surrogate of $f_i$ defined as $\hat{f}_i(A) = \sum_{j\in\cC} \hat{g}_{i,j}(A)$, where
\begin{equation}\label{eq:4.3-2}
    \hat{g}_{i,j}(A) = \begin{cases}
    \frac{m_{i,j}(V_{\cC}-A)}{m_{i,j}(V_{\cC})},&\quad {i--j}~\mathrm{in}~{\cC}, \\
    0,&\quad \mathrm{otherwise}.
    \end{cases}
\end{equation}
Here $m_{i,j}(V_{\cC'})$ is the number of paths without cycles between $i$ and $j$ in $\cC'$ (deemed $0$ if $i$ or $j$ does not belong to $\cC'$ and $1$ if $i=j\in \cC'$) and $i--j$ means $i$ is either connected or equal to $j$. 
Comparing $\hat{g}_{i,j}$ with $g_{i,j}$, we see that $\hat{f}_i(A)$ is a lower bound of $f_i(A)$ for MinMaxC, which is tight when $\cC$ is a tree. 
We show that $\hat{f}_i$ is monotonic supermodular for all $i$ (\rref{appendix:alg}). 
Therefore, we consider \eqref{eq:4.3-2} with $f_i$ replaced by $\hat{f}_i$, which can be solved using the SATURATE algorithm \citep{krause2008robust}. 
Notably, the results returned by \texttt{Supermodular} can be quite different to those returned by \texttt{CliqueTree} since \texttt{Supermodular} is not constrained to pick a maximal clique; see Figure \ref{fig4}.

\begin{figure}[t]
     \centering
     \begin{subfigure}[b]{0.4\textwidth}
         \centering
         \includegraphics[width=0.75\textwidth]{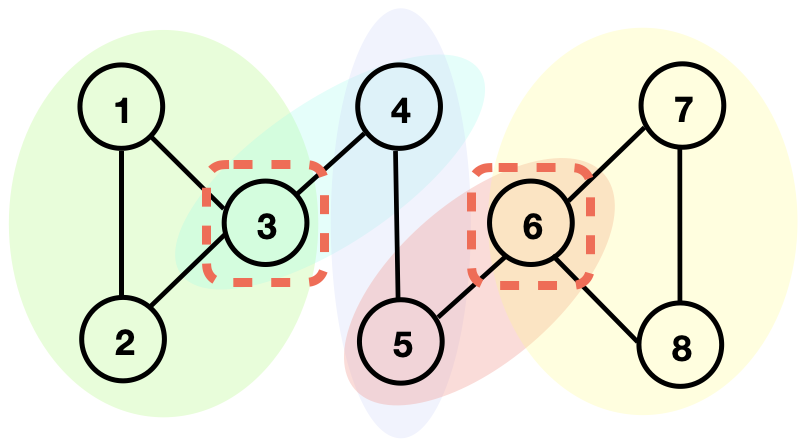}
         \caption{}
     \end{subfigure}
     \begin{subfigure}[b]{0.5\textwidth}
         \centering
         \includegraphics[width=0.7\textwidth]{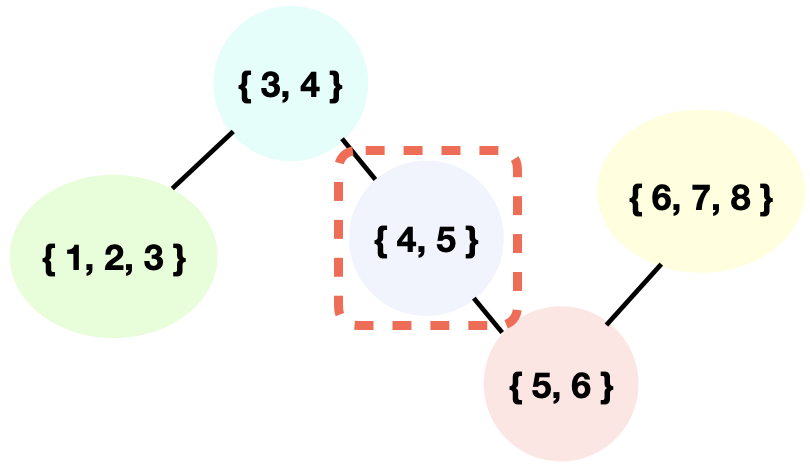}
         \caption{}
     \end{subfigure}
        \caption{Picking $2$ nodes in an undirected connected chordal graph $\cC$. \texttt{CliqueTree} picks $\{X_4,X_5\}$, while \texttt{Supermodular} picks the better $\{X_3,X_6\}$. \textbf{(a).} $\cC$; \textbf{(b).} Clique tree $\cT(\cC)$.
        }
        \label{fig4}
\end{figure}

\section{Theoretical Results}\label{sec:theory}
In this section we derive a \textit{worst-case} lower bound on the number of interventions for any algorithm to identify the source node in a chordal graph. 
Then we use this lower bound to show that our strategies are optimal up to a logarithmic factor.
This contrasts with the structure learning strategy, which may require exponentially more interventions than our strategy (Figure \ref{fig:shift-vs-structure-learning-gap}).

The worst case is with respect to all feasible orientations of an essential graph \citep{hauser2014two,shanmugam2015learning}, i.e., orientations corresponding to DAGs in the equivalence class. 
Given a chordal chain component $\cC$ of $\cG$, let $r_{\cC}$ be the number of maximal cliques in $\cC$, and $m_{\cC}$ be the size of the \textit{largest} maximal clique in $\cC$.
The following lemma provides a lower bound depending only on $m_{\cC}$.
\begin{lemma}\label{lemma:oracle-lb}
In the worst case over feasible orientations of $\cC$, any algorithm requires at least $\lceil \frac{m_{\cC}-1}{S}\rceil$ shift interventions to identify the source node, under the sparsity constraint $S$.
\end{lemma}

To give some intuition for this result, consider the case where the largest maximal clique is upstream of all other maximal cliques.
Given such an ordering, in the worst case, each intervention rules out only $S$ nodes in this clique (namely, the most downstream ones).
Now, we show that our two strategies need at most $ \lceil \log_2(r_{\cC}+1)\rceil\cdot\lceil\frac{m_{\cC}-1}{S}\rceil$ shift interventions for the same task.
\begin{lemma}\label{lemma:alg-lb}
In the worst case over feasible orientations of $\cC$, both \textnormal{\texttt{CliqueTree}} and \textnormal{\texttt{Supermodular}} require at most $ \lceil \log_2(r_\cC+1)\rceil\cdot\lceil\frac{m_\cC-1}{S}\rceil$ shift interventions to identify the source node, under the sparsity constraint $S$.
\end{lemma}

By combining \rref{lemma:oracle-lb} and \rref{lemma:alg-lb}, which consider subproblems of the causal mean matching problem, we obtain a bound on the number of shift interventions needed for solving the full causal mean matching problem. 
Let $r$ be the largest $r_{\cC}$ for all chain components $\cC$ of $\cG$:
\begin{theorem}\label{thm:2}
Algorithm \ref{alg:1} requires at most $\lceil \log_2(r+1)\rceil$ times more shift interventions, compared to that required by the optimal strategy, in the worst case over feasible orientations of $\cG$.
\end{theorem}

A direct application of this theorem is that, in terms of the number of interventions required to solve the causal mean matching problem, our algorithm is optimal in the worst case when $r=1$, i.e., when every chain component is a clique.
All proofs are provided in \rref{appendix:bound}.

\section{Experiments}\label{sec:experiments}

We now evaluate our algorithms in several synthetic settings.\footnote{Code is publicly available at: \url{https://github.com/uhlerlab/causal_mean_matching}.} Each setting considers a particular graph type, number of nodes $p$ in the graph and number of perturbation targets $|I^*|\leq p$ in the matching intervention. 
We generate $100$ problem instances in each setting. 
Every problem instance contains a DAG with $p$ nodes generated according to the graph type and a randomly sampled subset of $|I^*|$ nodes denoting the perturbation targets in the matching intervention.
We consider both, random graphs including Erdös-Rényi graphs \citep{erdHos1959renyi} and Barabási–Albert graphs \citep{albert2002statistical}, as well as structured chordal graphs, in particular, rooted tree graphs and moralized Erdös-Rényi graphs \citep{shanmugam2015learning}. 
The graph size $p$ in our simulations ranges from $10$ to $1000$, while the number of perturbation targets ranges from $1$ to $\min\{p,100\}$.

We compare our two subroutines for Algorithm \ref{alg:1}, \texttt{CliqueTree} and \texttt{Supermodular}, against three carefully constructed baselines. 
The \texttt{UpstreamRand} baseline follows Algorithm \ref{alg:1} where line 8 is changed to selecting $I$ randomly from $\cC$ without exceeding $S$, i.e., when there is no identified source node it randomly samples from the chain component with no incoming edge. 
This strategy highlights how much benefit is obtained from \texttt{CliqueTree} and \texttt{Supermodular} on top of upstream search. 
The \texttt{Coloring} baseline is modified from the coloring-based policy for structure learning \citep{shanmugam2015learning}, previously shown to perform competitively on large graphs \citep{squires2020active}. 
It first performs structure learning with the coloring-based policy, and then uses upstream search with known DAG. We also include an \texttt{Oracle} baseline, which does upstream search with known DAG.

In Figure~\ref{fig5} we present a subset of our results on Barabási–Albert graphs with $100$ nodes; similar behaviors are observed  in all other settings and shown in \rref{appendix:append-exp}. In Figure \ref{fig5:a}, we consider problem instances with varying size of $|I^*|$. 
Each algorithm is run with sparsity constraint $S=1$. 
We plot the number of extra interventions compared to \texttt{Oracle}, averaged across the $100$ problem instances.
As expected, \texttt{Coloring} requires the largest number of extra interventions.
This finding is consistent among different numbers of perturbation targets, since the same amount of interventions are used to learn the structure regardless of $I^*$. 
As $|I^*|$ increases, \texttt{CliqueTree} and \texttt{Supermodular} outperform \texttt{UpstreamRand}.
To further investigate this trend, we plot the rate of extra interventions\footnote{The rate is calculated by (\#\texttt{Strategy}-\#\texttt{UpstreamRand})/\#\texttt{UpstreamRand} where \# denotes the number of extra interventions compared to \texttt{Oracle} and \texttt{Strategy} can be \texttt{CliqueTree}, \texttt{Supermodular} or \texttt{UpstreamRand}.} used by \texttt{CliqueTree} and \texttt{Supermodular} relative to \texttt{UpstreamRand} in Figure \ref{fig5:b}. 
This figure shows that \texttt{CliqueTree} and \texttt{Supermodular} improve upon upstream search by up to $25\%$ as the number of perturbation targets increases. 
Finally, we consider the effect of the sparsity constraint $S$ in Figure \ref{fig5:c} with $|I^*|=50$. 
In line with the discussion in \rref{sec:4.3}, as $S$ increases, the task becomes easier for plain upstream search. 
However, when the number of perturbation targets is restricted, \texttt{CliqueTree} and \texttt{Supermodular} are superior, with \texttt{Supermodular} performing best in most cases.

\begin{figure}[t]
     \centering
     \begin{subfigure}[b]{0.3\textwidth}
         \centering
         \includegraphics[width=0.9\textwidth]{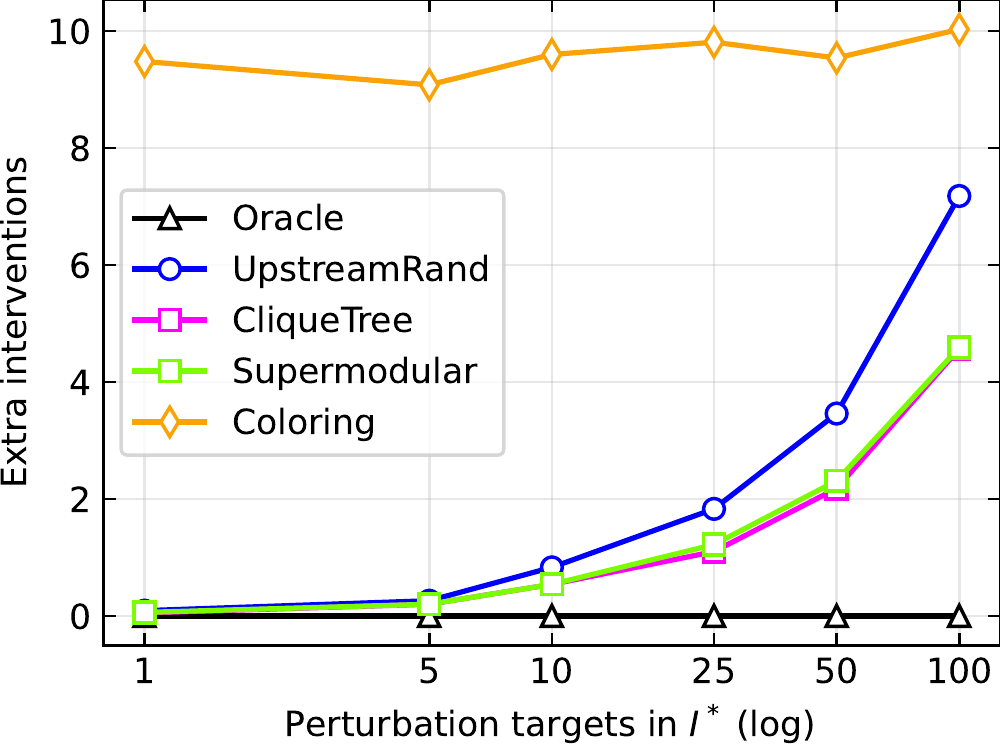}
         \caption{}
         \label{fig5:a}
     \end{subfigure}
     \hfill
     \begin{subfigure}[b]{0.3\textwidth}
         \centering
         \includegraphics[width=0.9\textwidth]{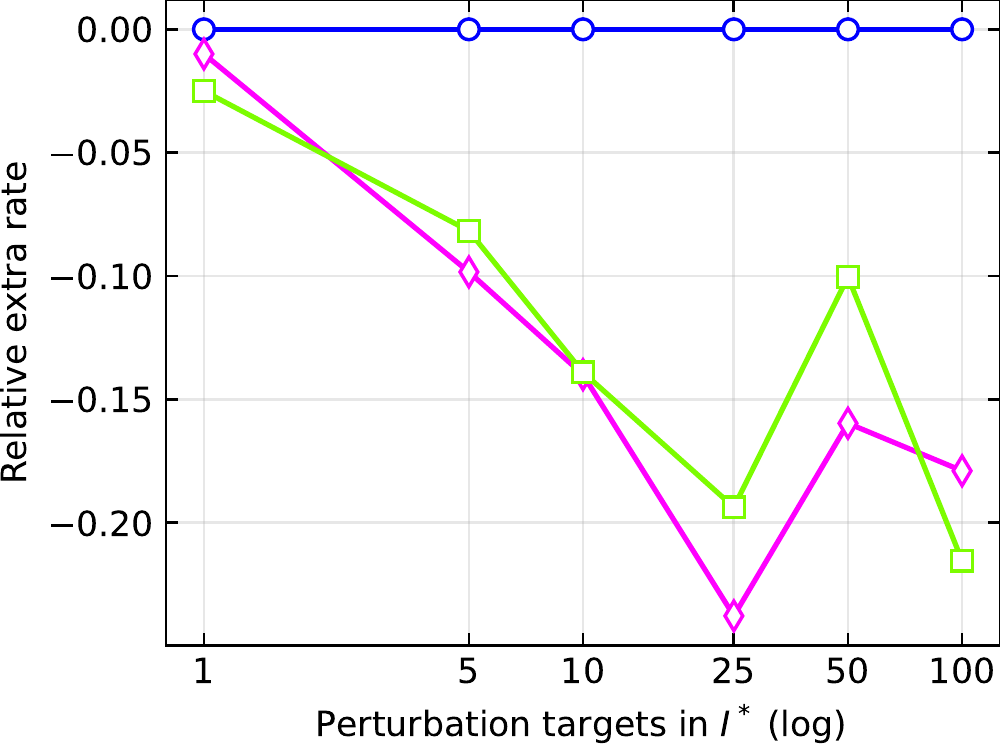}
         \caption{}
         \label{fig5:b}
     \end{subfigure}
    \hfill
     \begin{subfigure}[b]{0.3\textwidth}
         \centering
         \includegraphics[width=0.9\textwidth]{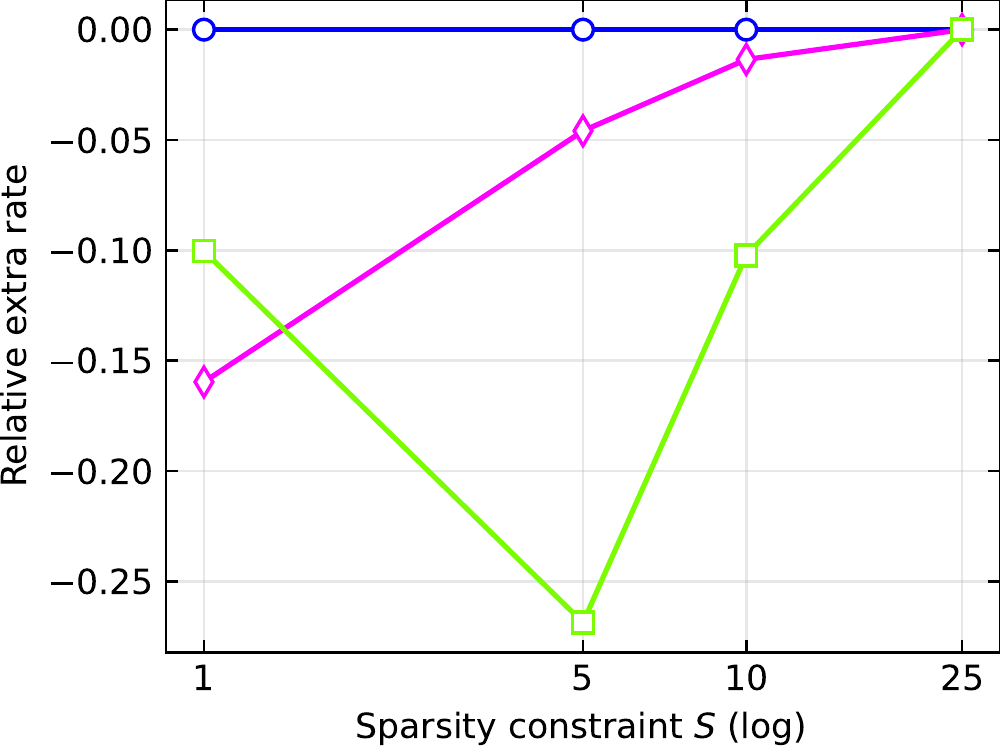}
         \caption{}
         \label{fig5:c}
     \end{subfigure}
        \caption{Barabási–Albert graphs with $100$ nodes. \textbf{(a).} Averaged (100 instances) numbers of extra interventions each algorithm (with sparsity constraint $S=1$) requires compared to \texttt{Oracle}, plotted against number of perturbation targets in $I^*$; \textbf{(b).} Rates of extra interventions \texttt{CliqueTree} and \texttt{Supermodular} ($S=1$) required relative to \texttt{UpstreamRand}, plotted against number of perturbation targets in $I^*$; \textbf{(c).} Relative extra rate ($|I^*|=50$), plotted against sparsity constraint $S$.
        }
        \label{fig5}
\end{figure}

\section{Discussion}\label{sec:discussion}

In this work, we introduced the \textit{causal mean matching} problem, which has important applications in medicine and engineering.
We aimed to develop active learning approaches for identifying the matching intervention using shift interventions.
Towards this end, we characterized the shift interventional Markov equivalence class and showed that it is in general more refined than previously defined equivalence classes.
We proposed two strategies for learning the matching intervention based on this characterization, and showed that they are optimal up to a logarithmic factor.
We reported experimental results on a range of settings to support these theoretical findings.

\textbf{Limitations and Future Work.} This work has various limitations that may be interesting to address in future work.
First, we focus on the task of matching a desired \textit{mean}, rather than an entire distribution.
This is an inherent limitation of deterministic shift interventions: as noted by \citet{hyttinen2012learning}, in the linear Gaussian setting, these interventions can \textit{only} modify the mean of the initial distribution.
Thus, matching the entire distribution, or other relevant statistics, will require broader classes of interventions.
Assumptions on the desired distribution are also required to rule out possibly non-realizable cases.
Second, we have focused on causal DAG models, which assume acyclicity and the absence of latent confounders.
In many realistic applications, this could be an overly optimistic assumption, requiring extensions of our results to the cyclic and/or causally insufficient setting.
Finally, throughout the main text, we have focused on the noiseless setting; we briefly discuss the noisy setting in \rref{appendix:noisy}, but there is much room for more extensive investigations.

\section*{Acknowledgements } C.~Squires was partially supported by an NSF Graduate Fellowship. All authors were partially supported by NSF (DMS-1651995), ONR (N00014-17-
1-2147 and N00014-18-1-2765), the MIT-IBM Watson AI Lab, and a Simons Investigator Award to C.~Uhler.

\bibliographystyle{apalike}
\bibliography{main}

\begin{thebibliography}{}

\bibitem[Andersson et~al., 1997]{andersson1997characterization+}
Andersson, S.~A., Madigan, D., Perlman, M.~D., et~al. (1997).
\newblock A characterization of markov equivalence classes for acyclic
  digraphs.
\newblock {\em Annals of statistics}, 25(2):505--541.

\bibitem[Galinier et~al., 1995]{galinier1995chordal}
Galinier, P., Habib, M., and Paul, C. (1995).
\newblock Chordal graphs and their clique graphs.
\newblock In {\em International Workshop on Graph-Theoretic Concepts in
  Computer Science}, pages 358--371. Springer.

\bibitem[Greenewald et~al., 2019]{greenewald2019sample+}
Greenewald, K., Katz, D., Shanmugam, K., Magliacane, S., Kocaoglu, M.,
  Boix~Adsera, E., and Bresler, G. (2019).
\newblock Sample efficient active learning of causal trees.
\newblock In {\em Advances in Neural Information Processing Systems},
  volume~32. Curran Associates, Inc.

\bibitem[Hagberg et~al., 2008]{hagberg2008exploring}
Hagberg, A., Swart, P., and S~Chult, D. (2008).
\newblock Exploring network structure, dynamics, and function using networkx.
\newblock Technical report, Los Alamos National Lab.(LANL), Los Alamos, NM
  (United States).

\bibitem[Hauser and B{\"u}hlmann, 2014]{hauser2014two+}
Hauser, A. and B{\"u}hlmann, P. (2014).
\newblock Two optimal strategies for active learning of causal models from
  interventional data.
\newblock {\em International Journal of Approximate Reasoning}, 55(4):926--939.

\bibitem[Koller and Friedman, 2009]{koller2009probabilistic+}
Koller, D. and Friedman, N. (2009).
\newblock {\em Probabilistic graphical models: principles and techniques}.
\newblock MIT press.

\bibitem[Krause et~al., 2008]{krause2008robust+}
Krause, A., McMahan, H.~B., Guestrin, C., and Gupta, A. (2008).
\newblock Robust submodular observation selection.
\newblock {\em Journal of Machine Learning Research}, 9(12).

\bibitem[Kruskal, 1956]{kruskal1956shortest}
Kruskal, J.~B. (1956).
\newblock On the shortest spanning subtree of a graph and the traveling
  salesman problem.
\newblock {\em Proceedings of the American Mathematical society}, 7(1):48--50.

\bibitem[Lalou et~al., 2018]{lalou2018critical+}
Lalou, M., Tahraoui, M.~A., and Kheddouci, H. (2018).
\newblock The critical node detection problem in networks: A survey.
\newblock {\em Computer Science Review}, 28:92--117.

\bibitem[Meek, 1995]{meek2013causal+}
Meek, C. (1995).
\newblock Causal inference and causal explanation with background knowledge.
\newblock In {\em Proceedings of the Eleventh Conference on Uncertainty in
  Artificial Intelligence}, UAI'95, page 403–410, San Francisco, CA, USA.
  Morgan Kaufmann Publishers Inc.

\bibitem[Sedgewick, 2001]{sedgewick2001algorithms}
Sedgewick, R. (2001).
\newblock {\em Algorithms in C, part 5: graph algorithms}.
\newblock Pearson Education.

\bibitem[Shen et~al., 2012]{shen2012exact}
Shen, S., Smith, J.~C., and Goli, R. (2012).
\newblock Exact interdiction models and algorithms for disconnecting networks
  via node deletions.
\newblock {\em Discrete Optimization}, 9(3):172--188.

\bibitem[Squires et~al., 2020]{squires2020permutation+}
Squires, C., Wang, Y., and Uhler, C. (2020).
\newblock Permutation-based causal structure learning with unknown intervention
  targets.
\newblock In {\em Conference on Uncertainty in Artificial Intelligence}, pages
  1039--1048. PMLR.

\bibitem[Valiant, 1979]{valiant1979complexity}
Valiant, L.~G. (1979).
\newblock The complexity of enumeration and reliability problems.
\newblock {\em SIAM Journal on Computing}, 8(3):410--421.

\bibitem[Vandenberghe and Andersen, 2015]{vandenberghe2015chordal}
Vandenberghe, L. and Andersen, M.~S. (2015).
\newblock Chordal graphs and semidefinite optimization.
\newblock {\em Foundations and Trends in Optimization}, 1(4):241--433.

\bibitem[Yang et~al., 2018]{yang2018characterizing+}
Yang, K., Katcoff, A., and Uhler, C. (2018).
\newblock Characterizing and learning equivalence classes of causal dags under
  interventions.
\newblock In {\em International Conference on Machine Learning}, pages
  5541--5550. PMLR.

\end{thebibliography}


\begin{thebibliography}{}

\bibitem[Agrawal et~al., 2019]{agrawal2019abcd}
Agrawal, R., Squires, C., Yang, K., Shanmugam, K., and Uhler, C. (2019).
\newblock Abcd-strategy: Budgeted experimental design for targeted causal
  structure discovery.
\newblock In {\em The 22nd International Conference on Artificial Intelligence
  and Statistics}, pages 3400--3409. PMLR.

\bibitem[Albert and Barab{\'a}si, 2002]{albert2002statistical}
Albert, R. and Barab{\'a}si, A.-L. (2002).
\newblock Statistical mechanics of complex networks.
\newblock {\em Reviews of modern physics}, 74(1):47.

\bibitem[Andersson et~al., 1997]{andersson1997characterization}
Andersson, S.~A., Madigan, D., Perlman, M.~D., et~al. (1997).
\newblock A characterization of markov equivalence classes for acyclic
  digraphs.
\newblock {\em Annals of statistics}, 25(2):505--541.

\bibitem[Badsha et~al., 2019]{badsha2019learning}
Badsha, M., Fu, A.~Q., et~al. (2019).
\newblock Learning causal biological networks with the principle of mendelian
  randomization.
\newblock {\em Frontiers in genetics}, 10:460.

\bibitem[Blair and Peyton, 1993]{blair1993introduction}
Blair, J.~R. and Peyton, B. (1993).
\newblock An introduction to chordal graphs and clique trees.
\newblock In {\em Graph theory and sparse matrix computation}, pages 1--29.
  Springer.

\bibitem[Bubeck and Cesa-Bianchi, 2012]{bubeck2012regret}
Bubeck, S. and Cesa-Bianchi, N. (2012).
\newblock Regret analysis of stochastic and nonstochastic multi-armed bandit
  problems.
\newblock In {\em Foundations and Trends® in Machine Learning}, pages 1--122.

\bibitem[Cahan et~al., 2014]{cahan2014cellnet}
Cahan, P., Li, H., Morris, S.~A., Da~Rocha, E.~L., Daley, G.~Q., and Collins,
  J.~J. (2014).
\newblock Cellnet: network biology applied to stem cell engineering.
\newblock {\em Cell}, 158(4):903--915.

\bibitem[de~Kroon et~al., 2020]{de2020causal}
de~Kroon, A.~A., Belgrave, D., and Mooij, J.~M. (2020).
\newblock Causal discovery for causal bandits utilizing separating sets.
\newblock {\em arXiv preprint arXiv:2009.07916}.

\bibitem[Dominguez et~al., 2016]{dominguez2016beyond}
Dominguez, A.~A., Lim, W.~A., and Qi, L.~S. (2016).
\newblock Beyond editing: repurposing crispr--cas9 for precision genome
  regulation and interrogation.
\newblock {\em Nature reviews Molecular cell biology}, 17(1):5.

\bibitem[D’Alessio et~al., 2015]{d2015systematic}
D’Alessio, A.~C., Fan, Z.~P., Wert, K.~J., Baranov, P., Cohen, M.~A., Saini,
  J.~S., Cohick, E., Charniga, C., Dadon, D., Hannett, N.~M., et~al. (2015).
\newblock A systematic approach to identify candidate transcription factors
  that control cell identity.
\newblock {\em Stem cell reports}, 5(5):763--775.

\bibitem[Erd{\H{o}}s and R{\'e}nyi, 1960]{erdHos1959renyi}
Erd{\H{o}}s, P. and R{\'e}nyi, A. (1960).
\newblock On the evolution of random graphs.
\newblock {\em Publ. Math. Inst. Hung. Acad. Sci}, 5(1):17--60.

\bibitem[Friedman et~al., 2000]{friedman2000using}
Friedman, N., Linial, M., Nachman, I., and Pe'er, D. (2000).
\newblock Using bayesian networks to analyze expression data.
\newblock {\em Journal of computational biology}, 7(3-4):601--620.

\bibitem[Ghassami et~al., 2018]{ghassami2018budgeted}
Ghassami, A., Salehkaleybar, S., Kiyavash, N., and Bareinboim, E. (2018).
\newblock Budgeted experiment design for causal structure learning.
\newblock In {\em International Conference on Machine Learning}, pages
  1724--1733. PMLR.

\bibitem[Greenewald et~al., 2019]{greenewald2019sample}
Greenewald, K., Katz, D., Shanmugam, K., Magliacane, S., Kocaoglu, M.,
  Boix~Adsera, E., and Bresler, G. (2019).
\newblock Sample efficient active learning of causal trees.
\newblock In {\em Advances in Neural Information Processing Systems},
  volume~32. Curran Associates, Inc.

\bibitem[Hauser and B{\"u}hlmann, 2012]{hauser2012characterization}
Hauser, A. and B{\"u}hlmann, P. (2012).
\newblock Characterization and greedy learning of interventional markov
  equivalence classes of directed acyclic graphs.
\newblock {\em The Journal of Machine Learning Research}, 13(1):2409--2464.

\bibitem[Hauser and B{\"u}hlmann, 2014]{hauser2014two}
Hauser, A. and B{\"u}hlmann, P. (2014).
\newblock Two optimal strategies for active learning of causal models from
  interventional data.
\newblock {\em International Journal of Approximate Reasoning}, 55(4):926--939.

\bibitem[He and Geng, 2008]{he2008active}
He, Y.-B. and Geng, Z. (2008).
\newblock Active learning of causal networks with intervention experiments and
  optimal designs.
\newblock {\em Journal of Machine Learning Research}, 9(Nov):2523--2547.

\bibitem[Hyttinen et~al., 2012]{hyttinen2012learning}
Hyttinen, A., Eberhardt, F., and Hoyer, P.~O. (2012).
\newblock Learning linear cyclic causal models with latent variables.
\newblock {\em The Journal of Machine Learning Research}, 13(1):3387--3439.

\bibitem[Hyttinen et~al., 2013]{hyttinen2013experiment}
Hyttinen, A., Eberhardt, F., and Hoyer, P.~O. (2013).
\newblock Experiment selection for causal discovery.
\newblock {\em Journal of Machine Learning Research}, 14:3041--3071.

\bibitem[Jaber et~al., 2020]{jaber2020causal}
Jaber, A., Kocaoglu, M., Shanmugam, K., and Bareinboim, E. (2020).
\newblock Causal discovery from soft interventions with unknown targets:
  Characterization and learning.
\newblock {\em Advances in Neural Information Processing Systems}, 33.

\bibitem[Kocaoglu et~al., 2017]{kocaoglu2017cost}
Kocaoglu, M., Dimakis, A., and Vishwanath, S. (2017).
\newblock Cost-optimal learning of causal graphs.
\newblock In {\em International Conference on Machine Learning}, pages
  1875--1884. PMLR.

\bibitem[Koller and Friedman, 2009]{koller2009probabilistic}
Koller, D. and Friedman, N. (2009).
\newblock {\em Probabilistic graphical models: principles and techniques}.
\newblock MIT press.

\bibitem[Krause et~al., 2008]{krause2008robust}
Krause, A., McMahan, H.~B., Guestrin, C., and Gupta, A. (2008).
\newblock Robust submodular observation selection.
\newblock {\em Journal of Machine Learning Research}, 9(12).

\bibitem[Lalou et~al., 2018]{lalou2018critical}
Lalou, M., Tahraoui, M.~A., and Kheddouci, H. (2018).
\newblock The critical node detection problem in networks: A survey.
\newblock {\em Computer Science Review}, 28:92--117.

\bibitem[Lattimore et~al., 2016]{lattimore2016causal}
Lattimore, F., Lattimore, T., and Reid, M.~D. (2016).
\newblock Causal bandits: Learning good interventions via causal inference.
\newblock In {\em Advances in Neural Information Processing Systems},
  volume~29. Curran Associates, Inc.

\bibitem[Lee and Bareinboim, 2018]{lee2018structural}
Lee, S. and Bareinboim, E. (2018).
\newblock Structural causal bandits: where to intervene?
\newblock {\em Advances in Neural Information Processing Systems 31}, 31.

\bibitem[Meek, 1995]{meek2013causal}
Meek, C. (1995).
\newblock Causal inference and causal explanation with background knowledge.
\newblock In {\em Proceedings of the Eleventh Conference on Uncertainty in
  Artificial Intelligence}, UAI'95, page 403–410, San Francisco, CA, USA.
  Morgan Kaufmann Publishers Inc.

\bibitem[Rackham et~al., 2016]{rackham2016predictive}
Rackham, O.~J., Firas, J., Fang, H., Oates, M.~E., Holmes, M.~L., Knaupp,
  A.~S., Suzuki, H., Nefzger, C.~M., Daub, C.~O., Shin, J.~W., et~al. (2016).
\newblock A predictive computational framework for direct reprogramming between
  human cell types.
\newblock {\em Nature genetics}, 48(3):331.

\bibitem[Rothenh\"{a}usler et~al., 2015]{rothenhausler2015backshift}
Rothenh\"{a}usler, D., Heinze, C., Peters, J., and Meinshausen, N. (2015).
\newblock Backshift: Learning causal cyclic graphs from unknown shift
  interventions.
\newblock In {\em Advances in Neural Information Processing Systems},
  volume~28. Curran Associates, Inc.

\bibitem[Shanmugam et~al., 2015]{shanmugam2015learning}
Shanmugam, K., Kocaoglu, M., Dimakis, A.~G., and Vishwanath, S. (2015).
\newblock Learning causal graphs with small interventions.
\newblock In {\em Advances in Neural Information Processing Systems},
  volume~28. Curran Associates, Inc.

\bibitem[Spirtes et~al., 2000]{spirtes2000causation}
Spirtes, P., Glymour, C.~N., Scheines, R., and Heckerman, D. (2000).
\newblock {\em Causation, prediction, and search}.
\newblock MIT press.

\bibitem[Squires et~al., 2020a]{squires2020active}
Squires, C., Magliacane, S., Greenewald, K., Katz, D., Kocaoglu, M., and
  Shanmugam, K. (2020a).
\newblock Active structure learning of causal dags via directed clique trees.
\newblock In {\em Advances in Neural Information Processing Systems},
  volume~33, pages 21500--21511. Curran Associates, Inc.

\bibitem[Squires et~al., 2020b]{squires2020permutation}
Squires, C., Wang, Y., and Uhler, C. (2020b).
\newblock Permutation-based causal structure learning with unknown intervention
  targets.
\newblock In {\em Conference on Uncertainty in Artificial Intelligence}, pages
  1039--1048. PMLR.

\bibitem[Touchette and Lloyd, 2004]{touchette2004information}
Touchette, H. and Lloyd, S. (2004).
\newblock Information-theoretic approach to the study of control systems.
\newblock {\em Physica A: Statistical Mechanics and its Applications},
  331(1-2):140--172.

\bibitem[Triantafillou et~al., 2017]{triantafillou2017predicting}
Triantafillou, S., Lagani, V., Heinze-Deml, C., Schmidt, A., Tegner, J., and
  Tsamardinos, I. (2017).
\newblock Predicting causal relationships from biological data: Applying
  automated causal discovery on mass cytometry data of human immune cells.
\newblock {\em Scientific reports}, 7(1):1--11.

\bibitem[Verma and Pearl, 1991]{verma1991equivalence}
Verma, T. and Pearl, J. (1991).
\newblock {\em Equivalence and synthesis of causal models}.
\newblock UCLA, Computer Science Department.

\bibitem[Vierbuchen et~al., 2010]{vierbuchen2010direct}
Vierbuchen, T., Ostermeier, A., Pang, Z.~P., Kokubu, Y., S{\"u}dhof, T.~C., and
  Wernig, M. (2010).
\newblock Direct conversion of fibroblasts to functional neurons by defined
  factors.
\newblock {\em Nature}, 463(7284):1035--1041.

\bibitem[Wodo et~al., 2015]{wodo2015automated}
Wodo, O., Zola, J., Pokuri, B. S.~S., Du, P., and Ganapathysubramanian, B.
  (2015).
\newblock Automated, high throughput exploration of
  process--structure--property relationships using the mapreduce paradigm.
\newblock {\em Materials discovery}, 1:21--28.

\bibitem[Yabe et~al., 2018]{yabe2018causal}
Yabe, A., Hatano, D., Sumita, H., Ito, S., Kakimura, N., Fukunaga, T., and
  Kawarabayashi, K.-i. (2018).
\newblock Causal bandits with propagating inference.
\newblock In {\em International Conference on Machine Learning}, pages
  5512--5520. PMLR.

\bibitem[Yang et~al., 2018]{yang2018characterizing}
Yang, K., Katcoff, A., and Uhler, C. (2018).
\newblock Characterizing and learning equivalence classes of causal dags under
  interventions.
\newblock In {\em International Conference on Machine Learning}, pages
  5541--5550. PMLR.

\end{thebibliography}

\newpage
\appendix
\tableofcontents
\addtocontents{toc}{\protect\setcounter{tocdepth}{2}}
\newpage

\section{Preliminaries}\label{appendix:pre}
\subsection{Meek Rules}
Given any Markov equivalence class of DAGs with shared directed and undirected edges, the corresponding essential graph $\cE$ can be obtained using a set of logical relations known as Meek rules \citepappendix{meek2013causal+}. The Meek rules are stated in the following proposition.
\begin{proposition}[Meek Rules \citepappendix{meek2013causal+}]\label{prop:1}
We can infer all directed edges in $\cE$ using the following four rules:
\begin{enumerate}
    \item If $i\rightarrow j - k$ and $i$ is not adjacent to $k$, then $j\rightarrow k$.
    \item If $i\rightarrow j \rightarrow k$ and $i-k$, then $i\rightarrow k$.
    \item If $i-j, i-k, i-l, j\rightarrow k, l\rightarrow k$ and $j$ is not adjacent to $l$, then $i\rightarrow k$.
    \item If $i-j, i-k, i-l, j\leftarrow k, l\rightarrow k$ and $j$ is not adjacent to $l$, then $i\rightarrow j$.
\end{enumerate}
\end{proposition}
Figure \ref{afig1} illustrates these four rules.

\begin{figure}[ht]
     \centering
     \begin{subfigure}[b]{0.22\textwidth}
         \centering
         \includegraphics[width=\textwidth]{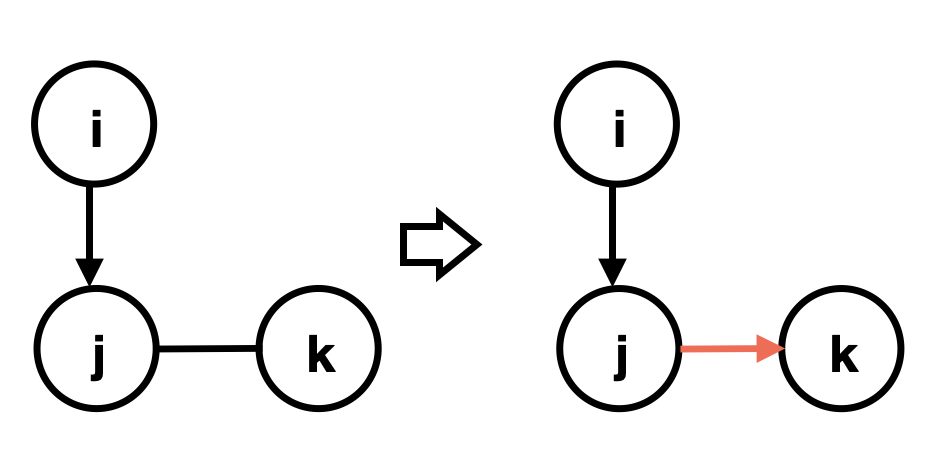}
         \caption{R1}
     \end{subfigure}
     \hfill
     \begin{subfigure}[b]{0.22\textwidth}
         \centering
         \includegraphics[width=\textwidth]{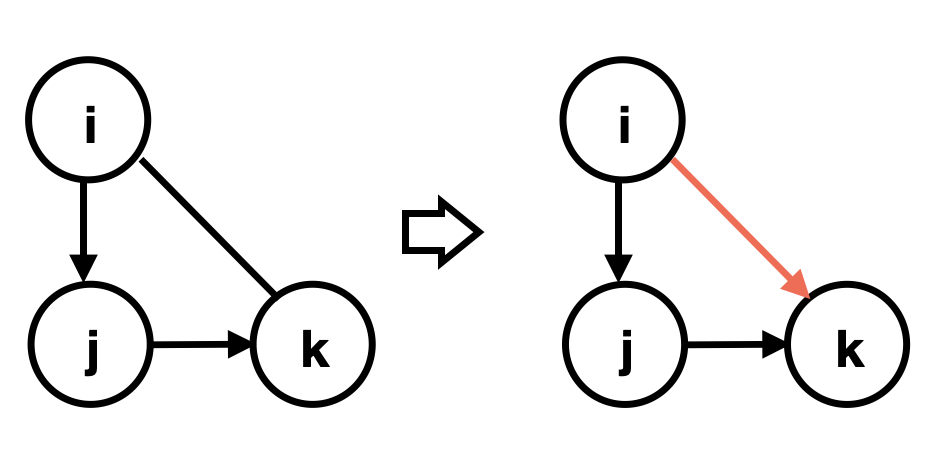}
         \caption{R2}
     \end{subfigure}
    \hfill
     \begin{subfigure}[b]{0.23\textwidth}
         \centering
         \includegraphics[width=0.95\textwidth]{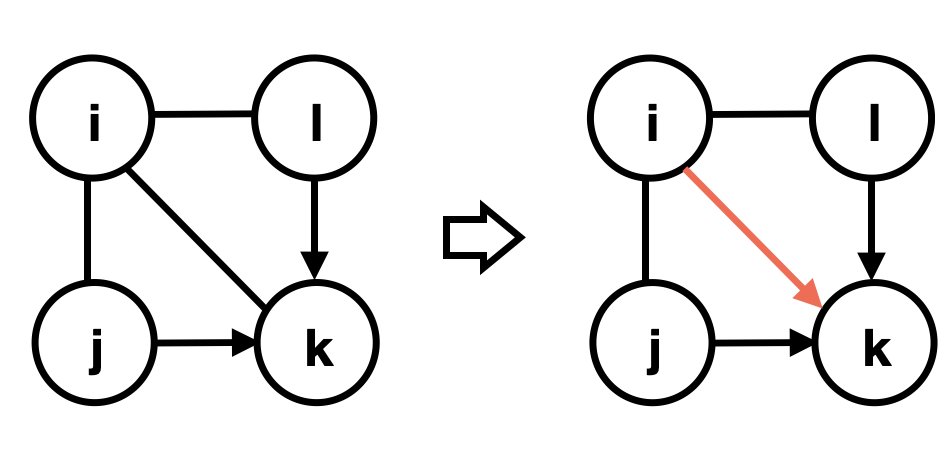}
         \caption{R3}
     \end{subfigure}
    \hfill
    \begin{subfigure}[b]{0.23\textwidth}
         \centering
         \includegraphics[width=0.95\textwidth]{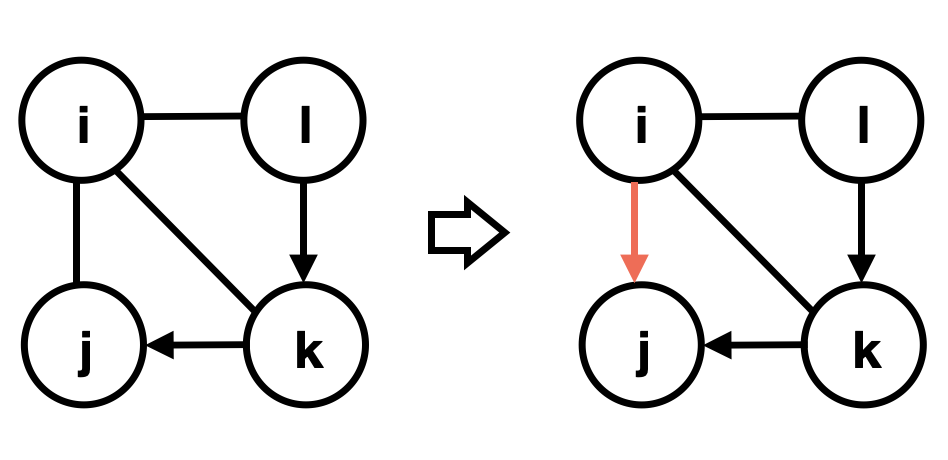}
         \caption{R4}
     \end{subfigure}
        \caption{Meek Rules.
        }
        \label{afig1}
\end{figure}

\section{Proof of Exact Matching}\label{appendix:exact}
\begin{proof}[Proof of Lemma \ref{lm:1}]
Without loss of generality, assume $1,2,...,p$ is the topological order of the underlying DAG $\cG$, i.e., $j\in \pa_{\cG}(i)$ implies $j<i$. We will first construct $I^*$ such that $\bbE_{\rmP^{I^*}}(X)=\bbE_{\rmQ}(X)$, and then show that $I^*$ is unique.

\textbf{Existence:} Denote $i_1$ as the smallest $i\in [p]$ such that $\bbE_{\rmP}(X_{i})\neq \bbE_{\rmQ}(X_i)$. Witout loss of generality we assume that $i_1$ exists (if $i_1$ does not exists, then $I^*=\varnothing$ suffices since $\bbE_{\rmP}(X)=\bbE_{\rmQ}(X)$). 

Let $I_1$ be the shift intervention with perturbation target $i_1$ and shift values $a_{i_1}= \bbE_{\rmQ}(X_{i_1})-\bbE_{\rmP}(X_{i_1})$. Since $\rmP^{I_1}(X_{i_1}=x+a_{i_1}|X_{\pa_{\cG}(i_1)})=\rmP(X_{i_1}=x|X_{\pa_{\cG}(i_1)})$ and $\rmP^{I_1}(X_{\pa_{\cG}(i_1)})=\rmP(X_{\pa_{\cG}(i_1)})$ by definition, we have
\[
\rmP^{I_1}(X_{i_1}=x+a_{i_1})=\rmP(X_{i_1}=x).
\]
Thus $\bbE_{\rmP^{I_1}}(X_{i_1}) = \bbE_{\rmP}(X_{i_1})+a_{i_1} = \bbE_{\rmQ}(X_{i_1})$. Also $\bbE_{\rmP^{I_1}}(X_{i}) = \bbE_{\rmQ}(X_{i})$ for $i<i_1$. Denote $i_2$ as the smallest $i\in [p]$ such that $\bbE_{\rmP^{I_1}}(X_i)\neq \bbE_{\rmQ}(X_i)$. If $i_2$ does not exists, then $I^*=I_1$ suffices. Otherwise $i_2>i_1$. 

Let $I_2$ be the shift intervention with perturbation target $i_1,i_2$ and corresponding shift values $a_{i_1}$ and $a_{i_2} = \bbE_{\rmQ}(X_{i_2})-\bbE_{\rmP^{I_1}}(X_{i_2})$. We have $\rmP^{I_2}(X_{i_2}=x+a_{i_2}|X_{\pa_{\cG}(i_2)})=\rmP(X_{i_2}=x|X_{\pa_{\cG}(i_2)})=\rmP^{I_1}(X_{i_2}=x|X_{\pa_{\cG}(i_2)})$ and $\rmP^{I_2}(X_{\pa_{\cG}(i_2)})=\rmP^{I_1}(X_{\pa_{\cG}(i_2)})$ by definition, the topological order, and $i_2>i_1$. Then
\[
    \rmP^{I_2}(X_{i_2}=x+a_{i_2}) = \rmP^{I_1}(X_{i_2}=x).
\]
Thus $\bbE_{\rmP^{I_2}}(X_{i_2}) = \bbE_{\rmP^{I_1}}(X_{i_2})+a_{i_2} = \bbE_{\rmQ}(X_{i_2})$. Also $\bbE_{\rmP^{I_2}}(X_{i})=\bbE_{\rmP^{I_1}}(X_{i})=\bbE_{\rmQ}(X_{i})$ for $i<i_2$. By iterating this process, we will reach $I_k$ for some $k\leq p$ such that there is no $i$ with $\bbE_{\rmP^{I_k}}(X_i)\neq \bbE_{\rmQ}(X_k)$. Taking $I^*=I_k$ suffices.

\textbf{Uniqueness:} If there exists $I_1^*\neq I_2^*$ such that $\bbE_{\rmP^{I_1^*}}(X)=\bbE_{\rmP^{I_2^*}}(X)=\bbE_{\rmQ}(X)$, let $i\in [p]$ be the smallest index such that either $i$ has different shift values in $I_1^*$ and $I_2^*$, or $i$ is only in one intervention's perturbation targets. In either case, we have $\rmP^{I_1^*}(X_{\pa_{\cG}(i)}) = \rmP^{I_2^*}(X_{\pa_{\cG}(i)})$ by the topological order and $\rmP^{I_1^*}(X_{i}=x|X_{\pa_{\cG}(i)})= \rmP^{I_2^*}(X_{i}=x+a|X_{\pa_{\cG}(i)})$ for some $a\neq 0$. Thus $\rmP^{I_1^*}(X_{i}=x)= \rmP^{I_2^*}(X_{i}=x+a)$ contradicting $\bbE_{\rmP^{I_1^*}}(X_i)=\bbE_{\rmP^{I_2^*}}(X_i)$.
\end{proof}

\section{Proof of Identifiability}\label{appendix:identify}
\subsection{Shift Interventional MEC}
\begin{proof}[Proof of Lemma \ref{lm:2}]
For any distribution $f$ that factorizes according to $\cG$ and shift intervention $I$, let $i\in I$ be any source w.r.t. $I$. By definition, $\an_{\cG}(i)\cap I=\varnothing$. Thus $\pa_{\cG}(i)$ contains neither a member nor a descendant of $I$, i.e., there does not exists $j\in \pa_{\cG}(i)$ and $k\in I$ such that there is a direct path from $k$ to $j$ or $k=j$. Hence we have $f^I(X_{\pa_{\cG}(i)}) = f(X_{\pa_{\cG}(i)})$, which gives
\[
    f^I(X_{i}=x+a_i) = f(X_{i}=x).
\]
Therefore $\bbE_{f^I}(X_i) = \bbE_{f}(X_i)+a_i$.

On the other hand, if $i\in I$ is not a source w.r.t. $I$, consider the following linear Gaussian model,
\[
X_j = \sum_{k\in \pa_{\cG}(j)} \beta_{kj} X_k + \epsilon_j, \quad \forall j\in [p],
\]
where $\beta_{kj}$ are deterministic scalars and $\epsilon_j\sim\mathcal{N}(0,1)$ are i.i.d. random variables.

Since $i$ is not a source in $I$, there exists a source $i'$ in $I$ such that there is a directed path $i' = i_0 \to i_1 \rightarrow \dots \rightarrow i_\ell$. 
From above, $\bbE_{f^I}(X_{i'})= \bbE_{f}(X_{i'})+a_{i'}$ for $a_{i'}\neq 0$.
Consider setting $\beta_{i_0,i_1} = 2|a_i|/a_{i'}$, $\beta_{i_k,i_{k+1}} = 1$ for $k = 1,\ldots,\ell-1$, and the remaining edge weights to $\epsilon > 0$. For $\epsilon$ sufficiently small, we have that $\bbE_{f^I} (X_i) \geq \bbE_f (X_i) + 1.5|a_i|$, i.e., we cannot have that $\bbE_{f^I}(X_i) = \bbE_f (X_i) + a_i$.
\end{proof}

\begin{proof}[Proof of Theorem \ref{thm:1}]
Denote $\cI=\{I_1,...,I_m\}$. For $k\in [m]=\{1,...,m\}$, let $\hat{I}_k$ and $\hat{I}_k'$ be the collection of source nodes in $I_k$ in $\cG_1$ and $\cG_2$, respectively. From \rref{def:1}, we know that $\cG_1$ and $\cG_2$ are in the same shift-$\cI$-MEC if and only if they are in the same $\cI$-MEC and, for any pair $(f,\{f^{I_k}\}_{k\in [m]})$ that is $\cI$-Markov w.r.t. both $\cG_1$ and $\cG_2$, it satisfies  
\begin{equation}\label{aeq:3}
   \bbE_{f^{I_k}}(X_i)=\bbE_{f}(X_i)+a_{i}, \quad \forall i \in \hat{I}_k, \forall k \in [m], 
\end{equation}
if and only if it also satisfies
\begin{equation}\label{aeq:4}
   \bbE_{f^{I_k}}(X_i)=\bbE_{f}(X_i)+a_{i}, \quad \forall i \in \hat{I}_k', \forall k \in [m].
\end{equation}

By Lemma \ref{lm:2}, we know that $\hat{I}_k'\subset \hat{I}_k$ for all $k\in [m]$. Otherwise we can find a pair $(f,\{f^{I_k}\}_{k\in [m]})$ that violates \eqref{aeq:4} for $i\in \hat{I}_k'\setminus \hat{I}_k$. Similarly, we have $\hat{I}_k\subset \hat{I}_k'$. Therefore $\hat{I}_k=\hat{I}_k'$. In this case, \eqref{aeq:3} is equivalent to \eqref{aeq:4}.

Hence, $\cG_1$ and $\cG_2$ are in the same shift-$\cI$-MEC if and only if they are in the same $\cI$-MEC and they have the same source nodes of $I$ for every $I\in \cI$. From Theorem 3.9 in \citetappendix{yang2018characterizing+}, we know that $\cG_1$ and $\cG_2$ are in the same $\cI$-MEC if and only if they share the same skeleton, $v$-structures and directed edges $\{i\rightarrow j|i\in I,j\notin I, I\in \cI, i-j\}$. Therefore, $\cG_1$ and $\cG_2$ are in the same shift-$\cI$-MEC if and only if they have the same skeleton, $v$-structures, directed edges $\{i\rightarrow j|i\in I,j\notin I, I\in \cI, i-j\}$, as well as source nodes of $I$ for every $I\in \cI$.
\end{proof}

Let $\cD$ be any DAG, suppose that $\cI=\{I_1,...,I_m\}$ and $\hat{I}_k$ is the collection of source nodes in $I_k$ in $\cD$ for $k\in [m]$. Then as a direct corollary of \rref{thm:1}, we can represent a shift interventional Markov equivalence class with a (general) interventional Markov equivalence class.
\begin{corollary}\label{cor:1}
Let $\hat{\cI}=\cI\cup\{\hat{I}_k|k\in [m]\}$; a DAG $\cD'$ is shift-$\cI$-Markov equivalent to $\cD$ if and only if $\cD'$ is $\hat{\cI}$-Markov equivalent to $\cD$.
\end{corollary}
\begin{proof}
The proof follws as a direct application of \rref{thm:1}, Theorem 3.9 in \citetappendix{yang2018characterizing+}, and the fact that there are no edges between nodes in $\hat{I}_k$.
\end{proof}

\subsection{Mean Interventional Faithfulness}
\begin{proof}[Proof of Lemma \ref{lm:3}]
If Assumption \ref{assumption:1} holds, then for any $i\notin T$, since $\bbE_{\rmP}(X_i)= \bbE_{\rmQ}(X_i)$, then $i\notin I^*$ and $\an_{\cG}(i)\cap I^*=\varnothing$. Let $j\in T$ such that there is an edge $i-j$ between $i$ and $j$. Since $\bbE_{\rmP}(X_j)\neq \bbE_{\rmQ}(X_j)$, there is either $j\in I^*$ or $\an_{\cG}(j)\cap I^*\neq\varnothing$. Therefore if $j\rightarrow i$, then $\an_{\cG}(i)\cap I^*\neq\varnothing$, a contradiction. Thus $j\leftarrow i$.

Conversely, if Assumption \ref{assumption:1} does not hold, then there exists $i\notin T$ (i.e., $\bbE_{\rmP}(X_i)= \bbE_{\rmQ}(X_i)$) such that either $i\in I^*$ or $\an_{\cG}(i)\cap I^*\neq \varnothing$. If $i\in I^*$, then since $\bbE_{\rmP}(X_i)= \bbE_{\rmQ}(X_i)$ and Lemma \ref{lm:2}, $i$ must not be a source in $I^*$. Therefore we only need to discuss the case where $i\notin T$ and $\an_{\cG}(i)\cap I^*\neq \varnothing$.

Let $k$ be a source of $\an_{\cG}(i)\cap I^*$, then $k$ must also be a source of $I^*$. Otherwise there is a directed path from $k'$ to $k$ where $k'\neq k$ and $k'\in I^*$. By definition of ancestors, we know from $k\in \an_{\cG}(i)$ that there is also $k'\in \an_{\cG}(i)$. Therefore $k'\in \an_{\cG}(i)\cap I^*$, which violates $k$ being a source of $\an_{\cG}(i)\cap I^*$.

Since $k$ is a source of $I^*$, by Lemma \ref{lm:1} and \ref{lm:2}, we know that $\bbE_{\rmP}(X_k)\neq \bbE_{\rmQ}(X_k)$, i.e., $k\in T$. Notice that $k\in \an_{\cG}(i)$, and thus we must have a directed path from $k\in T$ to $i\notin T$. Thus, there exists some $i-j,j\in T,i\notin T$ such that $j\rightarrow i$.
\end{proof}

Using Lemma \ref{lm:3}, we know that we can check the authenticity of Assumption \ref{assumption:1} by looking at the orientation of edges between $T$ and $[p]\setminus T$, which is achievable by any (general) intervention on $X_{T}$ (or $X_{[p]\setminus T}$).
\begin{corollary}\label{cor:2}
Assumption \ref{assumption:1} holds if and only if the $\{T\}$-essential graph (or $\{[p]\setminus T\}$-essential graph)  of $\cG$ has edges $j\leftarrow i$ for all $i-j,j\in T, i\notin T$.
\end{corollary}

\begin{proof}
The proof follows as a direct application of the graphical characterization of interventional equivalence class in \rref{sec:id-1} and the results in Lemma \ref{lm:3}.
\end{proof}

\section{Details of Algorithms}\label{appendix:alg}
\subsection{Decomposition of Shift Interventional Essential Graphs}
\textbf{Chain Graph Decomposition:} \citetappendix{hauser2014two+} showed that every interventional essential graph is a chain graph with undirected connected chordal chain components, where the orientations in one component do not affect any other components. This decomposition also holds for shift interventional essential graphs, since every shift interventional essential graph is also an interventional essential graph (Corollary \ref{cor:1}). Below, we show an example of this decomposition (Figure \ref{afig2}).

\begin{figure}[h]
     \centering
     \begin{subfigure}[b]{0.24\textwidth}
         \centering \includegraphics[width=\textwidth]{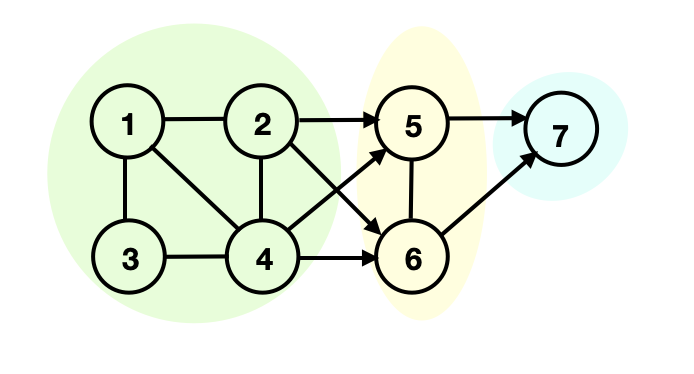}
         \caption{Essential Graph}
     \end{subfigure}
     \hfill
     \begin{subfigure}[b]{0.22\textwidth}
         \centering \includegraphics[width=0.65\textwidth]{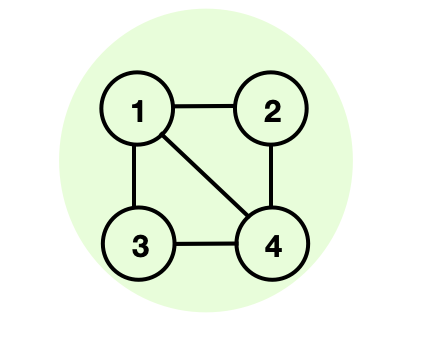}
         \caption{Chain Component 1}
     \end{subfigure}
    \hfill
     \begin{subfigure}[b]{0.22\textwidth}
         \centering \includegraphics[width=0.5\textwidth]{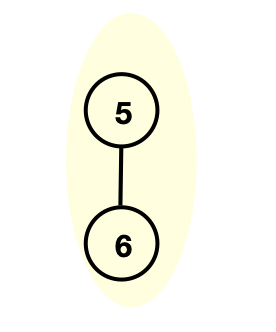}
         \caption{Chain Component 2}
     \end{subfigure}
     \hfill
     \begin{subfigure}[b]{0.22\textwidth}
         \centering \includegraphics[width=0.4\textwidth]{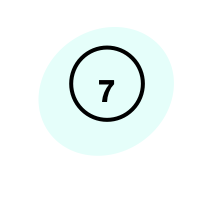}
         \caption{Chain Component 3}
     \end{subfigure}
        \caption{Chain graph decomposition of the essential graph in \textbf{(a)}.
        }
        \label{afig2}
\end{figure}

\begin{proof}[Proof of Lemma \ref{lm:4}]
Suppose an undirected connected chain component $\cC$ of the essential graph has two source nodes $i$ and $j$ w.r.t. $\cC$. Since $\cC$ is connected, there is a path between $i$ and $j$ in $\cC$; let $i-k_1-...-k_r-j$ be the shortest among all these paths. Because $i$ and $j$ are sources of $\cC$, there must be $i\rightarrow k_1$ and $k_r\leftarrow j$. Therefore, $\exists l\in \{1,...,r\}$ such that $k_{l-1}\rightarrow k_l\leftarrow k_{l+1}$ (let $k_0=i$ and $k_{r+1}=j$). By the shortest path definition, there is no edge between $k_{l-1}$ and $k_{l+1}$. Therefore there is a v-structure in $\cC$ induced by $k_{l-1}\rightarrow k_l\leftarrow k_{l+1}$. Since all DAGs in the same shift interventional equivalence class share the same v-structures, $k_{l-1}\rightarrow k_l\leftarrow k_{l+1}$ must be oriented in the essential graph. This violates $k_{l-1},k_l, k_{l+1}$ belonging to the same undirected chain component $\cC$. Thus, combining this with the fact that $\cC$ must have one source node, we obtain that $\cC$ has exactly one source node w.r.t. $\cC$.

Next we show that the source node of a chain component is also the source of $\cG$ if and only if there are no incoming edges to this component. Let $i$ be the source of the chain component $\cC$. On one hand, $i$ must be the source of $\cG$ if there is no incoming edges to $\cC$. On the other hand, if there is an incoming edge $j\rightarrow k$ for some $j\notin \cC$ and $k\in \cC$, then since the essential graph is closed under Meek R1 and R2 (Proposition \ref{prop:1}), we know that there must be an edge $j\rightarrow l$ for all neighbors $l$ of $k$. Following the same deduction and the fact that $\cC$ is connected, we obtain that $j\rightarrow l$ for all $l\in\cC$ (Figure \ref{afig3}). This means that $j\rightarrow i$ as well. Therefore $i$ cannot be a source of $\cG$.

\begin{figure}[ht]
    \centering
    \includegraphics[width=0.22\textwidth]{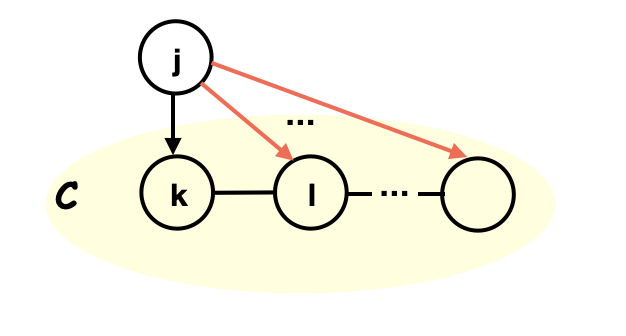}
    \caption{$j\rightarrow l$ for all $l\in\cC$.}
    \label{afig3}
    \vspace{-0.2in}
\end{figure}
\end{proof}

\subsection{NP-completeness of MinMaxC}
It was shown separately in \citetappendix{shen2012exact} and \citetappendix{lalou2018critical+} that the MinMaxC problem is NP-complete for general graphs and split graphs. Split graphs are a subclass of chordal graphs, where the vertices can be separated into a clique and an independent set (isolated nodes after removing the clique). 
Thus, MinMaxC is also NP-complete for chordal graphs.

\subsection{Clique Tree Strategy}
The clique tree strategy takes inputs of an undirected connected chordal graph $\cC$ and the sparsity constraint $S$, and outputs a shift intervention with no more than $S$ perturbation targets. If $\cC$ contains no more than $S$ nodes, then it returns any shift intervention with perturbation targets in $\cC$. If $\cC$ contains more than $S$ nodes, it first constructs a clique tree $\cT(\cC)$ of $\cC$ by the maximum-weight spanning tree algorithm \citepappendix{koller2009probabilistic+}. Then it iterates through the nodes in $\cT(\cC)$ (which are maximal cliques in $\cC$) to find a maximal clique $\mathcal{K}$ that breaks $\cT(\cC)$ into subtrees with sizes no more than half of the size of $\cT(\cC)$. If $\mathcal{K}$ has no more than $S$ nodes, then it returns any shift intervention with perturbation targets in $\mathcal{K}$. Otherwise, it samples $S$ nodes from $\mathcal{K}$ and returns any shift intervention with these $S$ nodes as perturbation targets. The following subroutine summarizes this procedure.

\begin{algorithm}[ht]
\SetAlgoLined
\KwIn{Chordal chain component $\cC$, sparsity constraint $S$.}
\eIf{$\cC$ has no more than $S$ nodes}{
 set $I$ as any shift intervention on $\cC$ with non-zero shift values\;
}{
let $C(\cC)$ be the maximal cliques of the chordal graph $\cC$\;
let $\cT(\cC)$ be a maximum-weight spanning tree of $\cC$ with $C(\cC)$ as nodes\;
set $\mathcal{K}=\varnothing$\;
\For{$K$ in $C(\cC)$}{
    get the subtrees of $\cT(\cC)$ after deleting node $C$\;
    \If{all subtrees has size $\leq \lceil (|C(\cC)|-1)/2\rceil$ 
    }{
    set $\mathcal{K}=K$\;
    break\;
    }
}
\If{$|\mathcal{K}|>S$}{set $\mathcal{K}$ as a random $S$-subset of $\mathcal{K}$\;}
set $I$ as any shift intervention on $\mathcal{K}$ with non-zero shift values\;
}
 \KwOut{Shift Intervention $I$}
 \caption{$\texttt{CliqueTree}(\cC,S)$}
 \label{alg:2}
\end{algorithm}

\textbf{Complexity:}
Let $N$ represent the number of nodes in $\cC$, i.e., $N=|\cC|$. All the maximal cliques of the chordal graph $\cC$ can be found in $O(N^2)$ time \citepappendix{galinier1995chordal}. We use Kruskal's algorithm for computing the maximum-weight spanning tree, which can be done in $O(N^2\log(N))$ \citepappendix{kruskal1956shortest}. The remaining procedure of iterating through $C(\cC)$ takes no more than $O(N^2)$ since chordal graphs with $N$ nodes have no more than $N$ maximal cliques \citepappendix{galinier1995chordal} and all subtree sizes can be obtained in $O(N)$. Therefore this subroutine can be computed in $O(N^2\log(N))$ time.

\subsection{Supermodular Strategy}
The supermodular procedure takes as input an undirected connected chordal graph $\cC$ as well as the sparsity constraint $S$, and outputs a shift intervention with perturbation targets by solving
\begin{equation}\label{aeq:5}
    \min_{A\subset V_{\cC}}\max_{i\in V_{\cC}} \hat{f}_i(A),\quad |A|\leq S,
\end{equation}
with the SATURATE algorithm \citepappendix{krause2008robust+}. Here $V_{\cC}$ represents nodes of $\cC$ and $\hat{f}_i(A) =\sum_{j\in V_{\cC}} \hat{g}_{i,j}(A)$ with $\hat{g}_{i,j}$ defined in \eqref{eq:4.3-2}. Algorithm \ref{alg:3} summarizes this subroutine.

\begin{algorithm}[ht]
\SetAlgoLined
\KwIn{Chordal chain component $\cC$, sparsity constraint $S$.}
\eIf{$\cC$ has no more than $S$ nodes}{
 set $I$ as any shift intervention on $\cC$ with non-zero shift values\;
}{
    let $A$ be the solution of \eqref{aeq:5} returned by SATURATE \citepappendix{krause2008robust+}\;
    set $I$ as any shift intervention on $A$ with non-zero shift values\; 
}
 \KwOut{Shift Intervention $I$}
 \caption{$\texttt{Supermodular}(\cC,S)$}
 \label{alg:3}
\end{algorithm}

\textbf{Supermodularity:}
First we give an example showing that $f_i$ defined in \eqref{eq:4.3-1} is not supermodular for chordal graphs, although it is clearly monotonic decreasing.
\begin{example}
Consider the chordal graph in Figure \ref{afig4}; we have $f_1(\{2\})-f_1(\varnothing)=3-4=-1 > -2 =1-3=f_1(\{2,3\})-f_1(\{3\})$. Therefore $f_1$ is not supermodular for this graph.
\begin{figure}[ht]
    \centering
    \includegraphics[width=0.12\textwidth]{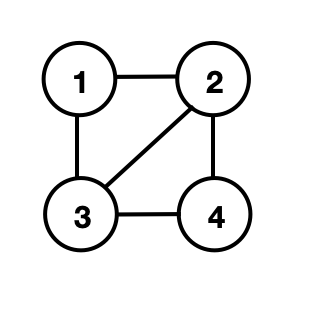}
    \caption{$f_i$ is not supermodular.}
    \label{afig4}
\end{figure}
\end{example}
Next we prove that $\hat{f}_i$ is supermodular and monotonic decreasing. 
\begin{proof}
Since $\hat{f}_i(A)=\sum_{j\in V_{\cC}} \hat{g}_{i,j}(A)$, we only need to show that every $\hat{g}_{i,j}$ is supermodular and monotonic decreasing. In the following, we refer to a path without cycles as a \textit{simple path}.

For any $A\subset B\subset V_{\cC}$, since $V_{\cC}-B$ is a subgraph of $V_{\cC}-A$, then any simple path between $i$ and $j$ in $V_{\cC}-B$ must also be in $V_{\cC}-A$. Hence $m_{i,j}(V_{\cC}-B)\leq m_{i,j}(V_{\cC}-A)$, which means that
\[
\hat{g}_{i,j}(A)\geq \hat{g}_{i,j}(B),
\]
i.e.,  $\hat{g}_{i,j}$ is monotonic decreasing.

For any $x\in V_{\cC}\setminus B$, the difference  $m_{i,j}(V_{\cC}-B)-m_{i,j}(V_{\cC}-B\cup\{x\})$ is the number of simple paths in $V_{\cC}-B$ between $i$ and $j$ that pass through $x$. Each of these paths must also be in $V_{\cC}-A$, since $V_{\cC}-B$ is a subgraph of $V_{\cC}-A$. Therefore, 
\[
m_{i,j}(V_{\cC}-B)-m_{i,j}(V_{\cC}-B\cup\{x\}) \leq m_{i,j}(V_{\cC}-A)-m_{i,j}(V_{\cC}-A\cup\{x\}),
\]
which means that
\[
\hat{g}_{i,j}(A\cup\{x\})-\hat{g}_{i,j}(A) \leq \hat{g}_{i,j}(B\cup\{x\})-\hat{g}_{i,j}(B),
\]
i.e., $\hat{g}_{i,j}$ is supermodular.
\end{proof}

\textbf{SATURATE algorithm \citepappendix{krause2008robust+}:} Having shown that $\hat{f}_i$ is monotonic supermodular, we solve the robust supermodular optimization problem in \eqref{aeq:5} with the SATURATE algorithm in \citepappendix{krause2008robust+}. SATURATE performs a binary search for potential objective values and uses a greedy partial cover algorithm to check the feasibility of these objective values; for a detailed description of the algorithm, see \citetappendix{krause2008robust+}.

\textbf{Complexity:} Let $N$ represent the number of nodes in $\cC$, i.e., $N=|\cC|$. SATURATE uses at most $O(N^2S\log(N))$ evaluations of supermodular functions $\hat{f}_i$ \citepappendix{krause2008robust+}. Each $\hat{f}_i$ computes all the simple paths between $i$ and all other $j$ in $\cC$. A modified depth-first search is used to calculated these paths \citepappendix{sedgewick2001algorithms}, which results in $\mathcal{F}(N)$ complexity. For general graphs, this problem is \verb|#|P-complete \citepappendix{valiant1979complexity}. However, this might be significantly reduced for chordal graphs. We are unaware of particular complexity results for chordal graphs, which would be an interesting direction for future work. The total runtime of this subroutine is thus bounded by $O(N^2\mathcal{F}(N)S\log(N))$.\footnote{For a more efficient implementation, one could replace the undirected graph with a DAG in its MEC (which can be found in linear time using L-BFS). All statements hold except that $\hat{f}_i$ is no longer necessarily tight for tree graphs. This replacement results in a total complexity of $O(N^4 S\log(N))$ for the subroutine, since directed simple paths can be counted in $O(N^2)$.} 

\subsection{Violation of Faithfulness}
From Corollary \ref{cor:2}, we know that we can check whether Assumption \ref{assumption:1} holds or not by any intervention on $X_{T}$ (or $X_{[p]\setminus T}$). However, we can run Algorithm \ref{alg:1} to obtain $I^*$ without Assumption \ref{assumption:1} because lines 2-14 in Algorithm \ref{alg:1} always return the correct $I^*$. 

Let $I\subset I^*$ be the
resolved part of $I^*$ in line 2, i.e., it is a shift intervention constructed by taking a subset of perturbation targets of $I^*$ and their corresponding shift values. Let $I^*-I$ be the remaining shift intervention constructed by deleting $I$ in $I^*$. Denote $T_{I}=\{i|i\in [p], \bbE_{\rmP^{I}}(X_i)\neq \bbE_{\rmQ}(X_i)\}$, which is returned by line 3. If $T_I\neq \varnothing$, then we have solved $I^*$. Otherwise we have:

\begin{lemma}
The source nodes w.r.t. $T_I$ must be perturbation targets of $I^*-I$ and their corresponding shift values are $\bbE_{\rmQ}(X_i)-\bbE_{\rmP^{I}}(X_i)$ (for source node $i$). 
\end{lemma}

\begin{proof}
Let $i$ be a source node w.r.t. $I^*-I$ and $a_i$ be its corresponding shift value. Since intervening on other nodes in $I^*-I$ does not change the marginal distribution of $i$, we must have that $a_i =\bbE_{(\rmP^{I})^{I^*-I}}(X_i)-\bbE_{\rmP^{I}}(X_i)$. And because $ (\rmP^{I})^{I^*-I}=\rmP^{I^*} = \rmQ$, we know that
\[
a_i = \bbE_{\rmQ}(X_i)-\bbE_{\rmP^{I}}(X_i).
\]
From this, we also have that $\bbE_{\rmP^{I}}(X_i)\neq \bbE_{\rmQ}(X_i)$ since $a_i\neq 0$. Therefore, all source nodes $i$ w.r.t. $I^*-I$ are in $T_I$ and their corresponding shift values are $\bbE_{\rmQ}(X_i)-\bbE_{\rmP^{I}}(X_i)$. 

Let $i$ be a source w.r.t. $T_I$, then $i$ must also be a source node w.r.t. $I^*-I$. Since $\bbE_{\rmP^{I}}(X_i)\neq \bbE_{\rmQ}(X_i)$, $i$ must be a source node in $I^*-I$ or has a source node in $I^*-I$ as its ancestor. If it is the latter case, then since all source nodes in $I^*-I$ must be in $T_I$, $i$ cannot be a source node w.r.t. $T_I$, a contradiction. Therefore the source w.r.t. $T_I$ must also be the source w.r.t. $I^*-I$. Combined with the result in the previous paragraph, we have that all source nodes $i$ w.r.t. $T_I$ are perturbation targets of $I^*-I$ and their corresponding shift values are $\bbE_{\rmQ}(X_i)-\bbE_{\rmP^{I}}(X_i)$. 
\end{proof}

This lemma shows that $U_T$ obtained in lines 5-11 of Algorithm \ref{alg:1} must be the perturbation targets of $I^*-I$ and line 12 gives the correct shift values. Therefore Algorithm \ref{alg:1} must return the correct $I^*$. 

However, to be able to obtain the shift-$\cI$-EG of $\cG$, we need mean interventional faithfulness to be satisfied by $I\in\cI$ (replacing $I^*$ with $I$ and $\rmQ$ with $\rmP^{I}$ in Assumption \ref{assumption:1}) as well as $\cI$-faithfulness \citepappendix{squires2020permutation+} to be satisfied by  $(\rmP,\{\rmP^I\}_{I\in\cI})$ with respect to $\cG$.

\section{Proof of Worst-case Bounds}\label{appendix:bound}

\subsection{Proof of Lemma \ref{lemma:oracle-lb}}
To show Lemma \ref{lemma:oracle-lb}, we need the following proposition, which states that we can orient any maximal clique of a chordal graph to be most-upstream without creating cycles and v-structures, and the orientation in this clique can be made arbitrary. It was pointed out in \citepappendix{vandenberghe2015chordal} using similar arguments that any clique of a chordal graph can be most-upstream. Here, we provide the complete proof. 

\begin{proposition}\label{prop:2}
Let $\cD$ be any undirected chordal graph and $K$ be any maximal clique of $\cD$, for any permutation $\pi_K$ of the nodes in $K$, there exists a topological order $\pi$ of the nodes in $\cD$ such that $\pi$ starts with $\pi_K$ and orienting $\cD$ according to $\pi$ does not create any v-structures.
\end{proposition}

\begin{proof}
A topological order $\pi$ of a chordal graph $\cD$, orienting according to which does not create v-structures, corresponds to the reverse of a \textit{perfect elimination order} \citepappendix{hauser2014two+}. A perfect elimination order is an order of nodes in $\cD$, such that all neighbors of $i$ in $\cD$ that appear after $i$ in this order must constitute a clique in $\cD$. Any chordal graph has at least one perfect elimination order \citepappendix{andersson1997characterization+}. In the following, we will use the reverse of a perfect elimination order to refer to a topological order that does not create v-structures.

To prove Proposition \ref{prop:2}, we first prove the following \textit{statement}: if $K\neq \cD$, then there exists a perfect elimination order of nodes in $\cD$ that starts with a node not in $K$. To show this, by Proposition 6 in \citetappendix{hauser2014two+}, we only need to prove that if $K\neq \cD$, then there is a node not in $K$, whose neighbors in $\cD$ constitute a clique. 

We use induction on the number of nodes in $\cD$: Consider $|\cD|=1$. Since $K$ is a maximal clique, $K=\cD$. This statement holds trivially. Suppose the statement is true for chordal graphs with size $n-1$. Consider $|\cD|=n$. Since $\cD$ is a chordal graph, it must have a perfect elimination order. If this perfect elimination order starts with $i\in K$, then there is no edge between $i$ and any node $j\notin K$. Otherwise, since it is a perfect elimination order starting with $i$ and $K\ni i$ is a clique, there must be edges $j-k$ for all $k\in K$. This is a contradiction to $K$ being a maximal clique. 
\begin{figure}[h]
    \centering
    \includegraphics[width=0.2\textwidth]{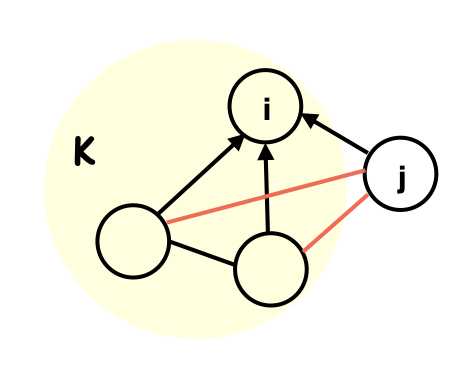}
\end{figure}

Consider the chordal graph $\cD'$ by deleting $i$ from $\cD$, $|\cD'|=n-1$. Let $K'$ be the maximal clique in $\cD'$ containing $K\setminus\{i\}$. If $K'=\cD'$, let $j$ be any node in $\cD\setminus K$. Since there is no edge $j-i$, and $\cD'\ni j$ is a clique. $j$'s neighbors in $\cD$ must also constitute a clique. If $K'\neq \cD'$, then by induction, we know that there exists $j\in \cD'\setminus K'$ such that $j$'s neighbors in $\cD'$ constitute a clique. Since there is no edge $j-i$, $j$'s neighbors in $\cD$ must also constitute a clique. Thus the statement holds for chordal graphs of size $n$. Therefore the statement holds.
\begin{figure}[h]
    \centering
     \begin{subfigure}[b]{0.3\textwidth}
         \centering
         \includegraphics[width=0.8\textwidth]{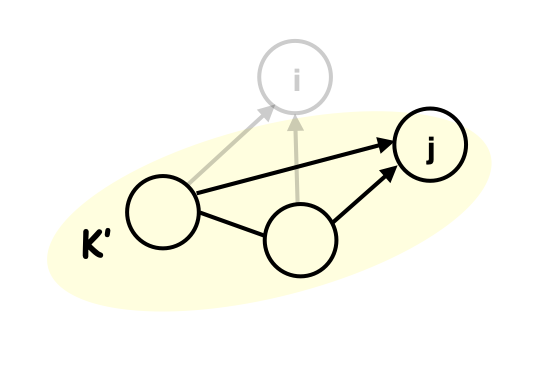}
     \end{subfigure}
     \begin{subfigure}[b]{0.3\textwidth}
         \centering
         \includegraphics[width=0.7\textwidth]{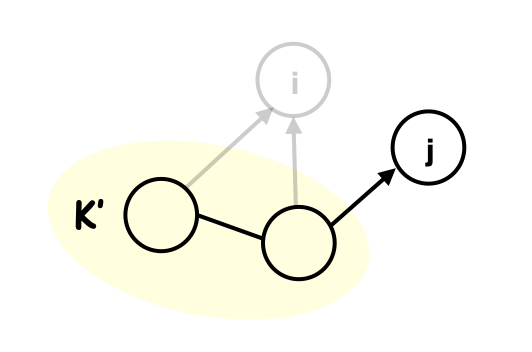}
     \end{subfigure}
\end{figure}

Now, we prove Proposition \ref{prop:2} by induction on the number of nodes in $\cD$: Consider $|\cD|=1$. Since $K$ is a maximal clique, $K=\cD$. Thus Proposition \ref{prop:2} holds trivially.

Suppose Proposition \ref{prop:2} holds for chordal graphs of size $n-1$. Consider $|\cD|=n$. If $K=\cD$, then Proposition \ref{prop:2} holds. If $K\neq\cD$, then by the above statement, there exists $j\in \cD\setminus K$, such that there exists a perfect elimination order of $\cD$ starting with $j$. Let $\cD'$ be the chordal graph obtained by deleting $j$ from $\cD$. By induction, there exists $\pi'$, a reverse of perfect elimination order, that starts with $\pi_K$. Let $\pi=(\pi',j)$; we must have that the reverse of $\pi$ is a perfect elimination order, since all neighbors of $j$ in $\cD$ constitute a clique. Therefore $\pi$ gives the wanted topological order and Proposition \ref{prop:2} holds for chordal graphs of size $n$. This completes the proof of Proposition \ref{prop:2}.
\end{proof}

\begin{proof}[Proof of Lemma \ref{lemma:oracle-lb}]
Given any algorithm $\mathcal{A}$, let $S_1,...,S_k$ be the first $k$ shift interventions given by $\mathcal{A}$. 
By Proposition \ref{prop:2}, we know that there exists a feasible orientation of $\cC$ such that the largest maximal clique $K$ of $\cC$ is most-upstream and that, for $k'=1,...,k$, $S_{k'}\cap K$ is most-downstream of $K-\cup_{l<k'} S_{l}$.
For example, in the figure below, suppose algorithm $\cA$ chooses $S_1=\{3\}$ based on (a) and $S_2=\{2\}$ based on (b).
There is a feasible orientation in (d) such that the largest clique $K=\{1,2,3\}$ is most-upstream and $S_{k'}\cap K$ is most-downstream of $K$, for $k'=1,2$.
\begin{figure}[h]
     \centering
     \begin{subfigure}[b]{0.22\textwidth}
         \centering \includegraphics[width=0.65\textwidth]{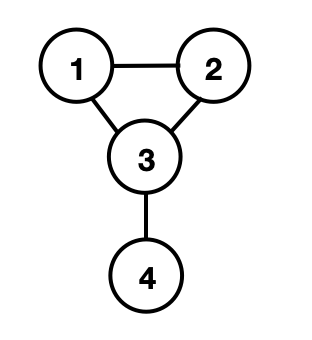} \caption{$\cC$}
     \end{subfigure}
     \hfill
     \begin{subfigure}[b]{0.22\textwidth}
         \centering \includegraphics[width=0.65\textwidth]{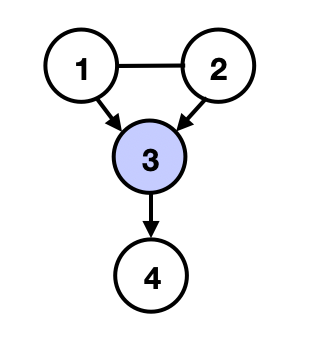}
         \caption{$\cA$ chooses $S_1=\{3\}$}
     \end{subfigure}
    \hfill
     \begin{subfigure}[b]{0.22\textwidth}
         \centering \includegraphics[width=0.65\textwidth]{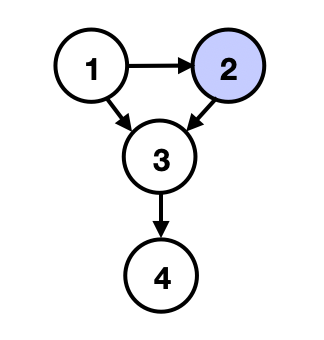}
         \caption{$\cA$ chooses $S_2=\{2\}$}
     \end{subfigure}
     \hfill
     \begin{subfigure}[b]{0.22\textwidth}
         \centering \includegraphics[width=0.65\textwidth]{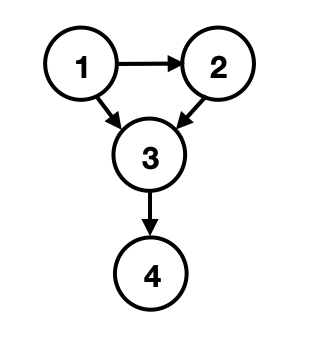} \caption{Orientation of $\cC$}
     \end{subfigure}
    \label{afigE1}
\end{figure}

Since $|S_{k'}|\leq S$ and $|K|=m_{\cC}$, in this worst case, it needs at least $\lceil \frac{m_{\cC}-1}{S}\rceil$ interventions to identify the source of $K$, i.e., the source of $\cC$ (minus $1$ because in this case, if there is only one node left, then it must be the source).
\end{proof}

\subsection{Proof of Lemma \ref{lemma:alg-lb}}

Let $K$ be the clique obtained by lines 7-13 in Algorithm \ref{alg:2}; when $\cC$ has more than $S$ nodes, we refer to $K$ as the \textit{central clique}. 
To prove Lemma \ref{lemma:alg-lb}, we need the following proposition.
This proposition shows that by looking at the undirected graph $\cC$, we can find a node in the central clique $K$ satisfying certain properties, which will become useful in the proof of Lemma \ref{lemma:alg-lb}.

\begin{proposition}\label{prop:3}
Let $\{\cT_{a}\}_{a\in A}$ be the connected subtrees of $\cT(\cC)$ after removing $K$.
For a node $k \in K$, let $A_k \subset A$ be the set of indices $a \in A$ such that the tree $\cT_a$ is connected to $K$ only through the node $k$. 
Let $\cT_{A_k}=\{\cT_{a}\}_{a\in A_k}$ be the collection of all such subtrees.
If there exists $a \in A \setminus A_k$ such that there is an edge between $\cT_a$ and $k$, let $\cT^*_{k}$ be the one with the largest number of maximal cliques; otherwise let $\cT^*_{k}=\varnothing$. 
Then there exists a node $k$ such that the number of maximal cliques in the subgraph induced by the subtrees $\cT_{A_k} \cup \{\cT^*_{k}\}$ and $k$ itself does not exceed $\lceil \frac{r-1}{2}\rceil$.
\end{proposition}

\begin{example}
As an example, the following figure shows the subtrees that are connected to $K$ only through node $1$, indexed by $A_1$ (blue). The largest subtree in $A\setminus A_1$ that is adjacent to node $1$ is denoted by $\cT_1^*$ (undimmed in green). 
\end{example}

\begin{figure}[h]
    \centering \includegraphics[width=0.36\textwidth]{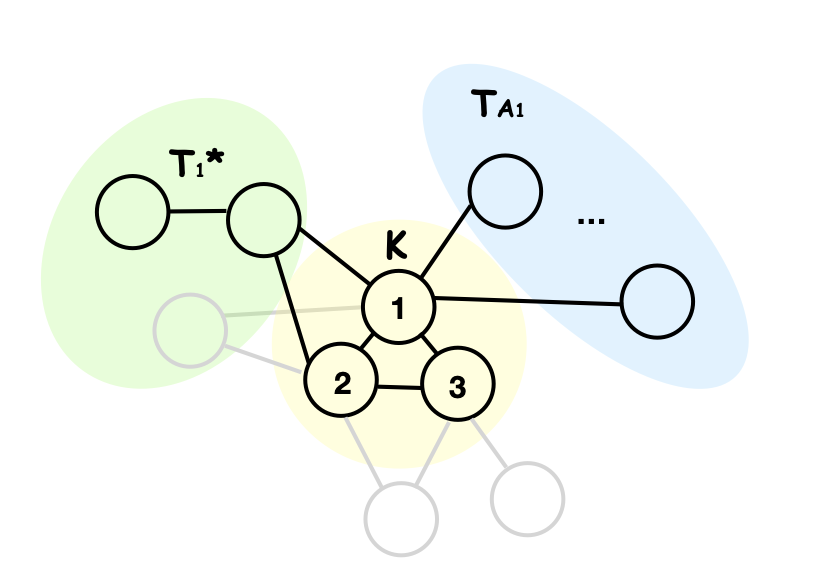}
    \caption{An example of $\cT_{A_k}$ and $\cT_{k}^*$ for $k=1$.}
    \label{afigE2}
\end{figure}

\begin{proof}[Proof of Proposition \ref{prop:3}]
Notice the following facts.

\textbf{Fact 1:} Let $\cT$ be any subtree in $\{\cT_{a}\}_{a\in A}$; then there must exist a node $i\in K$ such that there is no edge between $i$ and $\cT$.

\begin{quote}
\textit{Proof of Fact 1:} For any two nodes $i,i'\in K$, because $\cC$ is chordal and $\cT$ is connected, either the neighbors of $i$ in $\cT$ subset that of $i'$, or the the neighbors of $i'$ in $\cT$ subset that of $i$. Therefore we can order all nodes $K$, where all neighbors of $i$ in $\cT$ subset that of $i'$ that appear after $i$. Then if the first node in this order has some neighbor $t\in \cT$, all nodes in $K$ have $t$ as neighbor, contradicting $K$ being a maximal clique.
\end{quote}
\begin{figure}[ht]
     \centering
     \begin{subfigure}[b]{0.35\textwidth}
         \centering
         \includegraphics[width=0.45\textwidth]{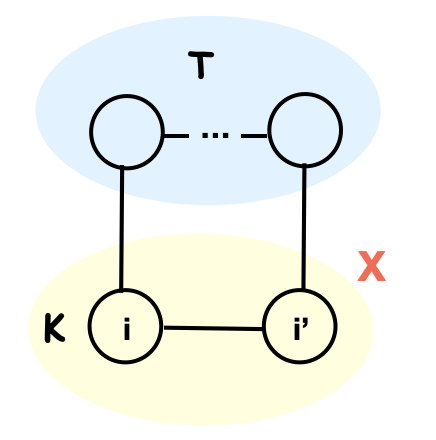}
         \caption{Contradicting chordal $\cC$}
     \end{subfigure}
     \begin{subfigure}[b]{0.35\textwidth}
         \centering
         \includegraphics[width=0.45\textwidth]{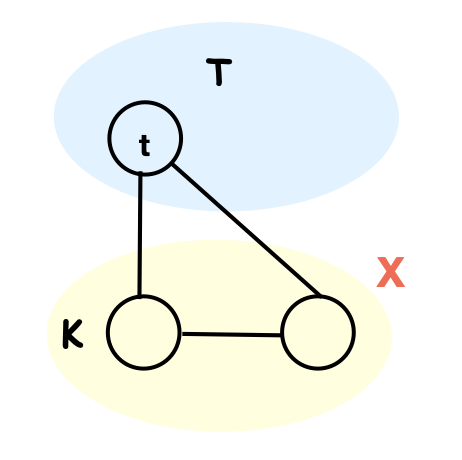}
         \caption{Contradicting maximal clique $K$}
     \end{subfigure}
\end{figure}

\textbf{Fact 2:} Let $\bar{\cT}$ be the collection of the subtrees where all edges connecting to $K$ are through a single node $k \in K$. We have that $\bar{\cT}$ is the union of disjoint sets $\cT_{A_k}, k\in K$. 

\begin{quote}
\textit{Proof of Fact 2:} This follows directly from the definition of $A_k$.
\end{quote}

\textbf{Fact 3:} Let $\cT^*$ be the collection of non-empty $\cT_{k}^*, k \in K$. Then $\cT^* \cap \bar{\cT}=\varnothing$. 
Furthermore, for any subtree in $\cT^*$, there is a node $i\in K$ such that there is no edge between $i$ and this subtree.

\begin{quote}
\textit{Proof of Fact 3:} This follows directly from the definition of $\cT_{k}^*$ and Fact 1.
\end{quote}

Now we prove Proposition \ref{prop:3}. If $\cT^*=\varnothing$, then since $K$ contains at least two nodes (otherwise $A=\varnothing$ and the proposition holds trivially) and the number of maximal cliques in $\bar{\cT}$ does not exceed $r-1$, using Fact 2, we have at least one $k\in K$ such that the number of maximal cliques in the subgraph induced by $\cT_{A_k}\cup\{\cT_{k}^*\}=\cT_{A_k}$ and $k$ itself does not exceed $\lceil \frac{r-1}{2}\rceil$.

If $\cT^* \neq\varnothing$. 
Let $\cT^*_{k'}$ be the subtree with the largest number of maximal cliques in $\cT^*$.
Let $k\in K$ be the node such that there is no edge between $k$ and the subtree $\cT^*_{k'}$ ($k$ exists because of Fact 3). 
Now suppose that the proposition does not hold. Then the number of maximal cliques in the subgraph induced by $\cT_{A_k}\cup\{\cT_{k}^*\}$ and $k$ itself must exceed $\lceil \frac{r-1}{2}\rceil$.
Also, the number of maximal cliques in the subgraph induced by $\cT_{A_{k'}}\cup\{\cT^*_{k'}\}$ and $k'$ itself exceeds $\lceil \frac{r-1}{2}\rceil$.
Notice that $(\cT_{A_k}\cup\{\cT_{k}^*\})\cap(\cT_{A_{k'}}\cup\{\cT_{k'}^*\})=\varnothing$. Therefore $\cT_{A_k}\cup\{\cT_{k}\}$ is connected to $\cT_{A_{k'}}\cup\{\cT^*_{k'}\}$ only through $K$. 
Hence the sum of numbers of maximal cliques in $\cT_{A_k}\cup\{\cT^*_{k}\}\cup\{\{k\}\}$ and $\cT_{A_{k'}}\cup\{\cT^*_{k'}\}\cup\{\{k'\}\}$ does not exceed $r$. 
We cannot have both $\cT_{A_k}\cup\{\cT^*_{k}\}\cup\{\{k\}\}$ and $\cT_{A_{k'}}\cup\{\cT^*_{k'}\}\cup\{\{k'\}\}$ having more than $\lceil \frac{r-1}{2}\rceil$ maximal cliques. 
Therefore the proposition must hold.
\end{proof}

\begin{proof}[Proof of Lemma \ref{lemma:alg-lb}]
For \texttt{CliqueTree}, we prove this lemma for a ``less-adaptive'' version for the sake of clearer discussions. In this ``less-adaptive'' version, instead of output $1$ intervention with $S$ perturbation targets sampled from the central clique $K$ (when it has more than $S$ nodes) in Algorithm \ref{alg:2}, we directly output $\lceil \frac{|K|-1}{S}\rceil$ interventions with non-overlapping perturbation targets in $K$. Each of these interventions has no more than $S$ perturbation targets and they contain at least $|K|-1$ nodes in $K$ altogether. Furthermore, we pick these interventions such that if they contain exactly $|K|-1$ nodes, then the remaining node satisfies Proposition \ref{prop:3}.

After these $\lceil \frac{|K|-1}{S}\rceil$ interventions, we obtain a partially directed $\cC$, which is a chain graph, with one of its chain components without incoming edges as input to \texttt{CliqueTree} in the next iteration of the inner-loop in Algorithm \ref{alg:1}. Denote this chain component as $\cC'$. We show that $\cC'$ has no more than $\left\lceil \frac{r-1}{2} \right\rceil$ maximal cliques each with no more than $m_{\cC}$ nodes. If $\lceil \frac{r-1}{2}\rceil=0$, then $r=1$ and this trivially holds since the source of $\cC$ must be identified. In the following, we assume $\lceil \frac{r-1}{2}\rceil>0$.

\textbf{Size of maximal cliques:} The maximal clique in $\cC'$ must belong to a maximal clique in $\cC$, and thus has no more than $m_{\cC}$ nodes. 

\textbf{Number of maximal cliques:} If the source node is identified, then $\cC'$ only has one node. This trivially holds. Now consider when the source node is not identified. We proceed in two cases. 

Case I: if these $\lceil \frac{|K|-1}{S}\rceil$ interventions contain all nodes in $K$, then they break the clique tree $\cT(\cC)$ into subtrees each with no more than $\lceil \frac{r-1}{2}\rceil$ maximal cliques. $\cC'$ must belong to one of these subtrees. Therefore it must have no more than $\lceil \frac{r-1}{2}\rceil$ maximal cliques.

Case II: if these $\lceil \frac{|K|-1}{S}\rceil$ interventions do not contain all nodes in $K$, then there is exactly one node left in $K$ that is not a perturbation target, which satisfies Proposition \ref{prop:3}. Denote this node as $k$ and the source node w.r.t. the intervened $|K|-1$ nodes as $i$. From Theorem \ref{thm:1}, we have that $i$ is identified and $\forall j\in K, j\neq k$, the orientation of edge $k-j$ is identified. 

If $i\rightarrow k$, then $i$ is the source w.r.t. $K$: if $i$ is the source w.r.t. $\cC$, then $\cC'=\{i\}$ has no more than $\lceil \frac{r-1}{2}\rceil$ maximal cliques; otherwise, there is a unique subtree of $\cT(\cC)$ after removing $K$ that has an edge pointing to $i$ in $\cC$ (it exists because $i$ is the source of $K$ but not the source of $\cC$; it is unique because there is no edge between subtrees and there is no v-structure at $i$), and therefore $\cC'$ must belong to this subtree which has no more than $\lceil \frac{r-1}{2}\rceil$ maximal cliques.

\begin{figure}[h]
     \centering
     \begin{subfigure}[b]{0.35\textwidth}
         \centering
         \includegraphics[width=0.7\textwidth]{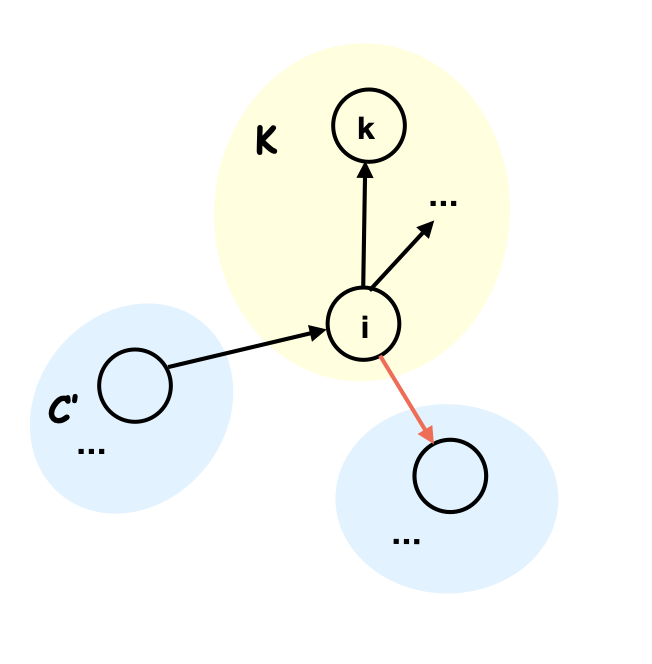}
         \caption{If $i\rightarrow k$}
     \end{subfigure}
     \begin{subfigure}[b]{0.35\textwidth}
         \centering
         \includegraphics[width=0.7\textwidth]{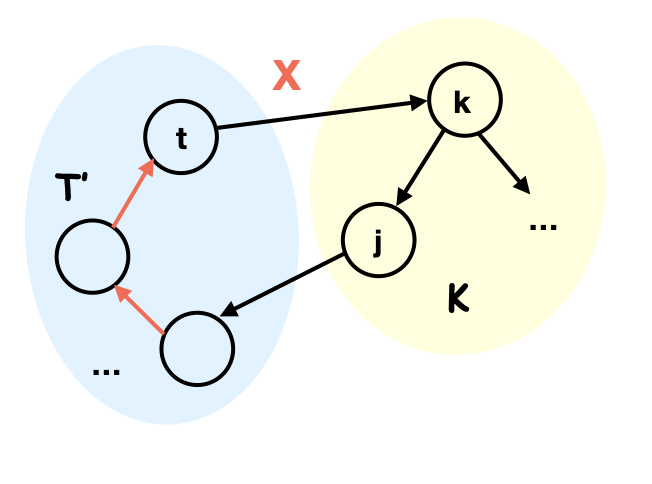}
         \caption{If $i\leftarrow k$, Fact 1}
     \end{subfigure}
\end{figure}

If $i\leftarrow k$, then $k$ is the source w.r.t. $K$: consider all the subtrees of $\cT(\cC)$ after removing $K$. We have the following two facts:

\textbf{Fact 1:} Let $\cT'$ be a subtree such that there is an edge between $\cT'$ and $K-\{k\}$ and all these edges are pointing towards $\cT'$. Then all edges between $k$ and $t\in \cT'$ must be oriented as $k\rightarrow t$. Thus $\cC'\cap \cT'=\varnothing$.

\begin{quote}
\textit{Proof of Fact 1:} Otherwise, suppose $t\in \cT'$ and $t\rightarrow k$. Let $j\in K-\{k\}$ such that there is an edge between $j$ and $\cT'$. Since $\cT'$ is connected, there must be a path from $j$ to $t$ in $\cT'$. Let $j=t_0-t_1-...-t_l-t_{l+1}=t$ be the shortest of these path. Since $t_0-t_1-...-t_l-t_{l+1}$ is shortest, there cannot be an edge between $t_{l'}$ and $t_{l''}$ with $l''-l'>1$. And since all edges between $\cT'$ and $K-\{k\}$ are pointing towards $\cT'$, there is an edge $j=t_0\rightarrow t_1$. Therefore to avoid v-structures, it must be $j=t_0\rightarrow t_1\rightarrow...\rightarrow t_l\rightarrow t_{l+1}=t$. This creates a directed cycle $k\rightarrow j \rightarrow ... \rightarrow t \rightarrow k$, a contradiction.
\end{quote}

\textbf{Fact 2:} There can be at most one subtree $\cT'$ such that there is an edge pointing from $\cT'$ to $K-\{k\}$ and also some $t\in \cT'$ such that $t\rightarrow k$ or $t-k$ is unidentified. Therefore at most one subtree $\cT'$ of this type can have $\cC'\cap\cT'\neq \varnothing$. 

\begin{quote}
\textit{Proof of Fact 2:} Otherwise suppose there are two different subtrees $\cT'_1,\cT'_2$ such that $K-\{k\}\ni j_1\leftarrow t_1\in \cT'_1, K-\{k\}\ni j_2\leftarrow t_2\in \cT'_2$. Since there is no edge $t_1-t_2$, we have $j_1\neq j_2$. Without loss of generality, suppose $j_1\rightarrow j_2$. Let $t$ be any node in $\cT_2'$ with an edge $t-k$, since $\cT_2'$ is connected, let $t=t_0'-t_1'-...-t_l'-t_{l+1}'=t_2$ be the shortest path between $t$ and $t_2$ in $\cT_2'$. Let $l'$ be the maximum in $0,1,...,l$ such that $t'_{l'}\leftarrow t'_{l'+1}$. If such $l'$ does not exist, then $t=t_0'\rightarrow t_1'\rightarrow...\rightarrow t_{l+1}'=t_2$. Since $j_1\rightarrow j_2$ and there is no v-structure at $j_2$, there must be an identified edge $j_1-t_{l+1}'=t_2$. Notice that there is no edge between $t_2$ and $t_1$ and $t_1\rightarrow j_1$, to avoid v-structure, it must be $j_1\rightarrow t_2$. The same deduction leads to identified edges $j_1\rightarrow t_{0}'=t$. Since $k\rightarrow j_1$ and there are no cycles, the edge $k\rightarrow t$ must be identified. If $l'$ exists, since $t=t_0'-t_1'-...-t_l'-t_{l+1}'=t_2$ is the shortest path and there is no v-structure, we must have $t=t_0'\leftarrow ... \leftarrow t_{l'+1}'$. Furthermore, since $l'$ is the largest, $t_{l'+1}'\rightarrow ... \rightarrow t_{l+1}' =t_2$. By a similar deduction as in the case where $l'$ does not exist, we must have an identified edge $j_1\rightarrow t_{l'+1}'$. Therefore $k\rightarrow j_1\rightarrow t_{l'+1}' \rightarrow ... t_0'=t$. To avoid directed cycles, $k\rightarrow t$ must be identified. Therefore all edges between $k$ and $\cT_2'$ are identified as pointing to $\cT_2'$.
\end{quote}

\begin{figure}[ht]
     \centering
     \begin{subfigure}[b]{0.40\textwidth}
         \centering
         \includegraphics[width=0.9\textwidth]{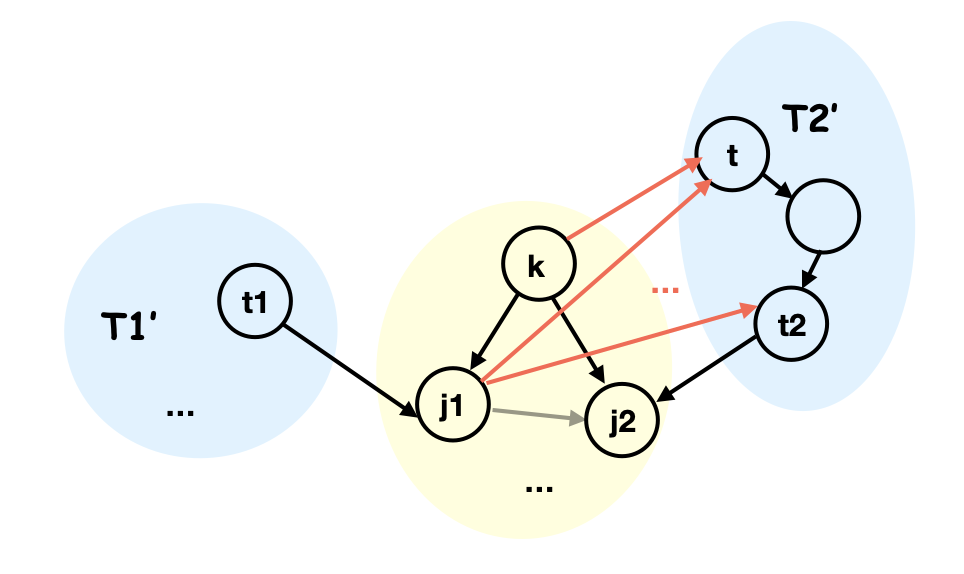}
         \caption{$l'$ does not exist}
     \end{subfigure}
     \begin{subfigure}[b]{0.40\textwidth}
         \centering
         \includegraphics[width=0.9\textwidth]{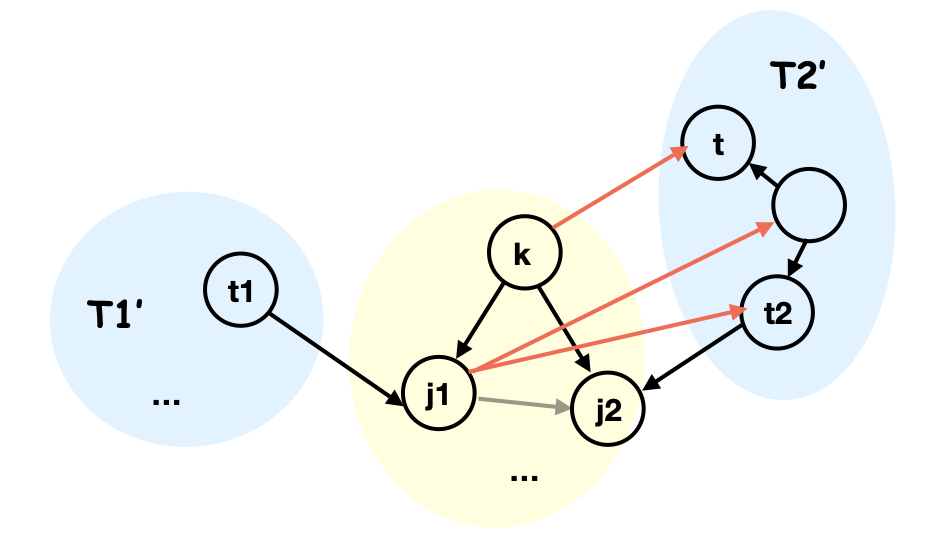}
         \caption{$l'$ exists}
     \end{subfigure}
\end{figure}

Using the above two facts, let $\cT'$ be the unique subtree in Fact 2 (if it exists); if there is no edge between $\cT'$ and $k$, then $\cC'$ must be in the subgraph induced by $k$ itself and $\cT_{A_k}$ in Proposition \ref{prop:3}, which has no more than $\lceil \frac{r-1}{2}\rceil$ maximal cliques. If there is an edge between $\cT'$ and $k$, we know that $\cC'$ must be in the joint set of $k$, $\cT'$ and $\cT_{A_k}$. Since the number of maximal cliques in $\cT'$ must be no more than that of $\cT_{k}^*$ in Proposition \ref{prop:3}, we know that $\cC'$ has no more than $\lceil \frac{r-1}{2}\rceil$ maximal cliques.

Therefore, after $\lceil \frac{|K|-1}{S}\rceil\leq \lceil \frac{m_{\cC}-1}{S}\rceil$ interventions, we reduce the number of maximal cliques to at most $\lceil \frac{r-1}{2}\rceil$ while maintaining the size of the largest maximal clique $\leq m_{\cC}$. Using this iteratively, we obtain that \texttt{CliqueTree} identifies the source node with at most $\lceil \log_2(r_\cC+1)\rceil\cdot \lceil \frac{m_{\cC}-1}{S}\rceil$ interventions.

For \texttt{Supermodular}, we do not discuss the gap between $\hat{g}_{i,j}$ and $g_{i,j}$ and how well SATURATE solves \eqref{aeq:5}. In this case, it is always no worse than the \texttt{CliqueTree} in the worst case over the feasible orientations of $\cC$, since it solves MinMaxC optimally without constraining to maximal cliques. Therefore, it also takes no more than $\lceil \log_2(r_\cC+1)\rceil\cdot\lceil\frac{m_\cC-1}{S}\rceil$ to identify the source node.
\end{proof}

\subsection{Proof of Theorem \ref{thm:2}}

\begin{proof}[Proof of Theorem \ref{thm:2}]
This result follows from Lemma \ref{lemma:oracle-lb} and \ref{lemma:alg-lb}. Divide $I^*$ into $I_1,...,I_k$ such that $I_{k'}$ is the source node of $I^*-\cup_{l<k'}I_{l}$. Since shifting $I_{k'}$ affects the marginal of subsequent $I_{k''}$ with $k''>k'$, any algorithm needs to identify $I_1,...,I_k$ sequentially in order to identify the exact shift values. 

Suppose $I_1,...,I_{k'-1}$ are learned. For $I_{k'}$, consider the chain components of the subgraph of the shift-$\{\cup_{l< k'} I_l\}$-EG induced by $T=\{i|i\in [p], \bbE_{(\rmP^{\cup_{l<k'} I_l})}(X_i)\neq \bbE_{\rmQ}(X_i)\}$ with no incoming edge. Applying Lemma \ref{lm:4} for $\cI=\{\cup_{l< k'} I_l$\} and Observation \ref{obs:1} for this subgraph and $I_{k'}$, we deduce that there are exactly $|I_{k'}|$ such chain components and $I_{k'}$ has exactly one member in each of these chain components. Let $m_{k',1},...,m_{k',|I_{k'}|}$ be the sizes of the largest maximal cliques in these $|I_{k'}|$ chain components. By Lemma \ref{lemma:oracle-lb}, we know that any algorithm needs at least $\sum_{i=1}^{|I_{k'}|} \lceil \frac{m_{k',i}-1}{S}\rceil$ number of interventions to identify $I_{k'}$ in the worst case. However, since all these chain components contain no more than $r$ maximal cliques, by Lemma \ref{lemma:alg-lb}, we know that our strategies need at most $\lceil \log_2(r+1)\rceil\cdot\sum_{i=1}^{|I_{k'}|} \lceil \frac{m_{k',i}-1}{S}\rceil$ to identify $I_{k'}$. 

Applying this result for $k'=1,...,k$, we obtain that our strategies for solving the causal mean matching problem require at most $\lceil \log_2(r+1)\rceil$ times more interventions, compared to the optimal strategy, in the worse case over all feasible orientations.
\end{proof}

\section{Numerical Experiments}\label{appendix:append-exp}
\subsection{Experimental Setup}
\textbf{Graph Generation:} We consider two random graph models: Erdös-Rényi graphs \citep{erdHos1959renyi} and Barabási–Albert graphs \citep{albert2002statistical}. The probability of edge creation in Erdös-Rényi graphs is set to $0.2$; the number of edges to attach from a new node to existing nodes in Barabási–Albert graphs  is set to $2$. We then tested on two types of structured chordal graphs: rooted tree with root randomly sampled from all the nodes in this tree, and moralized Erdös-Rényi graphs \citep{shanmugam2015learning} with the probability of edge creation set to $0.2$.

\textbf{Multiple Runs:} For each instance in the settings of Barabási–Albert graphs with 100 nodes and $S=1$ in Figure \ref{fig5:a}, we ran the three non-deterministic strategies (\texttt{UpstreamRand},\texttt{CliqueTree}, \texttt{Supermodular}) for five times and observed little differences across all instances. Therefore, we excluded the error bars when plotting the results as they are visually negligible and the strategies are robust in these settings.

\textbf{Implementation:} We implemented our algorithms using the NetworkX package \citepappendix{hagberg2008exploring} and the CausalDAG package \url{https://github.com/uhlerlab/causaldag}. All code is written in Python and run on AMD 2990wx CPU. 

\subsection{More Empirical Results}

In the following, we present additional empirical result. The evaluations are the same as in  \rref{sec:experiments}. The following figures show that we observe similar behaviors as in Figure \ref{fig5} across different settings.

\textbf{Random graphs of size $\{10, 50, 100\}$:} Barabási–Albert and Erdös-Rényi graphs with number of nodes in $\{10, 50, 100\}$.

\begin{figure}[ht]
     \centering
     \begin{subfigure}[b]{0.3\textwidth}
         \centering
         \includegraphics[width=\textwidth]{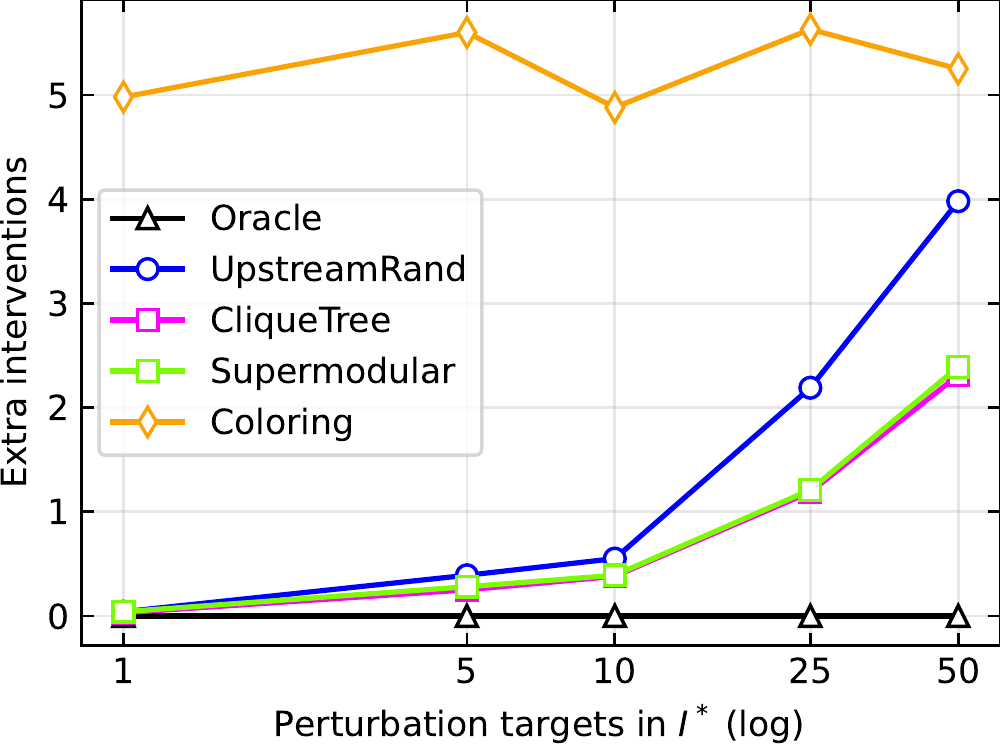}
         \caption{}
     \end{subfigure}
     \hfill
     \begin{subfigure}[b]{0.3\textwidth}
         \centering
         \includegraphics[width=\textwidth]{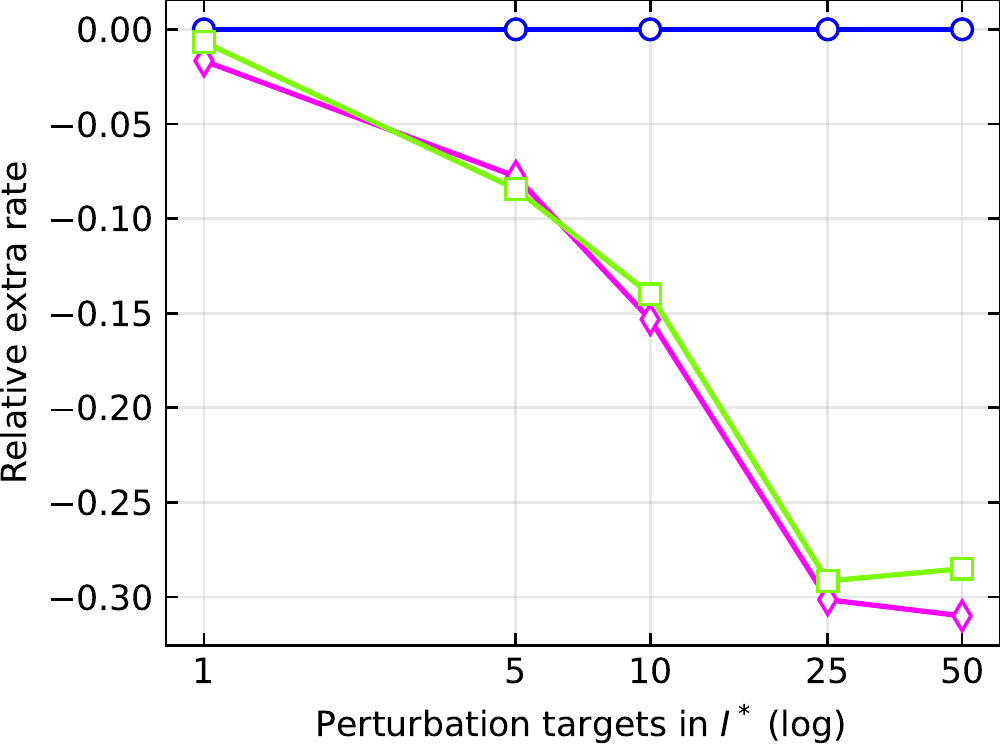}
         \caption{}
     \end{subfigure}
    \hfill
     \begin{subfigure}[b]{0.3\textwidth}
         \centering
         \includegraphics[width=\textwidth]{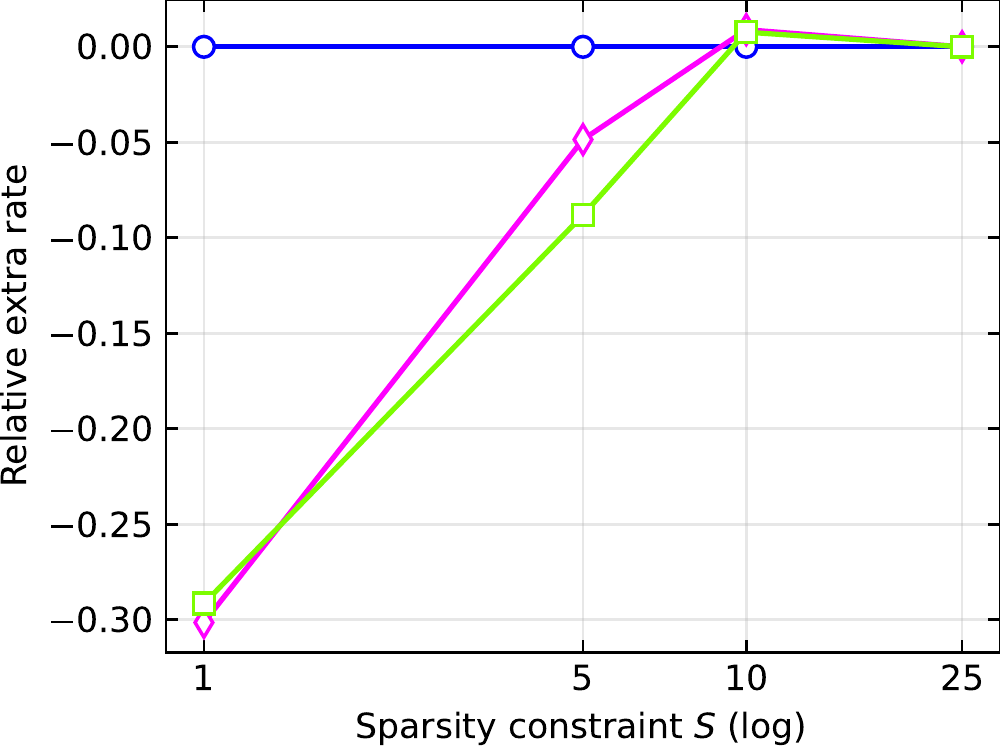}
         \caption{}
     \end{subfigure}
        \caption{Barabási–Albert graphs with $50$ nodes. \textbf{(a).} and \textbf{(b).} $S=1$; \textbf{(c).} $|I^*|=25$.
        }
    \vspace{0.1in}
\end{figure}

\begin{figure}[h!]
     \centering
     \begin{subfigure}[b]{0.329\textwidth}
         \centering
         \includegraphics[width=\textwidth]{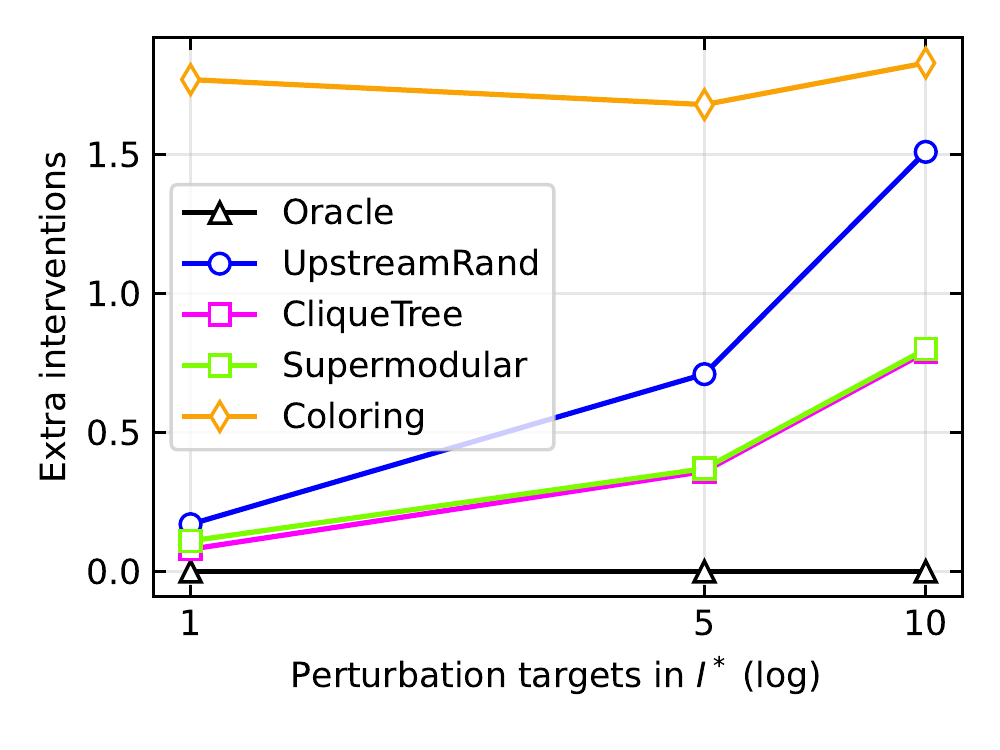}
         \caption{}
     \end{subfigure}
     \hfill
     \begin{subfigure}[b]{0.33\textwidth}
         \centering
         \includegraphics[width=\textwidth]{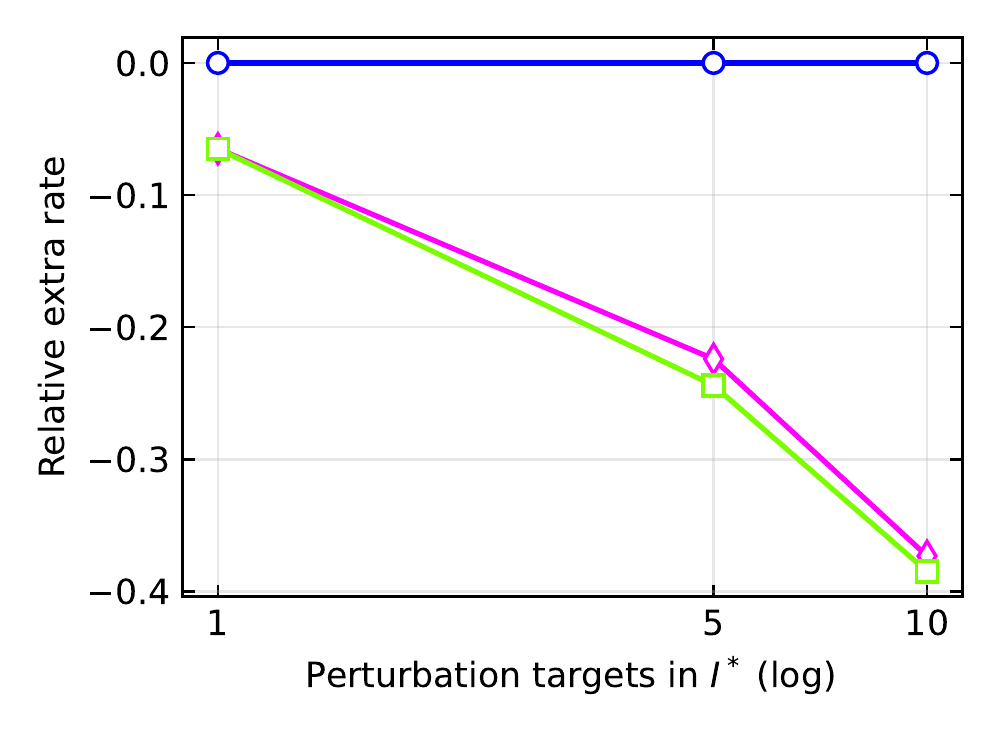}
         \caption{}
     \end{subfigure}
    \hfill
     \begin{subfigure}[b]{0.329\textwidth}
         \centering
         \includegraphics[width=\textwidth]{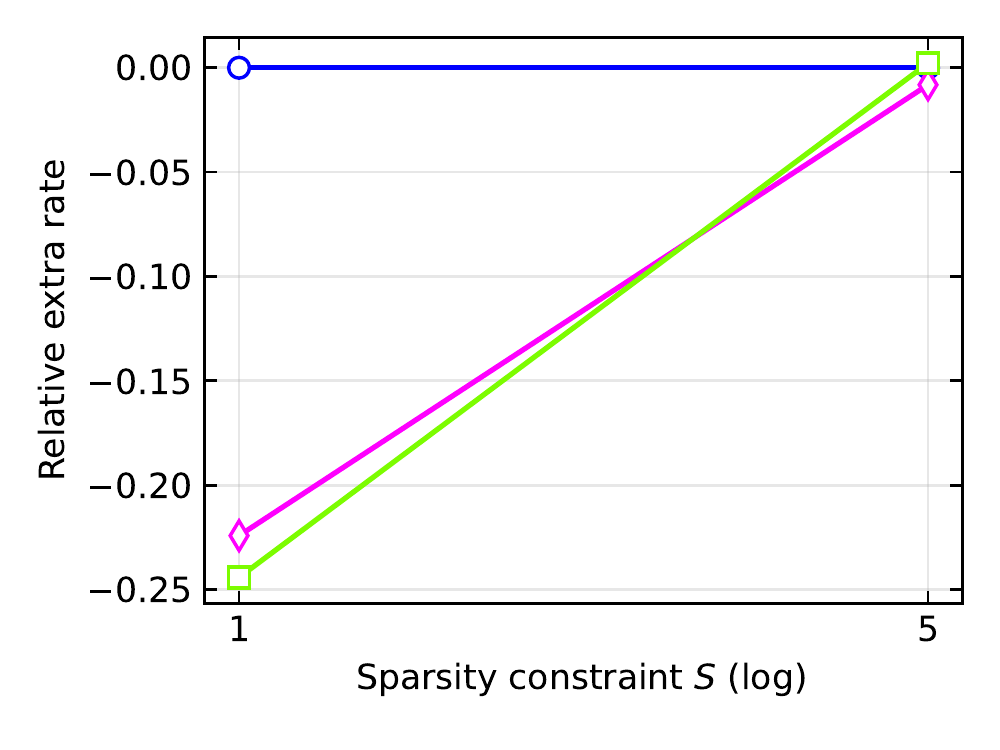}
         \caption{}
     \end{subfigure}
        \caption{Barabási–Albert graphs with $10$ nodes. \textbf{(a).} and \textbf{(b).} $S=1$; \textbf{(c).} $|I^*|=5$.
        }
    \vspace{0.1in}
\end{figure}

\begin{figure}[h!]
     \centering
     \begin{subfigure}[b]{0.3\textwidth}
         \centering
        \includegraphics[width=\textwidth]{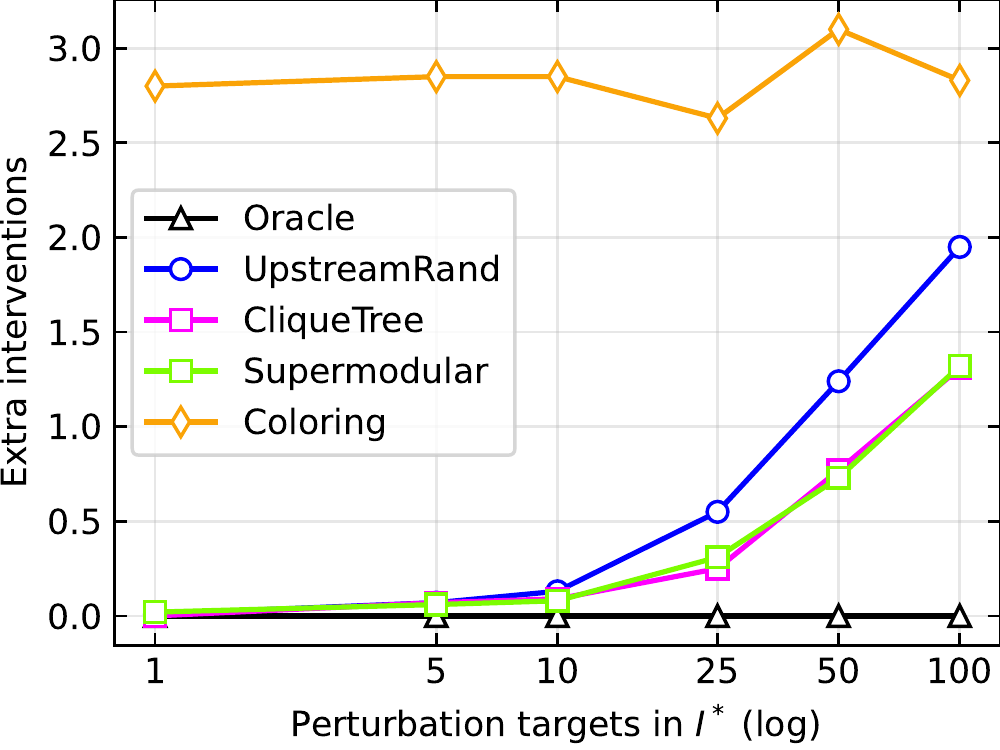}
         \caption{}
     \end{subfigure}
     \hfill
     \begin{subfigure}[b]{0.3\textwidth}
         \centering \includegraphics[width=\textwidth]{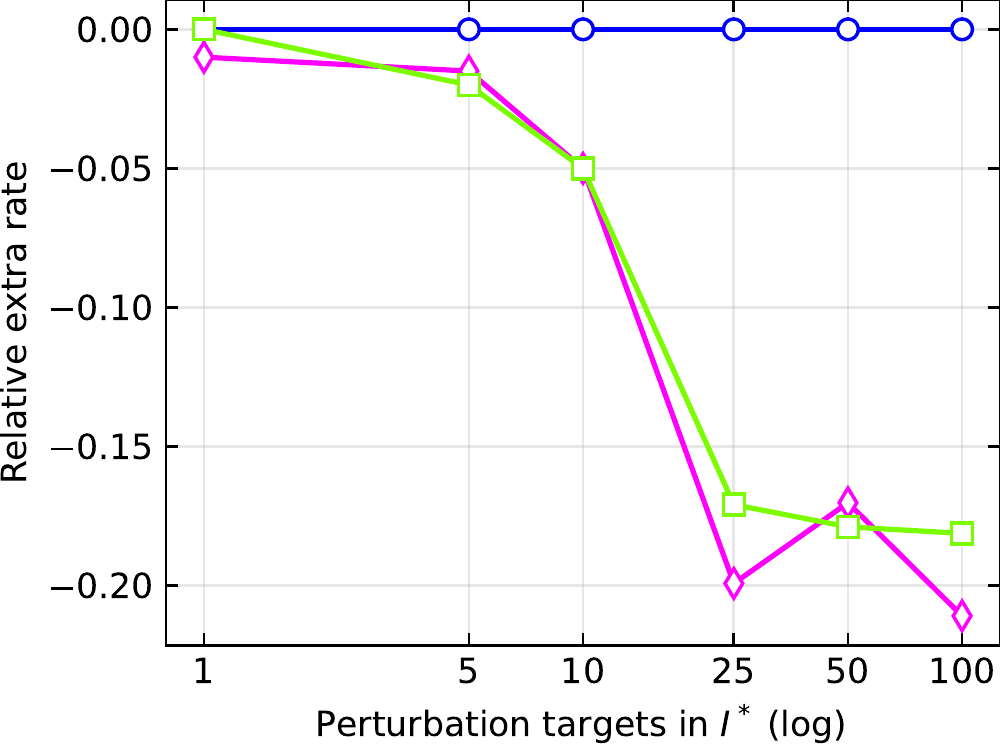}
         \caption{}
     \end{subfigure}
    \hfill
     \begin{subfigure}[b]{0.3\textwidth}
         \centering \includegraphics[width=\textwidth]{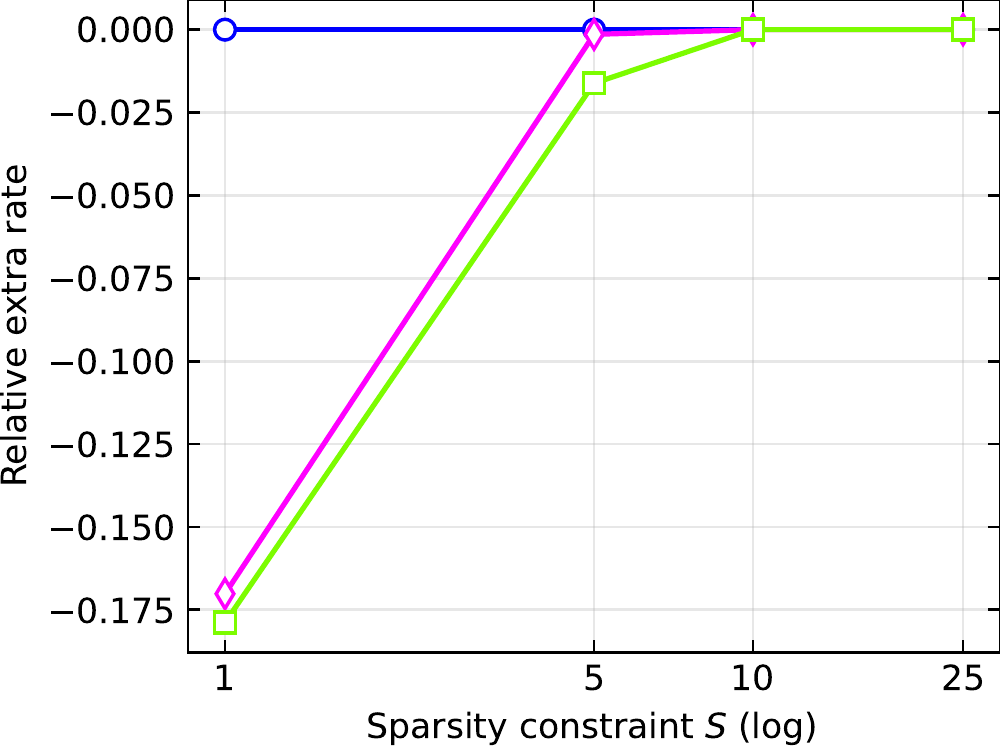}
         \caption{}
     \end{subfigure}
        \caption{Erdös-Rényi graphs with $100$ nodes. \textbf{(a).} and \textbf{(b).} $S=1$; \textbf{(c).} $|I^*|=50$.
        }
\end{figure}

\newpage

\begin{figure}[ht]
     \centering
     \begin{subfigure}[b]{0.28\textwidth}
         \centering
         \includegraphics[width=\textwidth]{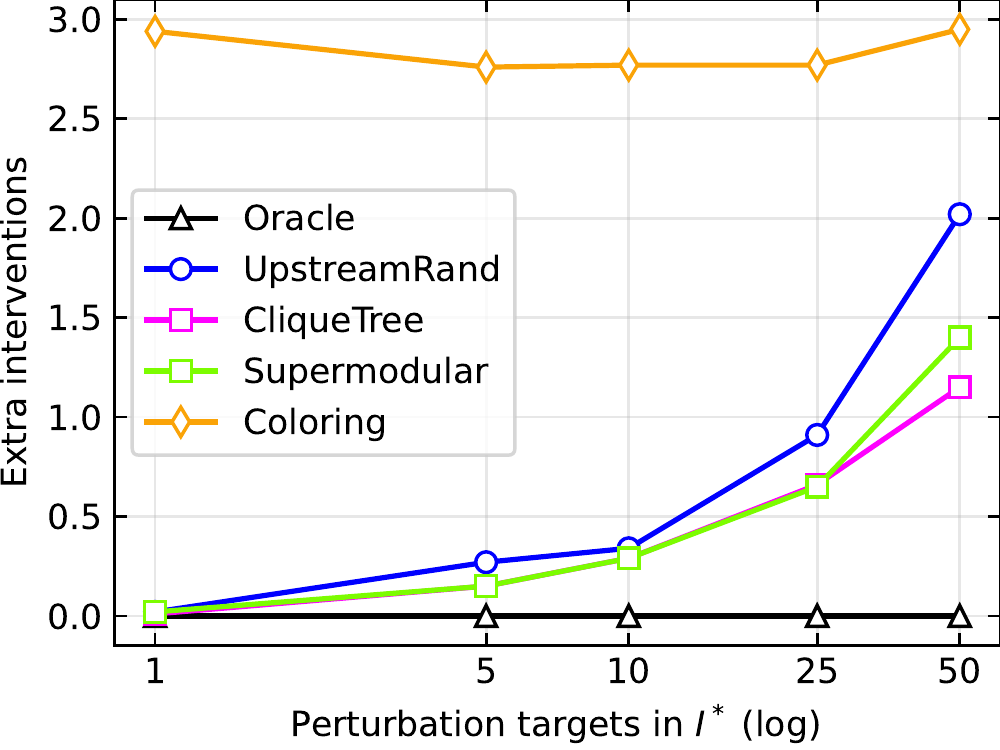}
         \caption{}
     \end{subfigure}
     \hfill
     \begin{subfigure}[b]{0.28\textwidth}
         \centering
         \includegraphics[width=\textwidth]{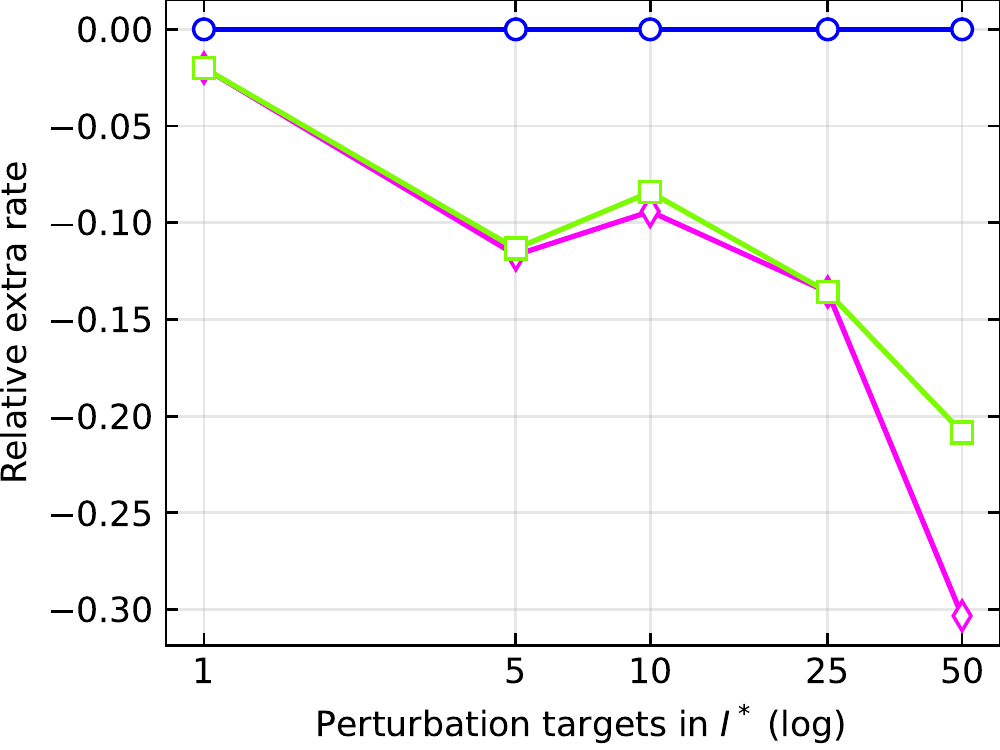}
         \caption{}
     \end{subfigure}
    \hfill
     \begin{subfigure}[b]{0.29\textwidth}
         \centering
         \includegraphics[width=\textwidth]{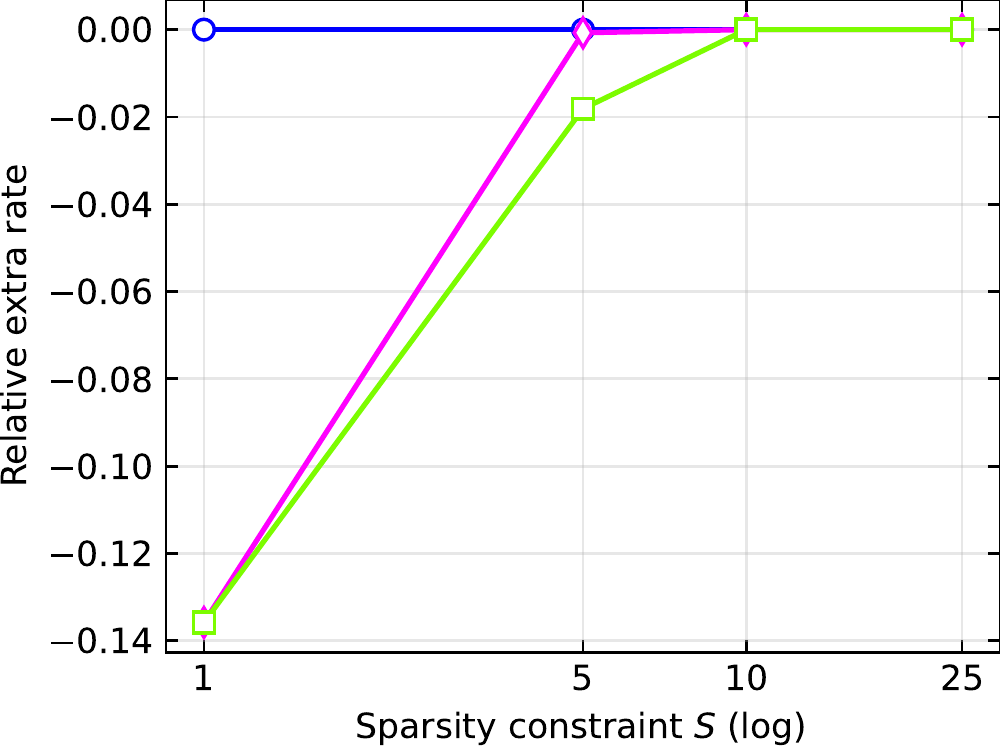}
         \caption{}
     \end{subfigure}
        \caption{Erdös-Rényi graphs with $50$ nodes. \textbf{(a).} and \textbf{(b).} $S=1$; \textbf{(c).} $|I^*|=25$.
        }
\end{figure}

\begin{figure}[h!]
     \centering
     \begin{subfigure}[b]{0.3\textwidth}
         \centering
         \includegraphics[width=\textwidth]{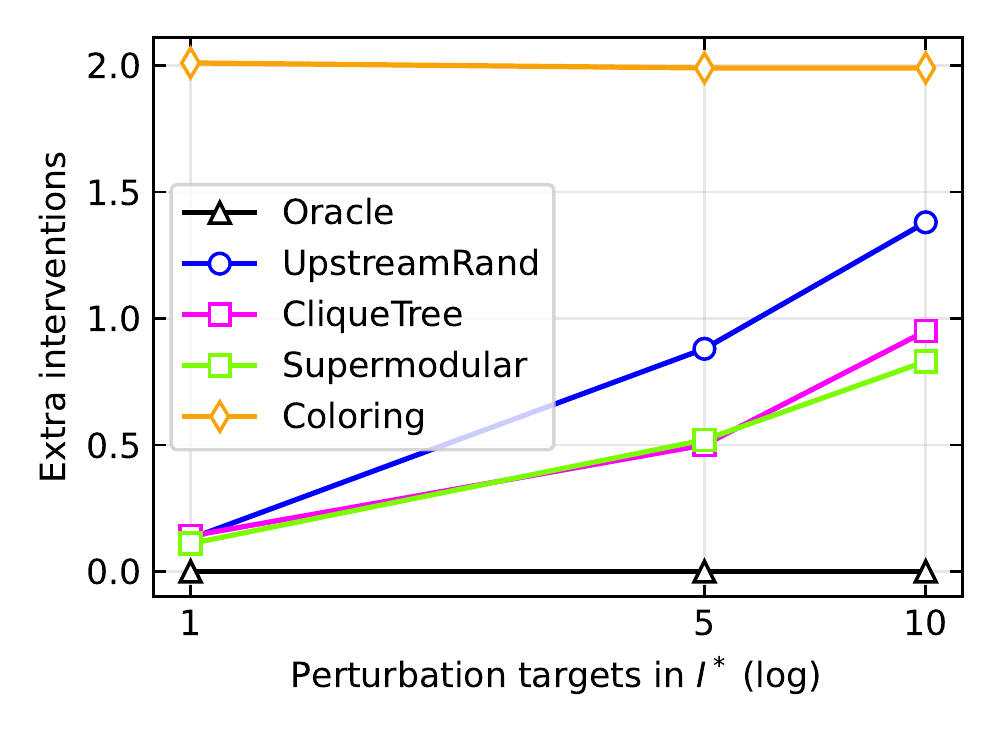}
         \caption{}
     \end{subfigure}
     \hfill
     \begin{subfigure}[b]{0.3\textwidth}
         \centering
         \includegraphics[width=\textwidth]{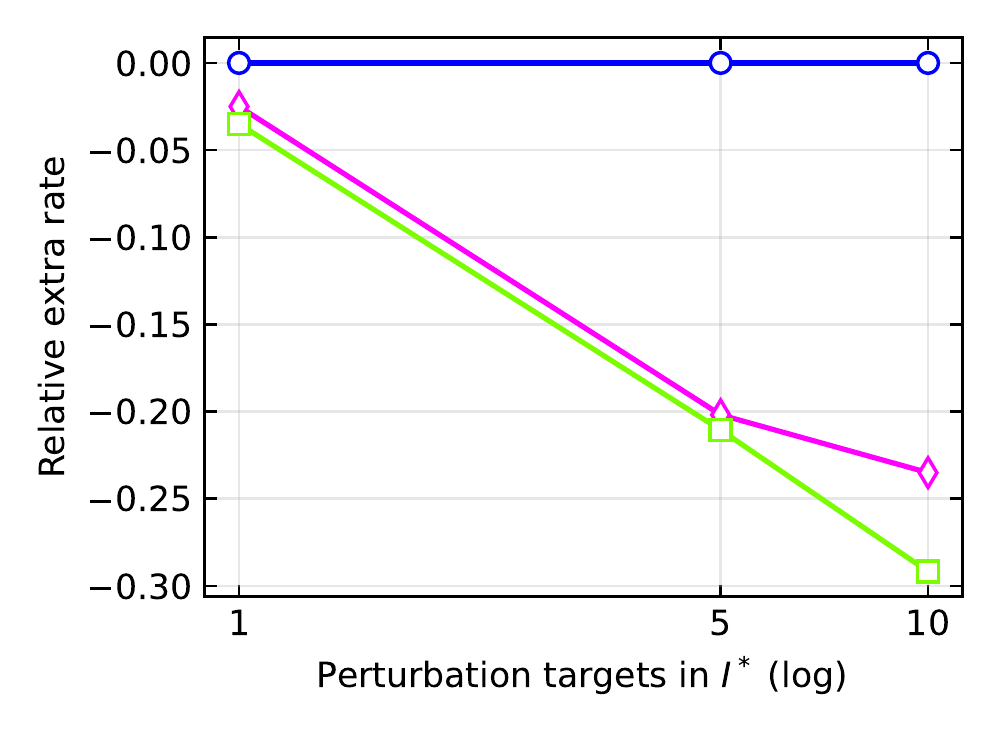}
         \caption{}
     \end{subfigure}
    \hfill
     \begin{subfigure}[b]{0.3\textwidth}
         \centering
         \includegraphics[width=\textwidth]{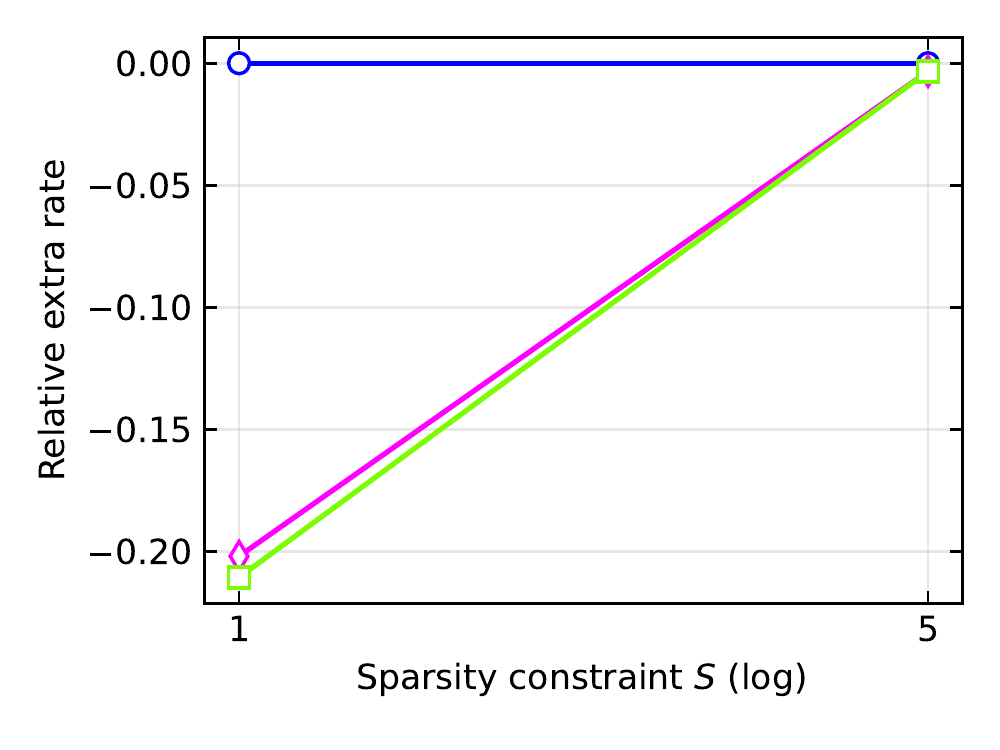}
         \caption{}
     \end{subfigure}
        \caption{Erdös-Rényi graphs with $10$ nodes. \textbf{(a).} and \textbf{(b).} $S=1$; \textbf{(c).} $|I^*|=5$.
        }
        \vspace{0.1in}
\end{figure}

\textbf{Larger Barabási–Albert graphs of size $1000$:} 

\begin{figure}[ht]
     \centering
     \begin{subfigure}[b]{0.3\textwidth}
         \centering
         \includegraphics[width=\textwidth]{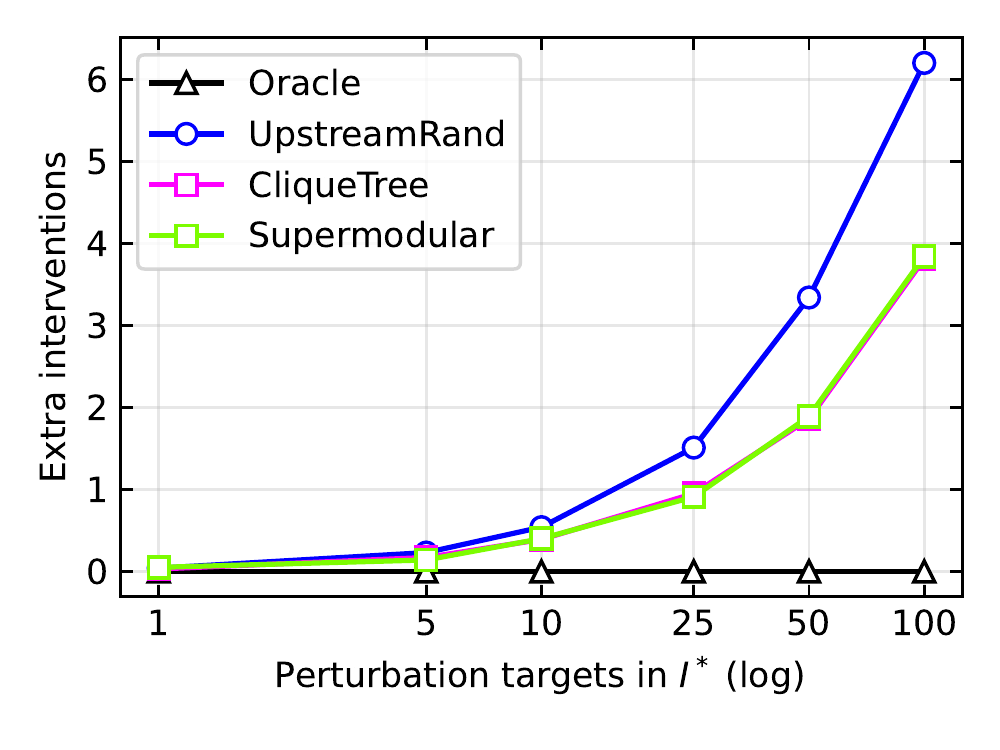}
         \caption{}
     \end{subfigure}
     \hfill
     \begin{subfigure}[b]{0.3\textwidth}
         \centering
         \includegraphics[width=\textwidth]{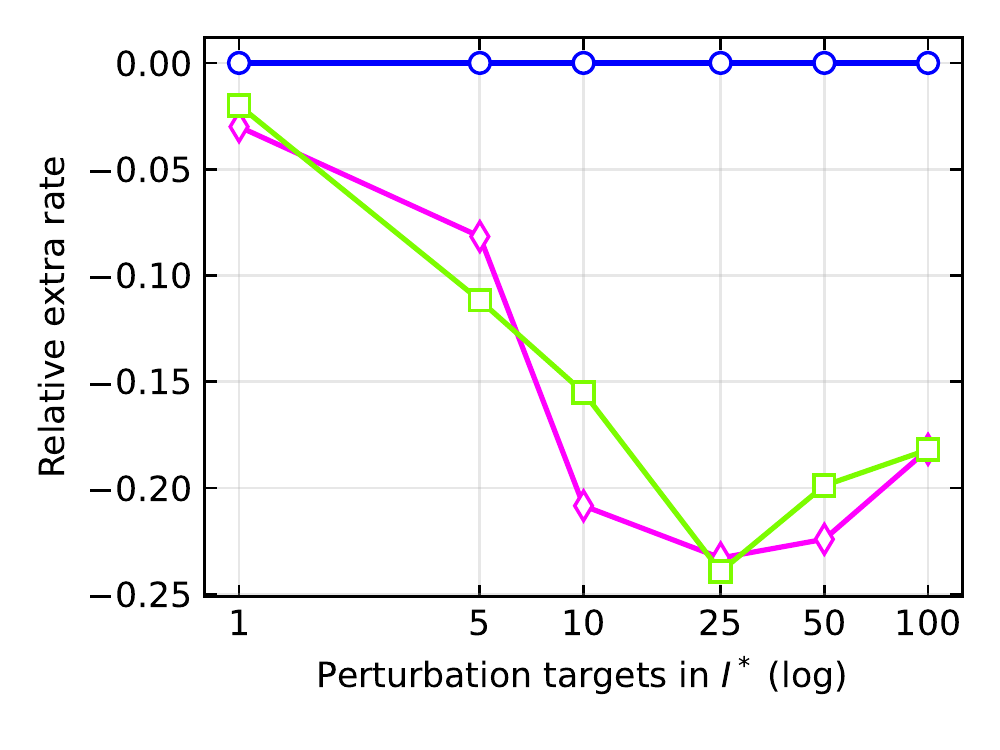}
         \caption{}
     \end{subfigure}
    \hfill
     \begin{subfigure}[b]{0.3\textwidth}
         \centering
         \includegraphics[width=\textwidth]{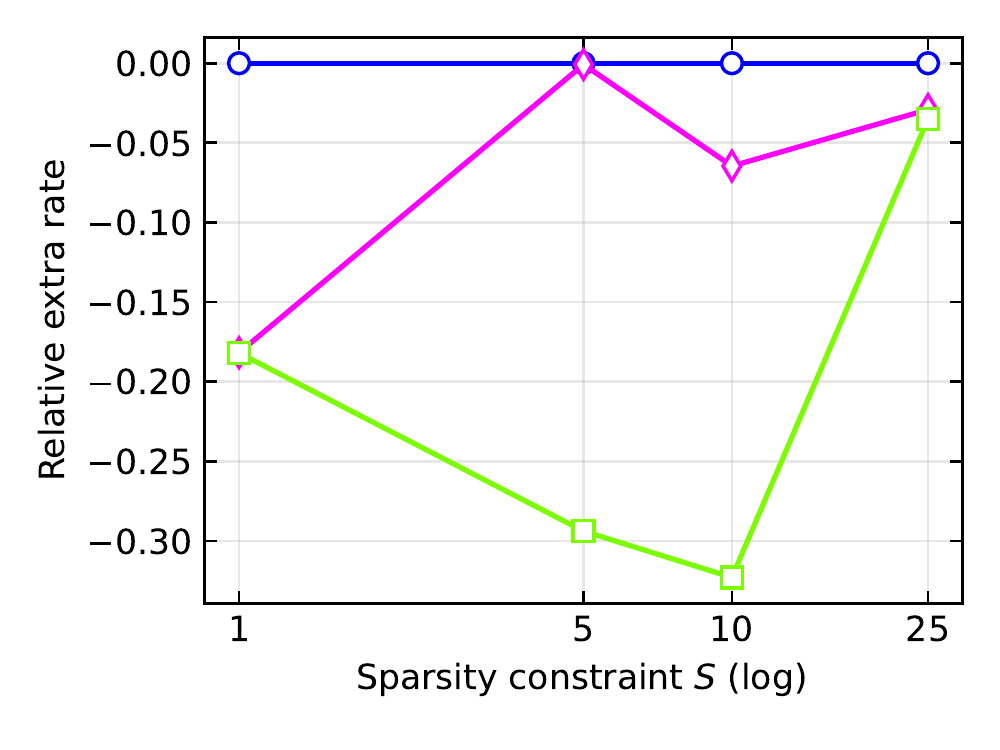}
         \caption{}
     \end{subfigure}
        \caption{Larger Barabási–Albert graphs with $1000$ nodes (excluding \texttt{coloring} which takes more than $80$ extra interventions). \textbf{(a).} and \textbf{(b).} $S=1$; \textbf{(c).} $|I^*|=100$.
        }
        \vspace{0.1in}
\end{figure}

\textbf{Two types of structured chordal graphs:} 

\begin{figure}[h!]
     \centering
     \begin{subfigure}[b]{0.24\textwidth}
         \centering
         \includegraphics[width=\textwidth]{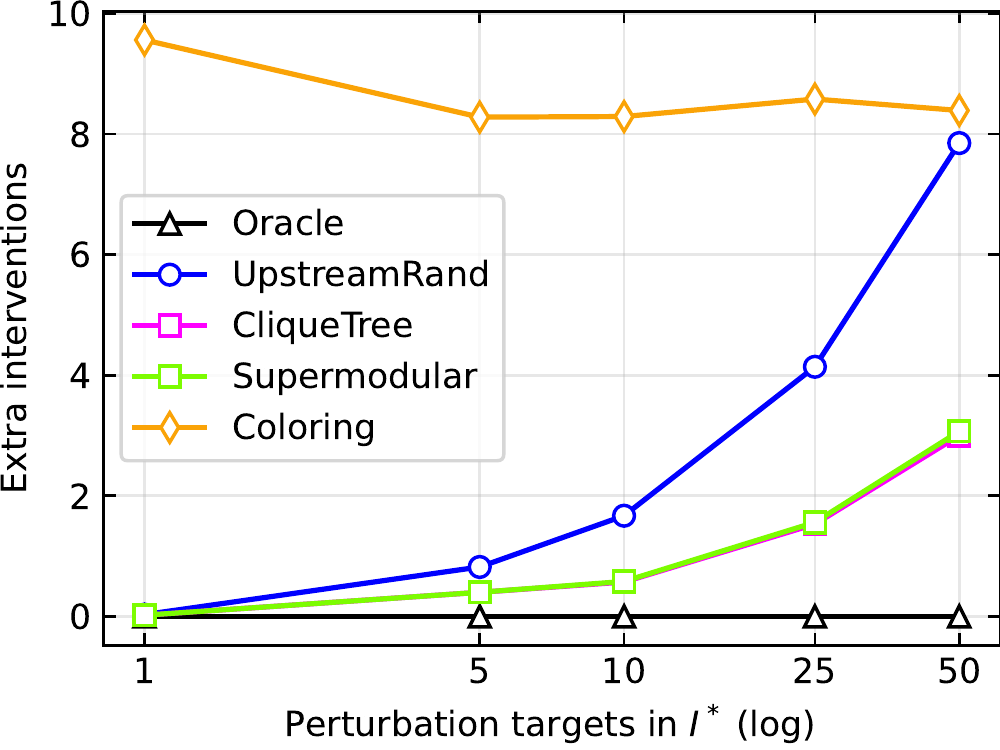}
         \caption{}
     \end{subfigure}
     \begin{subfigure}[b]{0.24\textwidth}
         \centering
         \includegraphics[width=\textwidth]{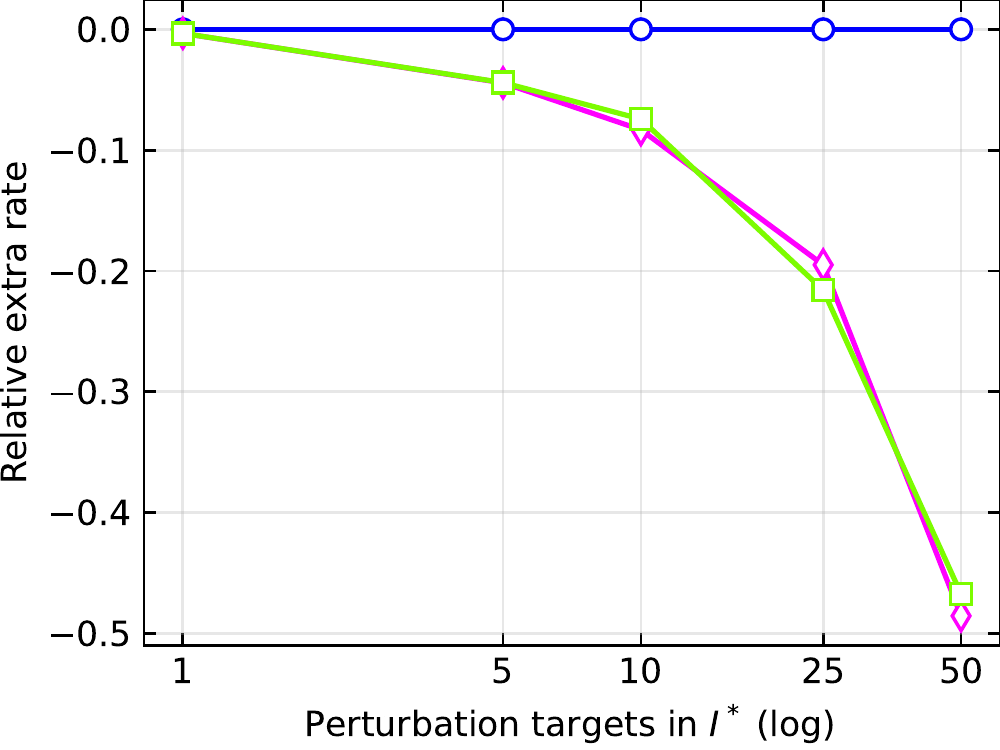}
         \caption{}
     \end{subfigure}
     \begin{subfigure}[b]{0.24\textwidth}
         \centering
         \includegraphics[width=\textwidth]{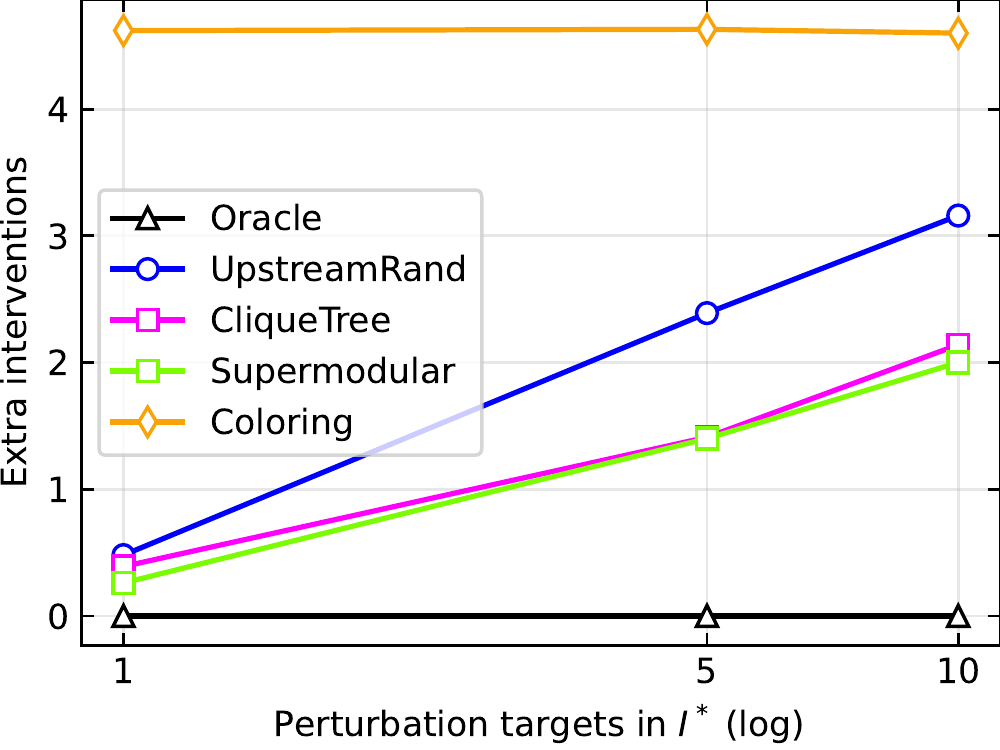}
         \caption{}
     \end{subfigure}
    \begin{subfigure}[b]{0.24\textwidth}
         \centering
         \includegraphics[width=\textwidth]{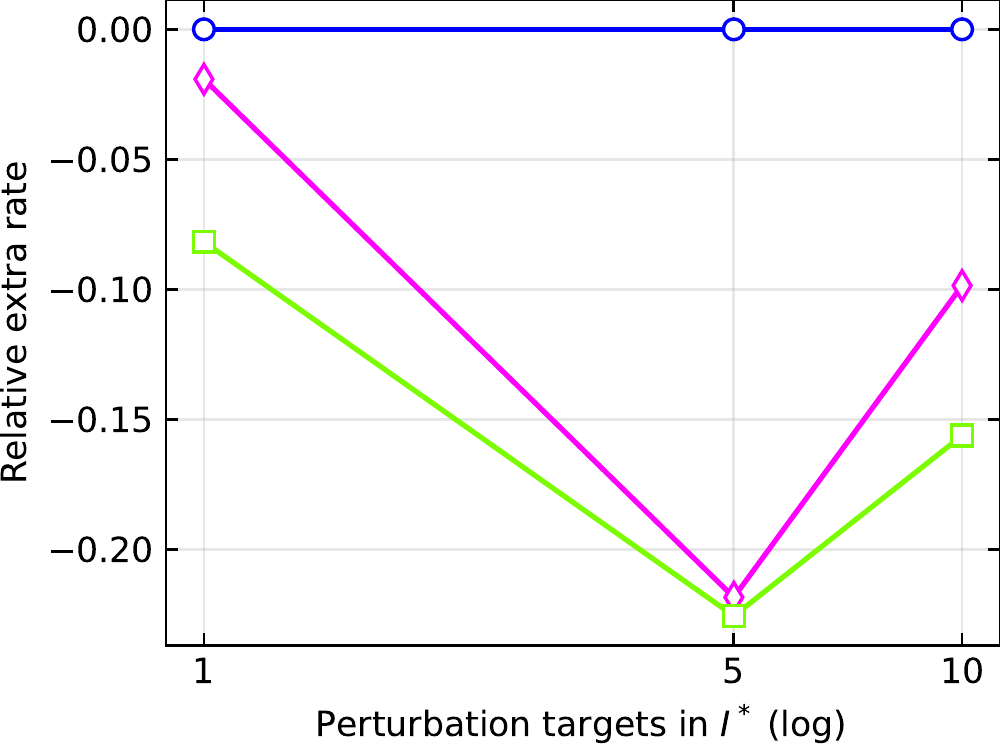}
         \caption{}
     \end{subfigure}
        \caption{Structured chordal graphs. \textbf{(a).} and \textbf{(b).} rooted tree graphs with $50$ nodes and $S=1$; \textbf{(c).} and \textbf{(d).} moralized Erdös-Rényi graphs with $10$ nodes and $S=1$.
        }
\end{figure}

\section{Discussion of the Noisy Setting}\label{appendix:noisy}
In the noisy setting, an intervention can be repeated many times to obtain an estimated essential graph. Each intervention results in a posterior update of the true DAG $\cG$ over all DAGs in the observational Markov equivalence class. For a tree graph $\cG$, this corresponds to a probability over all possible roots. To be able to learn the edges, \citetappendix{greenewald2019sample+} proposed a bounded edge strength condition on the noise for binary variables. Under this condition, they showed that the root node of a tree graph can be learned in finite steps in expectation with high probability. 

In our setting, to ensure that the source node w.r.t. an intervention can be learned, we need to repeat this intervention for enough times such that the expectation of each variable $X_i$ can be estimated. Furthermore, to ensure that the edges in the (general) interventional essential graph can be learned, we need a similar condition as in \citepappendix{greenewald2019sample+} for general chordal graphs and continuous variables.

\bibliographystyleappendix{apalike}
\bibliographyappendix{main}
\end{document}